\documentclass[11pt]{article}
\usepackage{amsmath,amssymb,amsfonts,amsthm,bbm}
\usepackage{verbatim}
\usepackage{setspace}
\usepackage{graphicx}
\usepackage{subcaption}
\usepackage{comment}
\usepackage{enumitem}
\usepackage{algorithm}
\usepackage{algpseudocode}
\usepackage[
pagebackref,
colorlinks=true,
pdfpagemode=UseNone,
citecolor=OliveGreen,
linkcolor=BrickRed,
urlcolor=BrickRed,
pdfstartview=FitH,
linktocpage=true]{hyperref}
\usepackage[dvipsnames]{xcolor}
\usepackage{bbm}
\usepackage[capitalise]{cleveref}

\newif\ifStat
\newif\ifNote
\newif\ifComments
\newif\ifThmitalic

\Stattrue
\Notefalse
\Commentsfalse
\Thmitalictrue

\ifComments
		\newcommand{\Tnote}[1]{{\color{blue} [Tselil: #1]}}
		\newcommand{\kz}[1]{\noindent{{\color{red}\textbf{\#\#\# Kangjie:} \textsf{#1} \#\#\#}}}
		\newcommand{\RMK}[1]{{\color{blue}{[Comment: #1]}}}
		\newcommand{\TODO}[1]{{\color{red}{[#1]}}}
		\newcommand{\rmk}[1]{{\color{blue}{[#1]}}}
\else
		\newcommand{\kz}[1]{}
		\newcommand{\Tnote}[1]{}
		\newcommand{\RMK}[1]{}
		\newcommand{\TODO}[1]{}
		\newcommand{\rmk}[1]{}
\fi

\ifStat
\renewcommand{\hat}{\widehat}
\renewcommand{\tilde}{\widetilde}
\renewcommand{\bar}{\overline}
\renewcommand{\top}{{\sf T}}
\else

\fi
\newcommand{\iidsim}{\overset{\scriptsize\iid}{\sim}}

\usepackage[top=1in, bottom=1in, left=1in, right=1in]{geometry}

\makeatletter
\newcommand*{\rom}[1]{\expandafter\@slowromancap\romannumeral #1@}
\makeatother

\newcommand{\<}{\langle}
\renewcommand{\>}{\rangle}

\newcommand{\cov}{\mathrm{Cov}}
\newcommand{\var}{\mathrm{Var}}

\newcommand{\op}{\mathrm{op}}

\newcommand{\iid}{\mathrm{i.i.d.}}
\newcommand{\beq}{\begin{equation}}
\newcommand{\eeq}{\end{equation}}

\newcommand{\poly}{\mathrm{poly}}

\DeclareMathOperator*{\argmax}{\mbox{argmax}}

\newcommand{\norm}[1]{\left\|{#1}\right\|}

\newcommand{\R}{\mathbb{R}}

\newcommand{\E}{\mathbb{E}}

\newcommand{\veps}{\varepsilon}

\renewcommand{\P}{\mathbb{P}}
\renewcommand{\S}{\mathbb{S}}

\renewcommand{\d}{\textup{d}}

\newcommand{\bone}{\mathrm{\bf 1}}

\newcommand{\cI}{\mathcal{I}}
\newcommand{\cC}{\mathcal{C}}
\newcommand{\cD}{\mathcal{D}}

\newcommand{\cS}{\mathcal{S}}

\newcommand{\cV}{\mathcal{V}}

\newcommand{\cA}{\mathcal{A}}

\newcommand{\cE}{\mathcal{E}}
\newcommand{\cF}{\mathcal{F}}

\def\normal{{\sf N}}

\newcommand{\sN}{\mathsf{N}}

\DeclareMathOperator*{\pliminf}{p-lim\, inf}

\DeclareMathOperator*{\plim}{p-lim}

\newcommand{\eps}{\varepsilon}

\newcommand{\sym}{\mathrm{sym}}

\ifNote

\theoremstyle{plain}
\newtheorem{thm}{Theorem}
\newtheorem{claim}{Claim}
\newtheorem{defn}{Definition}
\newtheorem{lem}[thm]{Lemma}
\newtheorem*{lemma*}{Lemma}
\newtheorem{cor}[thm]{Corollary}
\newtheorem{prop}[thm]{Proposition}
\newtheorem{ass}{Assumption}
\newtheorem{fact}{Fact}
\theoremstyle{definition}
\newtheorem{exm}{Example}
\newtheorem{rem}{Remark}
\newtheorem{conj}{Conjecture}

\else\ifThmitalic

\newtheorem{thm}{Theorem}[section]

\newtheorem{defn}{Definition}[section]
\newtheorem{lem}[thm]{Lemma}
\newtheorem*{lemma*}{Lemma}

\newtheorem{prop}[thm]{Proposition}

\theoremstyle{definition}

\newtheorem{rem}{Remark}[section]

\else

\fi
\fi

\newcommand{\asat}{\alpha_{\star}}
\newcommand{\aalg}{\alpha_{\mathrm{alg}}}
\newcommand{\aup}{\alpha_{\mathrm{up}}}
\newcommand{\alow}{\alpha_{\mathrm{low}}}
\DeclareMathOperator{\Ent}{H}

\usepackage{mathrsfs}
\usepackage{amsmath}
\usepackage{todonotes}
\usepackage{perceptron}

\def\anon{0}
\def\authornotes{0}

\ifnum\authornotes=1
    \newcommand\KZ[1]{\textcolor{red}{[KZ: #1]}}
    \newcommand{\shuangping}[1]{\footnote{\color{purple}Shuangping: {#1}}}
    \newcommand{\tselil}[1]{\footnote{\color{ForestGreen}Tselil: {#1}}}

    \newcommand{\tnote}[1]{{\color{ForestGreen}[Tselil: #1]}}
    \newcommand{\snote}[1]{{\color{Purple}[Shuangping: #1]}}

\else
    \newcommand\KZ[1]{}
    \newcommand{\shuangping}[1]{}
    \newcommand{\tselil}[1]{}

    \newcommand{\tnote}[1]{}
    \newcommand{\snote}[1]{}

\fi

\ifnum\anon=0
\author{
Shuangping Li\thanks{Stanford University. \texttt{fifalsp@stanford.edu}.} 
\and Tselil Schramm\thanks{Stanford University.  \texttt{tselil@stanford.edu}. Supported by NSF CAREER award \# 2143246.}
\and Kangjie Zhou\thanks{Stanford University. \texttt{kangjie@stanford.edu}. Supported by the NSF through award DMS-2031883 and the Simons Foundation through Award 814639 for the Collaboration on the Theoretical Foundations of Deep Learning and by the NSF grant CCF-2006489.}
}
\else
\author{Author name(s) withheld for double-blind review}
\fi

\title{Discrepancy Algorithms for the Binary Perceptron}
\date{\today}

\begin{document}

\maketitle

\begin{abstract}
The binary perceptron problem asks us to find a sign vector in the intersection of independently chosen random halfspaces with intercept $-\kappa$.
We analyze the performance of the canonical discrepancy minimization algorithms of Lovett-Meka and Rothvoss/Eldan-Singh for the asymmetric binary perceptron problem.
We obtain new algorithmic results in the $\kappa = 0$ case and in the large-$|\kappa|$ case.
In the $\kappa\to-\infty$ case, we additionally characterize the storage capacity and complement our algorithmic results with an almost-matching overlap-gap lower bound.

\end{abstract}
\thispagestyle{empty}

{ \hypersetup{hidelinks} \tableofcontents }
\thispagestyle{empty}

\setcounter{page}{1}


\section{Introduction}
\subsection{The Binary Perceptron Model}
The \emph{Binary Perceptron} is a simple model of a two-layer neural network as a Boolean function.
Related models have been studied since the 1960's \cite{Joseph60,Wendel62,Cover65}, and systematic study of the Binary Perceptron began in statistical physics in the 1980's \cite{gardner1988space}.
Given a fixed margin parameter $\kappa \in \R$ and a data matrix $X \in \R^{M \times N}$ where $X_{ij} \iidsim \normal (0, 1)$, the associated Binary Perceptron is the function $f_X(\theta) = \Ind(X\theta \ge \kappa \sqrt{N})$, and the \emph{solution set} $S$ is the set of all $\theta \in \{ \pm 1 \}^N$ on which $f_X(\theta) = 1$:
\begin{equation}\label{eq:ising_perc}
	S = S_{M, N} (X) := \left\{ \theta \in \{ \pm 1 \}^N: \ \frac{1}{\sqrt{N}} \left\langle X_i, \theta \right\rangle \ge \kappa, \ \forall i \in [M] \right\},
\end{equation}
where $X_i^\top$ is the $i$-th row of $X$. 
Written this way, the Binary Perceptron can be thought of as a dense constraint satisfaction problem (CSP) with constraints defined by random half-spaces in $\R^N$, rather than by Boolean predicates with fixed-size support.
One also sometimes studies the associated \emph{Symmetric Binary Perceptron} model, in which $f^{\sym}_X(\theta) = \Ind(|X \theta| \le |\kappa|\sqrt{N})$ \cite{aubin2019storage}. 

\medskip

Just as with other random CSPs, the first natural question is whether the solution set $S$ tends to be non-empty at a given constraint density.
The \emph{storage capacity} of the Binary Perceptron is defined as the ratio $M_{\star}(N, \kappa)/N$, where $M_\star = M_\star(N, \kappa)$ is the largest $M$ such that $S_{M, N} (X)$ is non-empty.
It was recently shown that the Binary Perceptron exhibits a sharp threshold sequence \cite{xu2021sharp,NS23}: that is, there exists a sequence $\asat(N,\kappa)$ such that $\P \left( S_{\alpha N, N} (X) \neq \varnothing \right)$ transitions from $1-o_N(1)$ to $o_N(1)$ in an $o_N(1)$ window around $\alpha = \asat(N,\kappa)$. The value of $\asat(N,\kappa)$ is not yet known for general values of $\kappa$, though Krauth and M\`{e}zard \cite{KM89} conjectured that  $\asat(N,\kappa)$ concentrates around an explicit constant $\asat(\kappa)$, with $\asat(0) \approx 0.83$. This has been rigorously verified very recently in \cite{ding2019capacity,Huang24}, under an explicit numerical condition.

\medskip
The second natural question is, given that the solution set is nonempty, is there an algorithm that finds some $\theta \in S$ in time $\poly(M,N)$?
In statistical physics, the existence of efficient algorithms for finding solutions to CSPs is thought to be closely linked to the geometry of the solution set $S$ (when thought of as a subset of $\{\pm 1\}^N)$.
If $S$ is poorly connected, in that solutions tend to be far apart in the hypercube, then the intuition is that algorithms (and natural processes) should have a hard time finding, or traveling between, solutions in $S$.

The Binary Perceptron is particularly interesting in this context.
On the one hand, efficient algorithms are known at some densities and for some value of $\kappa$. 
For example, Kim and Roche \cite{KR98} give a polynomial-time algorithm for the case $\kappa=0$ at densities $\alpha \le 0.005$.\footnote{The algorithm consists of iteratively choosing the $t$-th coordinate $\theta_t$ as a weighted majority of the $\{X_{it}\}_{i=1}^M$, with weights varying over time $t$. The analysis is somewhat miraculous.} 
In simulation, belief propagation algorithms appear to succeed up to $\alpha \approx 0.74$ \cite{BZ06,BBBZ07}.
No efficient algorithms are known to work at densities closer to the storage capacity $\asat(0) \approx 0.83$.
Thus the model has an apparent information-computation gap (a regime of $\alpha$ for which solutions typically exist in $S$, but we do not know any efficient algorithms to find them), though it is possible this gap could be closed after further algorithmic study.

On the other hand, the Binary Perceptron is conjectured \cite{KM89,HWK13,HK14} (and, in the symmetric case, rigorously known  \cite{PX21,ALS22contiguity}) to exhibit what is called \emph{strong freezing} at \textbf{every} density $0<\alpha<\asat (\kappa)$: an overwhelming $1-o(1)$ fraction of $\theta\in S$ are totally isolated, with all $\theta' \in S \setminus \{\theta\}$ at distance $\Omega(N)$ from $\theta$.
This apparent contradiction between the physical intuition (which would suggest that efficient algorithms should not exist at any fixed density) and the algorithmic reality (algorithms are known to succeed at some densities, but not all) has inspired a more nuanced study of the computational landscape, on both the algorithms and the computational complexity side.
It has been shown that polynomial-time algorithms for the Symmetric Binary Perceptron find solutions which are in the $o(1)$-measure ``clustered'' portion of $S$ \cite{ALS22algs} for $\alpha$ small.
At the same time, for $\alpha$ close to $\asat(\kappa)$ (again in the symmetric case), $S$ is known to exhibit the multi-Overlap Gap Property (m-OGP) \cite{gamarnik2022algorithms}, a more nuanced form of clustering which is known to rule out sufficiently smooth algorithms (more on this below).
Still, the computational landscape is far from well-understood, and at present there is no setting of $\kappa$, in either the symmetric or asymmetric case, at which the density where algorithms are known to succeed is even within a constant factor of the OGP bound (or any other computational lower bound).

\subsection{Main Results}
Our main topic of study is the performance of computationally efficient algorithms for the problem of discrepancy minimization, slightly repurposed to solve the Asymmetric Binary Perceptron.
Though the algorithms we employ were discovered a decade ago, to our knowledge their application in the context of the Asymmetric Binary Perceptron is new (for the Symmetric Binary Perceptron, discrepancy minimization algorithms give guarantees as a black box, as observed in e.g. \cite{gamarnik2022algorithms}).
The simplicity of the algorithms allows us to obtain relatively sharp guarantees.
Along the way, we also obtain bounds on the storage capacity in the case $\kappa < 0$, and prove that at some densities the model has an overlap gap property, which is an obstruction to some kinds of algorithms.
\Cref{fig:results} summarizes our results in the context of the information-computation landscape.

\begin{figure}[ht!]
\begin{center}

\begin{tikzpicture}
    \draw[->] (0,0) -- (10,0) node[anchor=north west] {$\alpha$};
    \draw[-] (0,-0.1) -- (0,0.1);
    
    \draw[dashed] (5,0) -- (5,1.5);
    \node[anchor=north] at (5,0) {\small $\qquad\aup$};
    \node[anchor=south] at (7,.5) {UNSAT};
    \node[anchor=south] at (4.75,.5) {\textbf{?}};

    \draw[dashed,color=blue,ultra thick] (4.5,0) -- (4.5,1.5);
    \node[anchor=north,color=blue] at (4.5,0) {\small $\aalg\qquad$};
    \node[anchor=north] at (4.75,-0.3) {\small $\frac{2}{\pi}\frac{1}{\kappa^2}$};
    \draw[dashed] (4.75,-0.3) -- (4.75,0);
    \node[anchor=south] at (2.5,.5) {EASY};
    \draw[<->] (4.5,1.5) -- (5,1.5) node[anchor=south] {\small $o(\frac{1}{\kappa^2})$};
    
    \node[anchor=south west] at (-2.5,0.5) {$\kappa \to\infty $};
\end{tikzpicture}

\vspace{4mm}
\begin{tikzpicture}
    \draw[->] (0,0) -- (10,0) node[anchor=north west] {$\alpha$};
    \draw[-] (0,-0.1) -- (0,0.1);
    
    \draw[dashed,ultra thick,color=blue] (7.25,0) -- (7.25,1.5);
    \draw[dashed,ultra thick,color=blue] (7.75,0) -- (7.75,1.5);
    \node[anchor=north,color=blue] at (7.75,0) {\small $\qquad\aup$};
    \node[anchor=north,color=blue] at (7.25,0) {\small $\alow\qquad$};
    \node[anchor=south] at (9,.5) {UNSAT};
    \node[anchor=north] at (7.5,-0.3) {\small $\frac{\log 2}{\Phi(\kappa)}$};
    \draw[dashed] (7.5,-0.3) -- (7.5,0);
    \draw[<->] (7.25,1.5) -- (7.75,1.5) node[anchor=south] {\small $o(\frac{1}{\Phi(\kappa)})$};
    \node[anchor=south] at (7.5,.5) {\textbf{?}};
    
    \draw[dashed,color=blue,ultra thick] (3,0) -- (3,1.5);
    \node[anchor=north,color=blue] at (3,0) {\small $\begin{array}{c}\alpha_{\mathrm{OGP}}\quad\\ \qquad\sim \frac{\log^2|\kappa|}{\kappa^2\Phi(\kappa)}\end{array}$};
    \node[anchor=south] at (5,.4) {HARD {\small (OGP)}};
    \node[anchor=south] at (2.2,.5) {\textbf{???}};

    \draw[dashed,color=blue,ultra thick] (1.5,0) -- (1.5,1.5);
    \node[anchor=north,color=blue] at (1.5,0) {\small $\begin{array}{c}\alpha_{\mathrm{alg}}\\ \sim \frac{1}{\kappa^2\Phi(\kappa)}\end{array}$};
    \node[anchor=south] at (0.6,.5) {EASY};
    
    \node[anchor=south west] at (-2.5,0.5) {$\kappa \to-\infty $};
\end{tikzpicture}

\vspace{7mm}

\begin{tikzpicture}
    \draw[->] (0,0) -- (10,0) node[anchor=north west] {$\alpha$};
    \draw[-] (0,-0.1) -- (0,0.1);
    
    \draw[dashed] (8.3,0) -- (8.3,1.5);
    \node[anchor=north] at (8.3,0) {\small $\begin{array}{c}\alpha_{\star}\\\quad \approx0.83\end{array}$};
    \node[anchor=south] at (9.2,.5) {UNSAT};

    \draw[dashed] (7,0) -- (7,1.5);
    \node[anchor=north] at (7,0) {\small $\begin{array}{c}\alpha_{\mathrm{BP}} \\ \approx 0.74\end{array}$};
    \node[anchor=south] at (7.6,.5) {\textbf{???}};
    \node[anchor=south] at (4,.43) {EASY {\small(empirically)}};
    
    \draw[dashed,color=blue, ultra thick] (1.2,0) -- (1.2,1.5);
    \node[anchor=north,color=blue] at (1.2,0) {\small $\begin{array}{c}\alpha_{\mathrm{alg}} \\\quad \approx 0.1\end{array}$};
    \node[anchor=south] at (0.62,.62) {\footnotesize EASY};
    
    \draw[dashed] (0.1,0) -- (0.1,1.5);
    \node[anchor=north] at (0.05,0) {\small$\begin{array}{c} \alpha_{\mathrm{KR}} \\ =0.005\end{array}$};
    \node[anchor=south west] at (-2.5,0.5) {$\kappa = 0$};
    
\end{tikzpicture}
\end{center}
\caption{The information-computation landscape for the Binary Perceptron. 
Our results appear in blue.
When  $\kappa \to \infty$, the upper bound $\aup$ on the storage capacity is due to \cite{stojnic2013discrete}.
In the case $\kappa = 0$, $\alpha_{\mathrm{KR}}$ is the density at which the Kim-Roche algorithm provably succeeds \cite{KR98}; in simulation, belief-propagation algorithms succeed up to density $\alpha_{\mathrm{BP}}$ \cite{BBBZ07,BZ06}, and the satisfiability threshold $\asat$ was conjectured in \cite{KM89} and rigorously established in \cite{ding2019capacity,Huang24}. 
}
\label{fig:results}
\end{figure}

\paragraph{Our Algorithm.} 
At a high level, our algorithm builds on the idea of solving a relaxed version of the binary perceptron problem using a continuous optimization procedure---specifically, a linear program with a random objective. 
The LP solution often lands near a corner of the hypercube, revealing partial information about the final sign vector. 
We then iteratively round the remaining coordinates to $\pm 1$ using discrepancy-style methods, ensuring feasibility throughout.

While these techniques are standard in discrepancy minimization, applying them in this average-case, high-dimensional setting requires significantly more care. 
A key insight is that by characterizing the empirical distribution of the LP solution’s margins and the fraction of saturated coordinates, we can control how much “slack” we have after the first step, enabling a warm start for the rounding procedure.

More formally, our algorithm consists of two stages. 
The first stage solves a linear program and outputs a solution $\hat{\theta}$ that satisfies the margin constraints, $X \hat{\theta} \ge \kappa \sqrt{N} 1$, and lies in the convex relaxation of the hypercube, $\hat{\theta} \in [-1, 1]^N$. 
This stage is closely related to the LP-based discrepancy algorithms of \cite{rothvoss2017constructive,eldan2018efficient}. 
In the second stage, we use the Edge-Walk algorithm of Lovett and Meka \cite{lovett2015constructive} to sequentially round all the coordinates of $\hat{\theta}$ to $\pm 1$, producing a genuine solution to the binary perceptron problem. 
As shown in \Cref{fig:results}, we will discuss below the performance of our two-stage algorithm in three different regimes of $\kappa$.

\paragraph{Results for Large Positive $\kappa$.} We discuss our results here when $\kappa$ is large and positive.

\begin{thm}\label{thm:informal_positive}[Informal, see \cref{thm:alg_kappa_positive}]
When $\kappa \to \infty$ and $\alpha = \frac{2}{\pi \kappa^2} (1+o_\kappa(1)) $, our algorithm finds a solution for the binary perceptron model with margin $\kappa$ with high probability. 
\end{thm}
Comparing this to previous satisfiability results, it was shown in \cite{stojnic2013discrete} that for $\kappa \to \infty$,
\begin{equation*}
	\aup (\kappa) = \frac{2}{\pi} \cdot \E_{G \sim \normal(0, 1)} \left[ ( \kappa - G )_+^2 \right]^{-1} \sim \frac{2}{\pi \kappa^2}
\end{equation*}
is an upper bound on $\asat (\kappa)$. With this, our \cref{thm:alg_kappa_positive} immediately implies that for large positive $\kappa$, the computational-to-statistical gap is negligible (if it exists at all). Therefore, our algorithm succeeds in finding a $\kappa$-margin solution for nearly all $\alpha$ in the satisfiability region, achieving sharpness up to a small $o_\kappa(1)$ multiplicative window.



\paragraph{Results for Large Negative $\kappa$.}
The information-computation landscape of the binary perceptron model becomes more complex as $\kappa \to -\infty$. For this regime, our first main result characterizes the asymptotics of the satisfiability threshold $\asat (\kappa)$.
We use $\Phi$ to denote the Gaussian CDF.

\begin{thm}[Informal, see \Cref{sec:sat}]
    When $\kappa \to - \infty$, the satisfiability threshold for the binary perceptron model satisfies $\asat (\kappa) \sim \frac{\log 2}{\Phi(\kappa)}$.
\end{thm}

The performance of our algorithm in this setting is described by the following theorem:

\begin{thm}[Informal, see \cref{thm:kappaverynegative}]
When $\kappa \to -\infty$ and $\alpha = O \left( \frac{1}{\Phi(\kappa) \kappa^2} \right)$, our algorithm finds a solution for the binary perceptron model with high probability. 
\end{thm}

We notice that there is a large $\Omega (\kappa^2)$ multiplicative gap between the satisfiability threshold and our algorithm as $\kappa$ becomes large a negative. To further understand the algorithmic landscape of this problem, we establish hardness results against stable algorithms (defined in \Cref{sec:mogp}) via the multi-Overlap Gap Property (m-OGP):

\begin{thm}[Informal, see \Cref{thm:mogp} and \Cref{thm:hardness}]
    As $\kappa \to -\infty$, the binary perceptron model exhibits m-OGP when $\alpha = \Omega\left (\frac{\log^2 |\kappa|}{\Phi(\kappa) \kappa^2} \right)$. Consequently, stable algorithms are unable to find a solution when $\alpha = \Omega\left (\frac{\log^2 |\kappa|}{\Phi(\kappa) \kappa^2} \right)$.
\end{thm}

This result shows that our algorithmic threshold aligns with the OGP hardness result up to a polylogarithmic factor in $\kappa$. It remains unclear whether better algorithms exist or if our hardness result against stable algorithms could be further strengthened. This also highlights an interesting phenomenon: the information-computation gap appears significantly larger compared to the case when $\kappa$ is positive. This phenomenon is discussed in more detail below.

\paragraph{Results for $\kappa = 0$.}
Finally, we focus on the regime $\kappa = 0$:
\begin{thm}[Informal, see \cref{thm:fullcoloring0} and \cref{thm:fullcoloring02}]
When $\kappa =0$ and $\alpha \leq 0.1$, our algorithm finds a solution for the binary perceptron model with high probability. 
\end{thm}

In comparison, the previously best-known algorithm \cite{KR98} works only when $\alpha \leq 0.005$.

\paragraph{Some remarks on our results.}
Our main results on the satisfiability threshold $\asat(\kappa)$ complement those established in prior works \cite{ding2019capacity,Huang24}. 
While these works rigorously verify that $\asat(0)$ matches the threshold conjectured by Krauth and Mézard, their analyses rely on intricate arguments and explicit numerical assumptions. 
In contrast, our proof is fully unconditional and proceeds via a clean application of the first and second moment methods. 
This not only simplifies the analysis, but also enables us to precisely characterize the asymptotic behavior of $\asat(\kappa)$ as $|\kappa| \to \infty$. 
As we will see, this asymptotic characterization is crucial for understanding how the information-computation gap scales in the large-$|\kappa|$ regime.

Beyond the sharp thresholds themselves, the overall computational landscape that emerges from our results is more interesting. 
It is easy to see that $\asat(\kappa) \to 0$ as $\kappa \to \infty$: the volume of each constraint decreases with $\kappa$, so as $\kappa$ increase the constraints become harder to satisfy. 
On the other hand, our results suggest that as $\kappa \to \infty$, the problem becomes algorithmically easier: if we imagine $\aalg^\star(\kappa)$ as the maximum limiting density at which polynomial-time algorithms provably succeed\footnote{We use $\aalg^\star(\kappa)$ merely as a device for discussion, and we will not define it rigorously.}, then we show that $\aalg^\star(\kappa)/\asat(\kappa) \to 1$ as $\kappa \to \infty$, and conversely, our overlap-gap result suggests that $\aalg^\star(\kappa)/\asat(\kappa)\to 0$ as $\kappa \to -\infty$. 

It would be interesting to understand this phenomenon further.
Perhaps one may borrow intuition from the \emph{Spherical Perceptron} problem, where there is a clear qualitative difference between the cases when $\kappa > 0$ and $\kappa \le 0$.
In the spherical problem, we are asked to find $\theta$ on the unit sphere, $\|\theta\|=1$, satisfying $X \theta \ge \kappa$. 
In the $\kappa > 0$ case the problem is convex: given any solution $\theta' \in \{\theta \in \mathbb{R}^N \mid \|\theta\|\le 1,\, X\theta \ge \kappa\}$, the unit vector $\theta^* = \theta'/\|\theta'\|$ will continue to satisfy the half-space constraints.
When $\kappa \le 0$, this no longer holds: the landscape becomes more complex \cite{franz2017universality}, and prior work \cite{montanari2024tractability} suggests that an information-computation gap may emerge as $\kappa \to -\infty$.

In the binary setting there is no immediate convexity, even in the $\kappa > 0$ case.
But one can instead study a convex relaxation: the set of $\theta'$ in the interior of the hypercube, satisfying the constraints with a bit of extra margin $\eps > 0$, that is $\theta' \in \{\theta \mid \theta \in [-1,1]^N,\,X\theta \ge (\kappa + \eps) \sqrt{N}\}$.
Instead of choosing an arbitrary $\theta'$ and scaling it up, imagine instead choosing the $\theta'$ that correlates best with some random direction $v\in \R^N$ (this is a linear program and can be done in polynomial time). 
This selection mechanism implies that $\theta'$ will be at a corner of the polytope.
Of course, this $\theta'$ need not be a sign vector.
Heuristically, the fraction of Boolean polytope constraints is $\frac{2}{2+\alpha}$ (there are $2N$ binary hyperplane constraints and $M = \alpha N$ half-space constraints).
If $\alpha$ is sufficiently small, it is reasonable to expect that the corner will be comprised mostly of hypercube constraints, and a majority of the coordinates of $\theta'$ will be in $\{\pm 1\}$ (of course, the value of $\kappa$ impacts how tight each half-space constraint is, but for the sake of simplicity let's put this aside for now). 
Then it may not be too difficult to round $\theta'$ to a nearby sign vector which is still feasible, given that we afforded ourselves the $+\eps$ margin.

As $\kappa \to \infty$, $\asat(\kappa) \to 0$, so any $\alpha < \asat(\kappa)$ is naturally small; on the other hand, as $\kappa$ decreases and $\asat(\kappa)$ increases, we might expect that we will only succeed for $\alpha$ well below the satisfiability threshold.
We have to note that this is a serious oversimplification, since constraint count is only part of the picture: as $\kappa$ increases, each half-space constraint becomes more restrictive. 
Accounting for this, we can squint and reconcile the fact that when $\kappa \to -\infty$, our algorithm succeeds for $\alpha = O(\frac{1}{\kappa^2}\asat(\kappa)) = O(\frac{1}{|\kappa|}\exp(\kappa^2/2))$, which can be much larger than $1$ (though there are more half-space constraints, they are further away from the origin and less likely to be tight when projecting $v$ to a corner).

This intuition is based entirely on the existence of specific (clusters of) solutions. 
It does not offer an explanation for why the problem becomes harder in general as $\kappa \to -\infty$.
It would be interesting to understand, via the lens of solution geometry or otherwise, what changes qualitatively about the computational complexity of the problem as $\kappa$ is decreased.


\subsection{Related Work}
\subsubsection{Algorithms from Discrepancy Minimization}
The \emph{discrepancy minimization} problem asks, given a set system of $M$ subsets of $N$ elements encoded by an incidence matrix $X \in \mathbb{R}^{M \times N}$, to choose a coloring of the elements $\theta \in \{\pm 1\}^N$ so that the maximum discrepancy among all sets is minimized:
\newcommand{\disc}{\mathrm{disc}}
\[
\disc(X) = \min_{\theta \in \{\pm 1\}^N} \|X\theta\|_\infty.
\]
Discrepancy minimization typically takes in a worst-case $X$, and asks for a good upper bound on $\disc(X)$.
A classic result of Spencer shows that $\disc(X) \le C\sqrt{N \log \left(\frac{M}{N} + 2\right)}$ \cite{spencer1985six}.
Spencer's proof was non-constructive, giving no sense of how to efficiently find a $\theta \in \{\pm1\}^N$ achieving the promised discrepancy bound.
Bansal later showed that such $\theta$ could be found efficiently \cite{Bansal10}, though his proof was also not constructive as it relied on Spencer's existence result.
In the decade since, algorithmic discrepancy has become popular.
There is now a variety of beautiful algorithms for discrepancy minimization, both in this classic setting and with a variety of twists (for example, \cite{lovett2015constructive,rothvoss2017constructive,eldan2018efficient,BS20, ALS21,BJM23}).

The Symmetric Binary Perceptron problem can be re-interpreted as a weighted, average-case discrepancy feasibility problem.
One samples $X \sim \normal(0,1)^{M \times N}$, and instead of asking for an upper bound on $\disc(X)$, we adopt a target $\kappa > 0$ and ask whether at a given density $\alpha = M/N$, 
\[
\mathbb{P}_X\left[ \disc(X) \le \kappa \sqrt{N}\right] \ge \frac{1}{2} ?
\]
Existing discrepancy minimization algorithms work for the Symmetric Binary Perceptron out-of-the-box (at some small enough densities $\alpha >0$).
For example, it is observed in \cite{gamarnik2022algorithms} that the online discrepancy minimization algorithm of Bansal and Spencer \cite{BS20} finds a solution to the Symmetric Binary Perceptron at densities $\alpha = O(\kappa^2)$, in the limit as $\kappa\to0$.

We re-purpose two of these algorithms to obtain new algorithmic results for the Asymmetric Binary Perceptron.
Though the Asymmetric Binary Perceptron is superficially different than discrepancy minimization, in effect our goal is to find a point in the intersection of $\{\pm 1\}^N$ with $M$ random half-spaces.
Some of the classic discrepancy minimization algorithms are almost already phrased in this way.
The random projection algorithm independently proposed by Rothvoss and Eldan \& Singh \cite{rothvoss2017constructive,eldan2018efficient} chooses a random direction $g \sim \normal (0,I_N)$, and projects it to the intersection of the hypercube and the slab constraints $|\langle X_i, \cdot  \rangle|\le \kappa\sqrt{N}$; this may be trivially modified to an intersection with half-space constraints $\langle X_i,\cdot \rangle \ge \kappa \sqrt{N}$.
The edge-walk algorithm of Lovett \& Meka \cite{lovett2015constructive} is a discrete-time Brownian motion $\theta_t$, starting from a feasible $\theta_0$, freezing motion in any direction which becomes tight either for a slab constraint $|\langle X_i, \theta_t\rangle| \le \kappa \sqrt{N}$ or for a hypercube constraint $|\langle \theta_t,e_j\rangle| \ge 1-o(1)$.
Again, this may be trivially modified to use half-space constraints instead of slab constraints.

However, while the algorithms themselves are similar in spirit, our analysis must go far beyond what is typically done in the discrepancy literature. In classical settings, the goal is usually to show that some constant fraction of coordinates hit $\{\pm 1\}$, under the assumption that the continuous relaxation has large Gaussian width. Such arguments are often coarse and designed for worst-case inputs. In contrast, our setting requires analyzing the average-case behavior of these algorithms on Gaussian random matrices $X$, with the goal of tightly characterizing both the feasibility and the asymptotic performance of the output.

To this end, we analyze a composition\footnote{
One application of the random projection or edge-walk algorithm will produce a partial coloring, or a vector $\theta\in \R^N$ with at least half of its entries as signs. 
The algorithms must be iteratively composed with themselves to produce a hypercube vector.
For technical reasons that will become clear later, for us it is more convenient to apply the random projection algorithm first, and then apply the edge-walk algorithm.
} 
of the random projection algorithm with the edge-walk algorithm in a manner that is specialized to Gaussian $X$, obtaining asymptotically sharp bounds in the $\kappa >0,\kappa \to \infty$ case, as well as new results in the $\kappa = 0$ and $\kappa < 0, |\kappa| \to \infty$ case. Achieving this requires a much sharper analysis than in prior discrepancy work. Specifically, we characterize the exact asymptotic fraction of coordinates that are frozen to $\pm 1$, and more importantly, the precise empirical distribution of the slack vector $X \theta$, where $\theta$ is the LP solution produced by the first stage. This level of understanding is essential to ensure that the edge-walk rounding step can be iterated without violating constraints across $O(\log N)$ rounds.

While our focus is on the asymmetric binary perceptron, we note that a rich body of literature on discrepancy minimization has also explored related average-case discrepancy problems, such as the discrepancy of random matrices and the feasibility of random integer programs, albeit in somewhat different settings. For instance, \cite{franks2020discrepancy,altschuler2022discrepancy} studied the discrepancy of random rectangular matrices, while \cite{chandrasekaran2014integer} examined the integer feasibility of random polytopes. The latter established a randomized polynomial-time algorithm for finding an integer point in a polytope, drawing a connection between integer feasibility and linear discrepancy. In addition, \cite{venkat2023efficient} investigated efficient algorithms for certifying lower bounds on the discrepancy
of random matrices, and \cite{borst2023integrality} analyzed integrality gaps in random integer programs through the lens of discrepancy. Furthermore, \cite{turner2020balancing} focused on balancing Gaussian vectors in high dimensions.

\subsubsection{Computational Complexity}

In \cite{gamarnik2022algorithms}, the authors prove that the Symmetric Binary Perceptron displays a \emph{multi-overlap-gap} property (m-OGP).
In brief, the OGP is a more refined statement about the solution geometry of $S$, or more specifically, about how the solution geometry of $S(X)$ changes when $X\to X'$ is perturbed slightly. 
In the most basic type of OGP, one shows that if $X'$ is an $\eps$-perturbation of $X$, with high probability there do not exist any pairs of solutions $\theta$ to $X$ and $\theta'$ to $X'$ with correlation $\frac{1}{N}\iprod{\theta,\theta'}$ in some range $[a(\eps),b(\eps)]$.
Intuitively, this captures the notion that the solution set $S(X)$ is unstable, in that it is far from Lipschitz in $X$.
A looser instability criterion for $S(X)$ can be captured by more sophisticated OGP statements, such as multi-OGP and branching OGP, which concern the solutions to an ensemble of multiple correlated instances, $X_1,\ldots,X_m$.

The OGP is known to rule out all algorithms that are Lipschitz or otherwise stable (in some technical sense) to perturbations in the input (see, e.g., \cite{rahman2017local, wein2022optimal, gamarnik2022random, huang2022tight} and the references therein). Many algorithms are stable in this sense, as demonstrated in the broader literature on random computational problems, including approximate message passing \cite{gamarnik2021overlap, huang2022tight}, low-degree polynomials \cite{gamarnik2020low, bresler2022algorithmic}, gradient descent \cite{huang2022tight}, and (as proven in \cite{gamarnik2022algorithms}) the Kim-Roche algorithm.
On the other hand, other commonly used algorithms, such as semidefinite programs, are not stable (e.g. for the maximum eigenvector problem).

In the $\kappa < 0, |\kappa| \to \infty$ regime, we prove that the Asymmetric Binary Perceptron displays a type of multi-OGP, proving that our algorithm is sharp within a $\log^2 |\kappa|$ factor.
Our proof is similar to that of \cite{gamarnik2022algorithms}.

\subsubsection{Other Algorithms for the Perceptron Model} 

\paragraph{Kim-Roche Algorithm.}
In the $\kappa = 0$ case, \cite{KR98} proposes a multi-scale majority algorithm that succeeds when $\alpha \leq 0.005$. It is an online algorithm with a conceptually clear approach: in each step, or stage, examine a block of adjacent columns of $X$ that have not yet been observed and choose the corresponding components of $\theta$ to reduce the number of rows of $X$ with a small cumulative inner product with $\theta$. Ultimately, with high probability, this algorithm outputs a vector $\theta$ of full dimension $N$ such that no row of $X$ has a small inner product with $\theta$. Despite the simplicity of the idea, the analysis is rather involved, and the parameter choices are very specifically tuned to the problem. Indeed, in each step, the choice of each $\theta_i$ depends on the entries of the matrix columns, making the inner products of the rows of $X$ with $\theta$ weakly dependent. This adds complexity to the analysis.
A Kim-Roche-based algorithm was also analyzed in \cite{ALS22algs} for general $\kappa$, but without sharp bounds in $\kappa$, as the primary purpose was not algorithmic.

\paragraph{AMP-Type Algorithms for Spherical Perceptron.}
Recently, there has been another line of works \cite{alaoui2020algorithmic, montanari2024exceptional} on designing efficient iterative algorithms for the spherical perceptron problem. Their approach is based on the so-called \emph{incremental} AMP (IAMP) scheme \cite{montanari2021optimization}, which belongs to the broader family of approximate message passing (AMP) algorithms. Roughly speaking, each IAMP step is composed of two operations: applying a Lipschitz function entry-wisely to current iterate and multiplying it by $X$ or $X^\top$. Of course, the choice of these Lipschitz functions governs the high-dimensional limit of the IAMP algorithm, thus crucial to its success in finding a solution of the spherical perceptron. In \cite{alaoui2020algorithmic, montanari2024exceptional}, the authors exploit the delicate structure of the Parisi variational principle associated with spherical perceptron to construct these Lipschitz functions. The corresponding IAMP algorithm, combined with a rounding procedure, provably outputs a feasible solution within the \emph{full replica symmetry breaking} (full-RSB) region of the Parisi variational problem. However, their analysis of the Parisi variational principle is quite complicated and does not yield an explicit algorithmic threshold.

\paragraph{Heuristic Algorithms Based on Belief Propagation.}
There have also been some heuristic algorithmic advances for the binary perceptron model. These include \cite{BZ06}, which proposes a version of the belief propagation (BP) algorithm with an additional reinforcement term. For $\kappa = 0$, simulations suggest that the algorithm works with capacities up to $\alpha_{\mathrm{BP}} \approx 0.74$, compared to the theoretical limit of $\asat \approx 0.83$. Further developments have yielded simpler algorithms with comparable performance \cite{BBBZ07, baldassi2009generalization, baldassi2015max}.

\subsection{Proof Overview}
We now briefly describe the approaches we use for the proofs of our main results, with detailed proofs given in subsequent sections.

\paragraph{Design rationale for our algorithm.}
Our algorithm combines an initial linear programming (LP) step---adapted from the random projection approach of Rothvoss and Eldan \& Singh \cite{rothvoss2017constructive,eldan2018efficient}---with a sequence of edge-walk rounding steps inspired by the work of Lovett and Meka \cite{lovett2015constructive}. While it is also possible in principle to iterate either the LP step or the edge-walk step alone, we chose this particular combination for both technical and analytical reasons.

Both the LP and edge-walk procedures can be viewed as forms of partial coloring: each step pushes more coordinates of $\theta$ to $\pm 1$ while maintaining feasibility. However, we are able to obtain a much sharper characterization of the LP step when the constraint matrix $X$ is Gaussian---thanks to tools such as Gordon’s inequality. In contrast, after the first LP step, the residual matrix (conditioned on previous outputs) is no longer Gaussian, and our sharp probabilistic tools no longer apply. This limits our ability to iterate the LP step analytically.

On the other hand, while the edge-walk algorithm is more flexible in that it does not rely on the Gaussianity of $X$, it presents its own challenges: each application requires selecting parameters to balance feasibility and rounding progress, which effectively amounts to solving a variational optimization problem. In our setting, we do not have closed-form solutions or tight bounds for these parameters, so iterating the edge-walk step directly becomes difficult to analyze sharply.

By combining the two methods, we take advantage of the best features of each: the first LP step yields a precise understanding of the asymptotic fraction of coordinates that are $\pm 1$, and the precise empirical distribution of the slack vector $X \theta$ under Gaussian input, and the subsequent edge-walk steps allow us to finish the rounding process without requiring the same distributional assumptions. 

\paragraph{Analysis of Our Algorithm.} Our algorithm combines an initial linear programming (LP) step with an iterative rounding step. To get a sharp characterization of the performance of our algorithm, we require precise asymptotics for statistics of the coordinates $\theta_i$ of the LP solution, e.g. the fraction that belong to $\{ \pm 1 \}$. Previous works \cite{rothvoss2017constructive, eldan2018efficient} provide constants that are insufficient for our purposes. Instead, we reduce the problem to analyzing a maximin problem and apply Gordon's inequality (or the Gaussian comparison inequality). This inequality enables us to reduce the analysis of an $M \times N$ Gaussian matrix to one-dimensional Gaussian vectors of lengths $M$ and $N$, whose asymptotics can be understood. This approach allows us to precisely describe the limiting behavior of the proportion of coordinates $\theta_i$ in ${\pm 1}$ and to characterize the limiting empirical distribution of the margins $\langle X_i, \theta \rangle$. See \Cref{thm:LP_behavior}.

More precisely, we reduce the optimal value of the LP to a two-dimensional problem---a max-min optimization over parameters $\rho$ and $\gamma$, which depend on the system parameters $\kappa$ and $\alpha$. The parameter $\rho$ is closely tied to the norm of the LP solution. Intuitively, it captures how many coordinates of $\theta$ are likely to lie in $\{\pm 1\}$. The parameter $\gamma$ is related to the norm of the dual variables associated with the inequality constraints, and encodes the typical “slackness” of those constraints. Together, $\rho$ and $\gamma$ govern the asymptotics of the optimal solution and allow us to precisely track both the fraction of saturated coordinates and the empirical distribution of the margins.

To get a sharp analysis of the iterative rounding steps, we require an analysis which is specialized to our average-case setting.
Following the edge-walk algorithm, we run the random walk algorithm on progressively smaller subsystems (with fewer columns), ensuring that the total additional margins accumulated remain bounded. 
Our challenge lies in the fact that the choice of columns depends on the matrix entries, disrupting the independence of Gaussian entries. A straightforward union bound over all column choices is inadequate. In fact, dependencies can increase the norms of certain matrix entries, but not collectively across all rows. Through careful analysis, we show that the row norms of the matrix restricted to smaller subsystems do not empirically differ significantly from the corresponding i.i.d. Gaussian case, which suffices for our algorithm.

Finally, we specialize our analysis to $\kappa < 0$, $\kappa = 0$ and $\kappa > 0$, respectively---the edge walk algorithm takes some parameters as input which govern the slack of certain constraints, and we optimize these somewhat for each case. As a consequence, we obtain new algorithmic results for binary perceptron in each of these settings.

\paragraph{Satisfiability Phase Transition for $\kappa < 0$.} We perform a straightforward first-moment calculation to obtain the phase transition upper bound. For the lower bound, a standard approach in random constraint satisfaction problems (CSPs) is to apply the second moment method. Specifically, one constructs a non-negative random variable $Z$ such that $Z > 0$ if and only if the problem has solutions. By the Paley--Zygmund inequality, calculating the first two moments of $Z$ provides a lower bound on the probability that the solution set is non-empty: \begin{equation}
    \P (Z > 0) \ge \frac{\E[Z]^2}{\E [Z^2]}.
\end{equation}
The simplest choice of $Z$ for the binary perceptron is the number of solutions. However, an explicit calculation shows that $\E [Z^2] \gg \E[Z]^2$ because, under certain rare events, there are atypically large number of solutions.

A natural idea is to weight each solution $\theta$ according to $\{ \langle X_i, \theta \rangle / \sqrt{N} \}_{i=1}^{M}$, the normalized margins on the data. Weighted second moment methods have already been used in both discrete and continuous random CSPs \cite{achlioptas2005rigorous, montanari2024tractability}. In our case, define 
\begin{equation}
    Z = \sum_{\theta \in \{ \pm 1 \}^N} \prod_{i=1}^{M} f \left( \frac{1}{\sqrt{N}} \left\langle X_i, \theta \right\rangle \right),
\end{equation}
where the weight function $f$ satisfies $f(t) = 0$ if and only if $t < \kappa$. We show that $f$ can be chosen such that $\E [Z^2] = O (\E [Z]^2)$ for $\alpha$ below a certain threshold, which we characterize in the limit $\kappa \to - \infty$. This calculation establishes a non-trivial lower bound on $\P (Z > 0)$. Furthermore, we combine this estimate with existing results on the sharp threshold property to show that $\P (Z > 0)$ actually converges to one, rather than just being bounded away from zero.

\paragraph{Algorithmic Hardness via OGP.} For our hardness result regarding stable algorithms, our proof closely follows the approach in \cite{gamarnik2022algorithms}. Using a first-moment calculation, we establish that, with high probability, a certain forbidden structure does not exist. This excludes the possibility of a stable algorithm through proof by contradiction: if such a stable algorithm were to exist, it would be possible to generate the forbidden structure using the algorithm, contradicting the first-moment computation.

\subsection{Organization}
Our paper is organized as follows:

\begin{itemize}
    \item \cref{sec:sat} discusses the satisfiability phase transition for the case when \(\kappa < 0\). We establish both upper bounds (\cref{thm:neg_upper_bd}) and lower bounds (\cref{thm:neg_lower_bd}) on the satisfiability threshold \(\alpha_\star(\kappa)\) using the first and second moment methods, respectively.

    \item \cref{sec:sat2} addresses the satisfiability phase transition for \(\kappa \geq 0\). We present known results about the upper bound on \(\alpha_\star(\kappa)\) and discuss the lower bound based on our algorithmic results.
    
    \item \cref{sec:alg} introduces our efficient algorithm for finding a \(\kappa\)-margin solution in \(S_{M,N}(X)\). We detail the two-stage process of the algorithm: first, solving a linear program, and then rounding the solution using the edge-walk algorithm of Lovett and Meka. 
    
    More specifically, \cref{sec:gordon} describes the linear programming stage, where we draw on Gaussian comparison inequalities to characterize the limiting behavior of the linear programming algorithm in \cref{thm:LP_behavior}. \cref{sec:rounding} explains the rounding stage, including the application of the edge-walk algorithm and the performance guarantees we obtain. 
The edge-walk algorithm requires us to choose a schedule of margin relaxations, and we have chosen different schedules depending on the value of $\kappa$ so as to optimize our results.
\cref{sec:algthreshold} presents the thresholds for our algorithm as a function of the margin, covering cases of large negative margins in \cref{thm:kappaverynegative}, zero margins in \cref{thm:fullcoloring0}, and large positive margins in \cref{thm:alg_kappa_positive}. \cref{sec:aux_lemmas} collects several auxiliary lemmas that support the theoretical analysis of our algorithm.

    \item \cref{sec:mogp} explores the multi-overlap gap property (m-OGP) in the \(\kappa \to -\infty\) regime. We demonstrate the presence of m-OGP in 
    \cref{sec:ogpstructure}, and discuss its implications for the failure of stable algorithms in 
    \cref{sec:fail_stable_alg}.

    \item \cref{app:improved} contains an improved partial coloring lemma (\cref{thm:improved_Lovett_Meka}) for the margin zero case, which allows us to extend our algorithm to work when $\alpha \leq 0.1$ (\cref{thm:fullcoloring02}).
\end{itemize}

\subsection{Notation}\label{sec:notation}
For any positive integer $n$, we let $[n] = \{1,2,\ldots, n\}$.
 For a scalar $a$, we 
write $a_+ = \max\{a, 0\}$ and $a_- = \max\{-a, 0\}$.
 We use $\norm{u}$ or $\norm{u}_2$ to denote the $\ell^2$-norm of a vector $u$, and we use $\| M \|_\op$ to denote the operator norm of a matrix $M$. 
For an event $\calE$, we use $\Ind[\calE]$ to denote the indicator of $\calE$. For $d \in \mathbb{N}$, we denote by $I_d$ the identity matrix of dimension $d$. For $x \in [0, 1]$, we denote by $\Ent (x) = - x \log x - (1 - x) \log (1 - x)$ the entropy function. We denote by $\delta_y$ the Dirac measure at $y \in \R$. For any $\theta, \theta' \in \{ \pm 1 \}^N$, $d_{H} (\theta, \theta')$ denotes their Hamming distance.

We will use $O_n(\cdot)$ and $o_n(\cdot)$ for the standard big-$O$ and small-$o$ notation, 
where $n$ is the asymptotic variable.
For two vectors $u, v$ of the same dimension, we may write $u \ge v$ or $u \le v$ to mean entry-wise inequality.
We always use $\Phi$ and $\phi$ to denote the CDF and PDF of standard normal distribution, respectively. 
We write $X \perp\!\!\!\perp Y$ if $X$ and $Y$ are two independent random variables (or random vectors).
For a sequence of random variables $\{ X_n \}$ and $c \in \R$, we write $\plim_{n \to \infty} X_n = c$ if $X_n$ converges to $c$ in probability.
Throughout this paper, we assume the asymptotics $\lim_{N \to \infty} (M/N) = \alpha \in (0, \infty)$.

We denote by $\normal(\mu, \Sigma)$ the Gaussian distribution with mean vector $\mu$ and covariance matrix $\Sigma$. For $V \subset \R^N$ a linear subspace, we denote by $\normal(V)$ the standard multi-dimensional Gaussian distribution supported on $V$. Namely, $G \sim \normal(V)$ if $G = \sum_{i=1}^{d} G_i v_i$ where $d = \dim V$, $G_i \iidsim \normal(0, 1)$, and $\{ v_1, \cdots, v_d \}$ is an orthonormal basis for $V$.

\section{Satisfiability Phase Transitions}
\subsection{Satisfiability Phase Transition for $\kappa < 0$}\label{sec:sat}

We now establish upper and lower bounds on the satisfiability threshold $\asat (\kappa)$ in the case $\kappa < 0$.
We will make use of the first- and second-moment methods, respectively.

\begin{thm}[Upper bound on $\asat (\kappa)$]\label{thm:neg_upper_bd}
	Define $\aup (\kappa) = - \log 2 / \log (1 - \Phi(\kappa))$, then $\asat (\kappa) \le \aup (\kappa)$.
Equivalently, if $\alpha > \aup (\kappa)$, with high probability we have $S_{M, N} (X) = \varnothing$.
\end{thm}

\begin{proof}
	We use the first moment method.
By direct calculation,
	\begin{equation*}
		\E \left[ \left\vert S_{M, N} (X) \right\vert \right] = 2^N \left( 1 - \Phi(\kappa) \right)^M = \exp \left( N \left( \log 2 + \frac{M}{N} \log \left( 1 - \Phi(\kappa) \right) \right) \right).
	\end{equation*}
	Therefore, as long as $\alpha > \aup (\kappa)$, $\E \left[ \left\vert S_{M, N} (X) \right\vert \right]$ will be exponentially small as $M, N \to \infty$, which concludes the proof.
\end{proof}

\begin{defn}[Lower bound on $\asat (\kappa)$]\label{def:neg_lower_bd}
		For $\kappa < 0$, let $c_*=c_*(\kappa) > 0$ be the unique positive solution of the equation 
	\begin{align*}
		c (1 - \Phi(\kappa + c)) = \phi (\kappa + c)\, .
	\end{align*}
	(Existence and uniqueness are guaranteed by \cref{lem:property_weight} $(a)$.) For $\alpha > 0$ and $q \in [-1, 1]$, define
	\begin{equation}
		\Psi \left( q; \alpha \right) = \alpha \log \E_q \left[ \exp(-c_*(\kappa) (G_1 + G_2)) \bone \left\{ G_1 \ge \kappa, G_2 \ge \kappa \right\} \right] + \Ent \left( \frac{1+q}{2} \right) + \log 2,\label{eq:psi}
	\end{equation}
	where $\E_q$ denotes expectation with respect to
	$(G_1, G_2)^\top \sim \normal \left( 0, \begin{bmatrix}
		1 & q \\
		q & 1
	\end{bmatrix} \right)$, and $\Ent(\cdot)$ is the entropy function.
	Finally, we define 
	\begin{align}
		\alpha_{\rm low} (\kappa):=\sup\Big\{
		\alpha > 0 \mbox{ s.t.  } \Psi \left( q; \alpha \right) \mbox{ is uniquely maximized at } 
		q = 0 \mbox{ and } \Psi'' \left(0; \alpha \right) < 0\Big\}\,,
	\end{align}	
	where in the above display, the second derivative of $\Psi$ is taken with respect to $q$.
\end{defn}

To establish the lower bound, we apply a weighted second-moment method. The definition above ensures that the weight function is chosen appropriately so that the second-moment argument is valid. Specifically, we use a weight proportional to $\exp (- c_*(\kappa) x) \bone \{ x \ge \kappa \}$. The condition that $\Psi \left( q; \alpha \right) $ is uniquely maximized at $q = 0$ guarantees that the dominant contribution to the second moment arises from pairs of solutions with approximately zero overlap. This structure is essential to ensure that the second moment is comparable to the square of the first moment.

\begin{thm}[Lower bound on $\asat (\kappa)$]\label{thm:neg_lower_bd}
	Let $\alow (\kappa)$ be defined as per \cref{def:neg_lower_bd}, then $\asat (\kappa) \ge \alow (\kappa)$.
Equivalently, if $\alpha < \alow (\kappa)$, with high probability we have $S_{M, N} (X) \neq \varnothing$.
\end{thm}

\begin{proof}
	For future convenience, we denote $\alpha_{M, N} = M/N$.
	By the sharp threshold property of Binary Perceptron \cite{talagrand1999self, xu2021sharp, NS23}, it suffices to show that
	\begin{equation}
		\liminf_{M, N \to \infty} \P \left( S_{M, N} (X) \neq \varnothing \right) > 0.
	\end{equation}
	To this end, we employ a reweighted second moment argument.
The argument is similar to that of \cite{montanari2024tractability}'s lower bound on the satisfiability phase transition in the Asymmetric Spherical Perceptron.
	We make use of the following characterizations that they have made of terms from \Cref{eq:psi}.

\begin{lem}[Lemma 1, Appendix A.1 in \cite{montanari2024tractability}]\label{lem:property_weight}
	The followings hold:
	\begin{enumerate}
		\item[(a)] For any $\kappa < 0$, there exists a unique $c_* = c_* (\kappa) > 0$ satisfying $c \left( 1 - \Phi (\kappa +c) \right) = \phi (\kappa + c)$.
Moreover, $c_* = O_{\kappa} (\exp(- \kappa^2/2))$ as $\kappa \to -\infty$.
		
        \item[(b)] Let $\kappa$ and $c_*$ be as described in (a), and set $f (x) = \exp (- c_*(\kappa) x) \bone \{ x \ge \kappa \}$.
Define
	\begin{equation}\label{eq:def_of_e(q)}
		e (q) = \E_q [f(G_1) f(G_2)],
	\end{equation}
    then it follows that $e'(0) = 0$, $e(0) = 1 + o_{\kappa} (1)$, and $e(1) - e(0) = (1 + o_{\kappa} (1)) \Phi (\kappa)$.
    
    \item[(c)] Uniformly for $q \in [-1, 1]$, we have
    \begin{equation}\label{eq:asymptotics_for_e'(q)}
    \begin{split}
    e'(q) \, = \frac{1+ o_{\kappa}(1)}{2 \pi \sqrt{1-q^{2}}} \exp \left(-\frac{\kappa^{2}}{1+q}\right)+ O_{\kappa} \left(\exp (-\kappa^2) \right).
    	\end{split}
    \end{equation}
    Further, for any $q \in [0, 1)$ there exists a constant $C \in (0, +\infty)$ such that
    \begin{equation}\label{eq:bound_on_e''(q)}
        \sup_{u \in [-q, q]} \vert e''(u) \vert \le \vert \kappa \vert^C \exp \left( - \frac{\kappa^2}{1+q} \right).
    \end{equation}

	\end{enumerate}
\end{lem}

Now, let $f$ be as defined in \cref{lem:property_weight} $(b)$, and define the random variable
	\begin{equation}
		Z_{M, N} = \sum_{\theta \in \{ \pm 1 \}^N} \prod_{i=1}^{M} f \left( \frac{1}{\sqrt{N}} \left\langle X_i, \theta \right\rangle \right).
	\end{equation}
	Since $f(x) > 0$ if and only if $x \ge \kappa$, it follows that $Z_{M, N} > 0 \Longleftrightarrow S_{M, N} (X) \neq \varnothing$.
By the Paley-Zygmund inequality, we have
	\begin{equation}
		\P \left( Z_{M, N} > 0 \right) = \E \left[ \bone \left\{ Z_{M, N} > 0 \right\} \right] \ge \frac{\E \left[ Z_{M, N} \bone \left\{ Z_{M, N} > 0 \right\} \right]^2}{\E \left[ Z_{M, N}^2 \right]} = \frac{\E \left[ Z_{M, N} \right]^2}{\E \left[ Z_{M, N}^2 \right]},
	\end{equation} 
	where the intermediate bound follows from Cauchy-Schwarz inequality.
Our aim is to show that $\alpha < \alow (\kappa)$ implies $\E [ Z_{M, N}^2 ] = O_{M, N} (\E [Z_{M, N}]^2)$, so that
	\begin{equation}
		\liminf_{M, N \to \infty} \P \left( S_{M, N} (X) \neq \varnothing \right) = \liminf_{M, N \to \infty} \P \left( Z_{M, N} > 0 \right) > 0.
	\end{equation}
	To this end, we compute and compare the first and second moments of $Z_{M, N}$.
By direct calculation,
	\begin{align*}
		\E \left[ Z_{M, N} \right] = \, & \sum_{\theta \in \{ \pm 1 \}^N} \E \left[ \prod_{i=1}^{M} f \left( \frac{1}{\sqrt{N}} \left\langle X_i, \theta \right\rangle \right) \right] = 2^N \E \left[ f(G) \right]^M \\
		= \, & \exp \left( N \left( \log 2 + \alpha_{M, N} \log \E \left[ f(G) \right] \right) \right) \\
		= \, & \exp \left( N \Psi \left( 0; \alpha_{M, N} \right) / 2 \right).
	\end{align*}
	To calculate $\E [Z_{M, N}^2]$, we first write $Z_{M, N}^2$ as a double summation:
	\begin{align*}
		Z_{M, N}^2 = \, \sum_{\theta_1 \in \{ \pm 1 \}^N} \sum_{\theta_2 \in \{ \pm 1 \}^N} \prod_{i=1}^{M} f \left( \frac{1}{\sqrt{N}} \left\langle X_i, \theta_1 \right\rangle \right) f \left( \frac{1}{\sqrt{N}} \left\langle X_i, \theta_2 \right\rangle \right).
	\end{align*}
	Next, denote $q = \langle \theta_1, \theta_2 \rangle / N$, we observe the following two useful facts:
	\begin{enumerate}
		\item The number of pairs $(\theta_1, \theta_2) \in \{ \pm 1 \}^N \times \{ \pm 1 \}^N$ satisfying $q = \langle \theta_1, \theta_2 \rangle / N$ is
			\begin{equation*}
				2^N \times \binom{N}{\frac{1+q}{2} N}.
			\end{equation*}
		\item For $i = 1, \cdots, M$, we have
			\begin{equation*}
				\left( \frac{1}{\sqrt{N}} \left\langle X_i, \theta_1 \right\rangle, \frac{1}{\sqrt{N}} \left\langle X_i, \theta_2 \right\rangle \right) \iidsim \normal \left( 0, \begin{bmatrix}
				1 & q \\
				q & 1
				\end{bmatrix} \right).
			\end{equation*}
	\end{enumerate}
	Therefore, it follows that
	\begin{align}
		\E \left[ Z_{M, N}^2 \right] = \, & \sum_{\theta_1 \in \{ \pm 1 \}^N} \sum_{\theta_2 \in \{ \pm 1 \}^N} \E \left[ \prod_{i=1}^{M} f \left( \frac{1}{\sqrt{N}} \left\langle X_i, \theta_1 \right\rangle \right) f \left( \frac{1}{\sqrt{N}} \left\langle X_i, \theta_2 \right\rangle \right) \right] \\
		= \, & \sum_{q \in Q_N} 2^N \times \binom{N}{\frac{1+q}{2} N} \times \E_q \left[ f(G_1) f(G_2) \right]^M,
	\end{align}
	where $Q_N = \{ -1, -1 + \frac{2}{N}, \ldots, 1 - \frac{2}{N}, 1 \}$ is the set of all possible values of $q$.
Now for any $K \in \mathbb{N}$, $K \ge 4$ that does not depend on $N$, define 
	\begin{equation*}
		Q_{N, K} = \{ -1 + 2\tfrac{K}{N}, -1 + 2\tfrac{K+1}{N}, \ldots, 1 - 2 \tfrac{K+1}{N}, 1 - 2\tfrac{K}{N} \}.
	\end{equation*}
	The following estimate is due to Stirling's formula:
	\begin{align}
		\binom{N}{\frac{1+q}{2} N} = \, & \left( 1 + o_K(1) \right) \frac{\sqrt{2}}{\sqrt{\pi N (1 - q^2)}} \times \exp \left( N \Ent \left( \frac{1 + q}{2} \right) \right), \quad q \in Q_{N, K}, \\
		\binom{N}{\frac{1+q}{2} N} \le \, & \binom{N}{K - 1} \le \frac{N^{K-1}}{(K-1)!}, \quad q \in Q_N \backslash Q_{N, K}.
	\end{align}
	It then follows that
	\begin{align}
		&\E \left[ Z_{M, N}^2 \right] \le \,  \sum_{q \in Q_{N, K}} \frac{(1 + o_K (1)) \sqrt{2}}{\sqrt{\pi N (1 - q^2)}} \times \exp \left( N \left( \log 2 + \Ent \left( \frac{1 + q}{2} \right) + \alpha_{M, N} \log \E_q \left[ f(G_1) f(G_2) \right] \right) \right) \\
		& \hspace{1.6cm} + \sum_{q \in Q_N \backslash Q_{N, K}} \frac{N^{K-1}}{(K-1)!} \times \exp \left( N \left( \log 2 + \alpha_{M, N} \log \E_q \left[ f(G_1) f(G_2) \right] \right) \right) \\
		& \stackrel{(i)}{\le} \,  \sum_{q \in Q_{N, K}} \frac{(1 + o_K (1)) \sqrt{2}}{\sqrt{\pi N (1 - q^2)}} \times \exp \left( N \Psi \left( q; \alpha_{M, N} \right) \right) + 2 N^{K - 1} \exp \left( N \sup_{q \in Q_N \backslash Q_{N, K}} \Psi \left( q; \alpha_{M, N} \right) \right),
	\end{align}
	where $(i)$ follows from the fact that $\vert Q_N \backslash Q_{N, K} \vert = 2K \le 2(K-1)!$. 
Now since $\lim_{M, N \to \infty} \alpha_{M, N} = \alpha < \alow (\kappa)$, we know that $\alpha_{M, N} < \alow (\kappa)$ for sufficiently large $M, N$.
According to \cref{def:neg_lower_bd}, $\sup_{q \in Q_N \backslash Q_{N, K}} \Psi \left( q; \alpha_{M, N} \right) \le \Psi \left( 0; \alpha_{M, N} \right) - \veps$ for some $\veps > 0$. 
We thus obtain that 
	\begin{align}
		\E \left[ Z_{M, N}^2 \right] \le \, \left( 1 +o_N (1) \right) \sum_{q \in Q_{N, K}} \frac{(1 + o_K (1)) \sqrt{2}}{\sqrt{\pi N (1 - q^2)}} \times \exp \left( N \Psi \left( q; \alpha_{M, N} \right) \right).
	\end{align}
	Similarly, since $\sup_{1/2 \le \vert q \vert \le 1} \Psi(q; \alpha_{M, N}) \le \Psi (0; \alpha_{M, N}) - \veps$, it follows that
	\begin{align}
		\E \left[ Z_{M, N}^2 \right] \le \, & \left( 1 +o_N (1) \right) \sum_{q \in Q_{N, K} \cap [-1/2, 1/2]} \frac{(1 + o_K (1)) \sqrt{2}}{\sqrt{\pi N (1 - q^2)}} \times \exp \left( N \Psi \left( q; \alpha_{M, N} \right) \right) \\
		\le \, & \left( 1 + o_N (1) \right) \sum_{q \in Q_{N, K} \cap [-1/2, 1/2]} \frac{\sqrt{2}}{\sqrt{\pi N (1 - q^2)}} \times \exp \left( N \Psi \left( q; \alpha_{M, N} \right) \right) \\
		\le \, & \left( 1 + o_N (1) \right) \sqrt{\frac{N}{2 \pi}} \times \int_{-1/2}^{1/2} \frac{1}{\sqrt{1 - q^2}} \exp \left( N \Psi \left( q; \alpha_{M, N} \right) \right) \d q.
	\end{align}
	Using Laplace's method (see, e.g., Chapter 4.2 of \cite{de1981asymptotic}), we obtain that
	\begin{align}
		\E \left[ Z_{M, N}^2 \right] \le \, & \left( 1 + o_N (1) \right) \sqrt{\frac{N}{2 \pi}} \times \sqrt{\frac{2 \pi}{- N \Psi'' \left( 0; \alpha_{M, N} \right)}} \times \exp \left( N \Psi \left( 0; \alpha_{M, N} \right) \right) \\
		= \, & \left( 1 + o_N (1) \right) \times \sqrt{\frac{1}{-  \Psi'' \left( 0; \alpha_{M, N} \right)}} \times \exp \left( N \Psi \left( 0; \alpha_{M, N} \right) \right) \\
		= \, & \left( 1 + o_N (1) \right) \times \sqrt{\frac{1}{-  \Psi'' \left( 0; \alpha_{M, N} \right)}} \times \E \left[ Z_{M, N} \right]^2 \\
		= \, & O_{M, N} \left( \E \left[ Z_{M, N} \right]^2 \right).
	\end{align}
	This completes the proof.
\end{proof}

We now compare $\alow (\kappa)$ and $\aup (\kappa)$, establishing that they are within a $1 + o_{\kappa} (1)$ multiplicative factor of each other as $\kappa \to - \infty$.
To be precise, we have the following:
\begin{thm}
	For any $\veps > 0$, there exists a $\kappa_0 = \kappa_0 (\veps) < 0$, such that for all $\kappa \le \kappa_0$,
	\begin{equation}
		(1 - \veps) \frac{\log 2}{\Phi(\kappa)} \le \alow (\kappa) \le \aup (\kappa) \le (1 + \veps) \frac{\log 2}{\Phi(\kappa)}.
	\end{equation}
	In other words, both $\alow (\kappa)$ and $\aup (\kappa)$ are $(1+o_{\kappa} (1)) \log 2 / \Phi(\kappa)$ as $\kappa \to - \infty$.
\end{thm}

\begin{proof}
	Since $- \log (1 - x) \sim x$ as $x \to 0$, it follows immediately that $\aup (\kappa) \sim \log 2 / \Phi(\kappa)$ as $\kappa \to - \infty$.
It then remains to show that $\alow (\kappa) \ge (1 - \veps) \log 2 / \Phi(\kappa)$ for sufficiently negative $\kappa$.
For any fixed $\veps > 0$, we will prove that
	\begin{equation}
		\alpha < (1 - \veps) \frac{\log 2}{\Phi(\kappa)} \implies \Psi \left( q; \alpha \right) \mbox{ is uniquely maximized at } 
		q = 0 \mbox{ and } \Psi'' \left(0; \alpha \right) < 0,
	\end{equation}
	where $\Psi ( q; \alpha )$ is defined in \cref{def:neg_lower_bd}.
Using the notation of \cref{lem:property_weight}, we write
	\begin{equation}
		\Psi \left( q; \alpha \right) = \alpha \log e(q) + \Ent \left( \frac{1+q}{2} \right) + \log 2,
	\end{equation}
	which leads to 
	\begin{align*}
		\Psi \left( q; \alpha \right) - \Psi \left( 0; \alpha \right) = \, & \alpha \log \frac{e(q)}{e(0)} + \Ent \left( \frac{1+q}{2} \right) - \log 2 \\
		= \, & \alpha \log \left( 1 + \frac{e(q) - e(0)}{e(0)} \right) + \Ent \left( \frac{1+q}{2} \right) - \log 2 \\
		\stackrel{(i)}{\le} \, & \alpha \frac{e(q) - e(0)}{e(0)} - \left( \log 2 - \Ent \left( \frac{1+q}{2} \right) \right) \\
		= \, & \frac{\alpha}{e(0)} \int_{0}^{q} e'(s) \d s - \left( \log 2 - \Ent \left( \frac{1+q}{2} \right) \right),
	\end{align*}
	where in $(i)$ we use the well-known inequality $\log(1 + x) \le x$.
	We thus only need to show that $\Psi'' \left(0; \alpha \right) < 0$, and that
	\begin{equation}\label{eq:cnt_fct_max}
		\frac{\alpha}{e(0)} \int_{0}^{q} e'(s) \d s < \log 2 - \Ent \left( \frac{1+q}{2} \right), \quad \forall q \in [-1, 1] \backslash \{ 0 \}.
	\end{equation}
	By direct calculation, we get
	\begin{align*}
		\Psi'' \left(0; \alpha \right) = \alpha \frac{e''(0) e(0) - e'(0)^2}{e(0)^2} - 1 = \alpha \frac{e''(0)}{e(0)} - 1, 
	\end{align*}
	since from \cref{lem:property_weight} (b) we know that $e'(0) = 0$.
Further, we know that $e(0) = 1 + o_{\kappa} (1)$ and $\vert e''(0) \vert \le \vert \kappa \vert^C \exp(- \kappa^2)$ from \cref{lem:property_weight}, thus leading to
	\begin{equation*}
		\Psi'' \left(0; \alpha \right) \le (1 - \veps) \frac{\log 2}{\Phi(\kappa)} \cdot \vert \kappa \vert^C \exp \left( -\kappa^2 \right) - 1 \le \vert \kappa \vert^{C'} \exp \left( - \frac{\kappa^2}{2} \right) - 1 < 0
	\end{equation*}
	for sufficiently large $\kappa < 0$.
The rest of this proof will be devoted to proving \cref{eq:cnt_fct_max}.
We consider three cases:
	\begin{itemize}
		\item [(1)] $q \in [-1/2, 1/2] \backslash \{ 0 \}$.
Notice from \cref{lem:property_weight} (b) that $e'(0) = 0$, $e(0) = 1 + o_{\kappa} (1)$, we get
			\begin{align*}
				\frac{\alpha}{e(0)} \int_{0}^{q} e'(s) \d s = \, & \left( 1 + o_{\kappa} (1) \right) \alpha \int_{0}^{q} \d s \int_{0}^{s} e''(u) \d u \\
				\stackrel{(i)}{\le} \, & \alpha \int_{0}^{q} \d s \int_{0}^{s} \vert \kappa \vert^C \exp \left( - \frac{\kappa^2}{1+q} \right) \d u \\
				\le \, & \frac{\log 2}{\Phi (\kappa)} \cdot \vert \kappa \vert^C \exp \left( - \frac{2 \kappa^2}{3} \right) \cdot \frac{q^2}{2} \\
				\le \, & \vert \kappa \vert^{C'} \exp \left( - \frac{\kappa^2}{6} \right) \cdot q^2 = o_{\kappa} (1) q^2,
			\end{align*}
			where $(i)$ is due to \cref{lem:property_weight} (c).
It is easy to check that
			\begin{equation}
				\log 2 - \Ent \left( \frac{1+q}{2} \right) \ge \frac{q^2}{2}, \quad \forall q \in [-1, 1].
			\end{equation}
			Hence, \cref{eq:cnt_fct_max} holds for all $q \in [-1/2, 1/2] \backslash \{ 0 \}$.
		\item [(2)] $q \in [-1, -1/2]$.
In this case, with the aid of \cref{eq:asymptotics_for_e'(q)} one gets that
\begin{align*}
	\frac{\alpha}{e(0)} \int_{0}^{q} e'(s) \d s = \, & (1 + o_{\kappa} (1)) \alpha \int_{0}^{q} \left( \frac{1 + o_{\kappa}(1)}{2 \pi \sqrt{1-s^{2}}} \exp \left(-\frac{\kappa^{2}}{1+s}\right)+ O_{\kappa} \left(\exp (-\kappa^2) \right) \right) \d s \\
	\stackrel{(i)}{\le} \, & C \alpha \left( \int_{0}^{\vert q \vert} \frac{\d s}{\sqrt{1 - s^2}} + \vert q \vert \right) \exp \left( - \kappa^2 \right) \stackrel{(ii)}{\le} \, C \alpha \exp \left( - \kappa^2 \right) \\
	\le \, & \frac{C}{\Phi(\kappa)} \exp \left( - \kappa^2 \right) \le \, C \vert \kappa \vert \exp \left( - \frac{\kappa^2}{2} \right) = o_{\kappa} (1),
\end{align*}
where $(i)$ holds since $\exp(- \kappa^2/(1+s)) \le \exp(- \kappa^2)$ when $s \le 0$, and $(ii)$ follows from the fact
\begin{equation*}
	\int_{0}^{\vert q \vert} \frac{\d s}{\sqrt{1 - s^2}} \le \int_{0}^{1} \frac{\d s}{\sqrt{1 - s^2}} = \frac{\pi}{2}.
\end{equation*}
Therefore, for sufficiently negative $\kappa$ we have
\begin{equation*}
	\frac{\alpha}{e(0)} \int_{0}^{q} e'(s) \d s = o_{\kappa} (1) < \log 2 - \Ent \left( \frac{1}{4} \right) \le \log 2 - \Ent \left( \frac{1+q}{2} \right), \quad \forall q \in \left[ -1, -\frac{1}{2} \right].
\end{equation*}

	\item [(3)] $q \in [1/2, 1]$.
Define an auxiliary function
		\begin{equation}
			F(q) = \frac{\alpha}{e(0)} \int_{0}^{q} e'(s) \d s + \Ent \left( \frac{1+q}{2} \right).
		\end{equation}
		Then, we know from \cref{lem:property_weight} that
		\begin{align*}
			F'(q) = \, & \frac{\alpha e'(q)}{e(0)} - \frac{1}{2} \log \left( \frac{1+q}{1-q} \right) = \, (1 + o_{\kappa} (1)) \alpha e'(q) - \frac{1}{2} \log \left( \frac{1+q}{1-q} \right) \\
			\stackrel{(i)}{\le} \, & \left( 1 - \frac{\veps}{2} \right) \frac{\log 2}{\Phi (\kappa)} \frac{1}{2 \pi \sqrt{1-q^{2}}} \exp \left(-\frac{\kappa^{2}}{1+q}\right) - \frac{1}{2} \log \left( \frac{1+q}{1-q} \right) \\
			\le \, & \left( 1 - \frac{\veps}{2} \right) \frac{\log 2}{\sqrt{2 \pi}} \frac{\vert \kappa \vert}{\sqrt{1-q^{2}}} \exp \left(-\frac{(1 - q)\kappa^{2}}{2(1+q)}\right) - \frac{1}{2} \log \left( \frac{1+q}{1-q} \right),
		\end{align*}
		where $(i)$ follows from \cref{eq:asymptotics_for_e'(q)} and the fact that
		\begin{equation*}
			O_{\kappa} \left(\exp (-\kappa^2) \right) = o_{\kappa} \left( \frac{1+ o_{\kappa}(1)}{2 \pi \sqrt{1-q^{2}}} \exp \left(-\frac{\kappa^{2}}{1+q}\right) \right), \quad \text{uniformly for} \ q \in \left[ \frac{1}{2}, 1 \right].
		\end{equation*}
		For any fixed $\eta \in (0, 1/2)$, we have $F'(q) < 0$ for all $q \in [1/2, 1 - \eta]$ and large enough $\kappa$.
As a consequence, the supremum of $F(q)$ over $[1/2, 1]$ can only be attained at $1/2$ or in $[1- \eta, 1]$, which further leads to the estimate:
		\begin{align}
			\sup_{q \in [1/2, 1]} F(q) \le \max \left\{ F \left( \frac{1}{2} \right), \ \sup_{q \in [1 - \eta, 1]} F(q)  \right\}.
		\end{align}
		Note that from our conclusion in case (1), $F(1/2) < \log 2$.
As for the second term, since $e'(q) \ge 0$ when $q \in [0, 1]$, we deduce that for all $q \in [1 - \eta, 1]$:
		\begin{align*}
			F(q) \le \, & \alpha \frac{e(1) - e(0)}{e(0)} + \Ent \left( 1 - \frac{\eta}{2} \right) = \left( 1 + o_{\kappa} (1) \right) \alpha \Phi(\kappa) + \Ent \left( 1 - \frac{\eta}{2} \right) \\
			\le \, & \left( 1 - \frac{\veps}{2} \right) \log 2 + \Ent \left( 1 - \frac{\eta}{2} \right) < \log 2
		\end{align*}
		if we choose $\eta$ to be such that $\Ent(1 - \eta/2) < (\veps / 2) \log 2$.
This proves $\sup_{q \in [1 - \eta, 1]} F(q) < \log 2$.
We finally conclude that
		\begin{align*}
			\sup_{q \in [1/2, 1]} F(q) < \log 2 \implies \frac{\alpha}{e(0)} \int_{0}^{q} e'(s) \d s < \log 2 - \Ent \left( \frac{1+q}{2} \right), \quad \forall q \in \left[ \frac{1}{2}, 1 \right].
		\end{align*}
	\end{itemize}
	Now we have proved that \cref{eq:cnt_fct_max} holds for all $q \in [-1, 1] \backslash \{ 0 \}$.
This completes the proof.
\end{proof}

\subsection{Satisfiability Phase Transition for $\kappa \ge 0$} \label{sec:sat2}

It has been proved in \cite{stojnic2013discrete} that for $\kappa \ge 0$,
\begin{equation*}
	\aup (\kappa) = \frac{2}{\pi} \cdot \E_{G \sim \normal(0, 1)} \left[ ( \kappa - G )_+^2 \right]^{-1}
\end{equation*}
is an upper bound on $\asat (\kappa)$ (in fact, for all $\kappa \in \R$, although this upper bound is worse than the one from first moment calculation for large negative $\kappa$).
Further, as $\kappa \to +\infty$, we have $\aup (\kappa) \sim (2 / \pi) \kappa^{-2}$.
Later in \cref{thm:alg_kappa_positive}, we will see that the algorithmic threshold is, likewise, $\alpha_{\rm alg} (\kappa)\sim (2 / \pi) \kappa^{-2}$ for large positive $\kappa$, implying that the computational-to-statistical gap (if any) is at most $1 + o_{\kappa} (1)$ as $\kappa \to \infty$.

\section{An Efficient Algorithm} \label{sec:alg} 
Our algorithm for finding a $\kappa$-margin solution in $S_{M, N} (X)$ includes two stages:
in the first step, we compute a solution to the linear program
\[
			\hat\theta = \argmax_\theta \ \left\langle v, \theta \right\rangle, \ \mbox{subject to} \ X \theta \ge \kappa \sqrt{N} \mathrm{1}, \ \mbox{and} \ \theta \in [-1, 1]^N
\]
for $v \sim \normal(0,I_N)$ a random direction independent of $X$.
Existing analyses \cite{rothvoss2017constructive,eldan2018efficient} imply that under very mild (worst-case) conditions on $X$, $\hat\theta$ should have a constant fraction of its entries set to signs $\pm 1$, so long as $\alpha < \alpha_0$, some fixed constant.
Here, since $X$ is Gaussian and independent of the random direction $v$, in \cref{sec:gordon}, we will be able to make use of Gaussian comparison inequalities to precisely characterize $\hat \theta$ at most densities $\alpha$, including getting a very good bound ($1-o_\kappa(1)$) on the fraction of entries of $\hat \theta$ which are fixed to $\pm 1$, and characterizing the empirical distribution of the margins $\{\iprod{X_i,\hat\theta}\}_{i\in [M]}$.

Suppose $S$ is the subset of entries of $\hat\theta$ that are fixed to signs, and denote by $\hat \theta_S$ the restriction of $\hat\theta$ to the corresponding coordinates, and by $X_S$ the restriction of $X$ to the corresponding columns.
After this first step, $X_S \hat\theta_S + X_{\bar{S}}\hat\theta_{\bar{S}} = X \hat\theta \ge \kappa \sqrt{N} \mathrm{1} $.
Following the usual approach in discrepancy algorithms, we will iterate on the unfixed coordinates of $\hat\theta$, by solving for a $\theta'$ in $\{ \pm 1 \}^{\bar{S}}$ satisfying $X_{\bar{S}} \theta' \ge \kappa \sqrt{N} \mathrm{1} - X_S\hat{\theta}_S$.
However, we cannot simply recycle our Gaussian comparison inequality analysis for the linear program, since now the hyperplane constraints $X_{\bar{S}}$ are not Gaussian any longer and have a nontrivial dependence on the intercept via $\hat\theta_S$.

Instead, we will round the remaining coordinates to signs using the Edge-Walk algorithm of Lovett and Meka \cite{lovett2015constructive}, using $\hat\theta_{\bar{S}}$ as a feasible starting point.
Again, an off-the-shelf worst-case analysis of \cite{lovett2015constructive} would allow us to conclude that the rounding step succeeds at some constant density.
Here, in order to obtain a sharper result, we use the precise information we have about the empirical distribution of the margins from step 1.

\subsection{Stage 1: Linear Programming}\label{sec:gordon}
The first stage of our algorithm is based on linear programming and inspired by the algorithms for discrepancy minimization in \cite{rothvoss2017constructive, eldan2018efficient}.
It is also the binary version of the linear programming algorithm for spherical perceptron in \cite{montanari2024tractability}:
\begin{enumerate}
	\item Draw $v \sim \normal(0, I_N)$ independent of the data matrix $X$.
	\item Solve the following linear programming problem:
		\begin{equation}\label{eq:LP_def}
			\mbox{maximize} \ \left\langle v, \theta \right\rangle, \ \mbox{subject to} \ X \theta \ge \kappa \sqrt{N} \mathrm{1}, \ \mbox{and} \ \theta \in [-1, 1]^N.
		\end{equation}
\end{enumerate}
Denote by $\hat{\theta}$ the output of this algorithm.
We will show that with high probability, a non-zero constant fraction of the coordinates $\{ \hat{\theta}_i \}_{i \in [N]}$ are in $\{ \pm 1 \}$, and obtain exact asymptotics for this fraction in the limit $M/N \to \alpha$.
As a side result, we also derive the limiting empirical distribution of the margins $\{ \langle X_i, \hat{\theta} \rangle / \sqrt{N} \}_{i=1}^{M}$, which turns out to be crucial for the second (rounding) stage of our algorithm.

In order to state our main result, we need the following definition.
\begin{defn}\label{def:LP_auxiliary}
	Let $G$ be a standard normal random variable.
For $\rho \in [0, 1]$, define
	\begin{equation}
		\varphi (\rho) = \, \inf_{t \ge 0} \left\{ \frac{\rho^2 t}{2} + \frac{1}{2 t} \E \left[ \min \left( \vert G \vert, t \right)^2 \right] + \E \left[ \left( \vert G \vert - t \right)_+ \right] \right\}
	\end{equation}
	and $t(\rho)$ as the unique $t \ge 0$ that achieves the above infimum (existence and uniqueness are guaranteed by \cref{lem:property_phi}).
Further, assume that
	\begin{equation}\label{eq:abs_bound_alpha}
		\alpha < \sup_{\rho \in [0, 1]} \left\{ \frac{\varphi(\rho)^2}{\E \left[ \left( \kappa - \rho G \right)_+^2 \right]} \right\},
	\end{equation}
	and let $(\rho, \gamma) = (\rho(\alpha, \kappa), \gamma(\alpha, \kappa))$ be the solution to the maximin problem
	\begin{equation}\label{eq:key_max_min}
		\max_{\rho \in [0, 1]} \min_{\gamma \ge 0} \left\{ - \gamma \sqrt{\alpha} \sqrt{\E \left[ \left( \kappa - \rho G \right)_+^2 \right]} + \sqrt{1 + \gamma^2} \varphi(\rho) \right\}
	\end{equation}
	(existence and uniqueness guaranteed by \cref{lem:property_maxmin}).
Finally, we define 
	\begin{equation}\label{def:tP}
		t(\alpha, \kappa) = t(\rho(\alpha, \kappa)) \ \mbox{and} \ P(\alpha, \kappa) = \operatorname{Law} \left( \max \left( \rho(\alpha, \kappa) G, \kappa \right) \right).
	\end{equation}
\end{defn}

The definitions of $\rho$, $\gamma$, and $t$ arise naturally from the analysis of the underlying optimization problem. Specifically, $\rho$ corresponds to the normalized norm of the variable $\theta$, while $\gamma$ captures the normalized norm of the dual variable associated with the constraint $\frac{1}{\sqrt{N}} X \theta \ge \kappa \mathbf{1}$. The parameter $t$, introduced in \cref{lem:calculus_1}, serves as an effective truncation level for a given vector $y$, balancing quadratic and linear penalties. This reduces the original high-dimensional optimization to a one-dimensional problem that preserves the essential trade-off structure.

\begin{thm}[Limiting behavior of the linear programming algorithm]\label{thm:LP_behavior} Under the settings of \cref{def:LP_auxiliary}, we have
	\begin{align}
		\label{eq:tight_fraction}
		\pliminf_{M/N \to \alpha} \frac{1}{N} \left\vert \left\{ i \in [N] : \hat{\theta}_i \in \{ \pm 1 \} \right\} \right\vert \ge \, & 2 \Phi \left( -t(\alpha, \kappa) \right), \\
		\label{eq:margin_conv}
		\plim_{M/N \to \alpha} W_2 \left( \frac{1}{M} \sum_{i=1}^{M} \delta_{\langle X_i, \hat{\theta} \rangle / \sqrt{N}}, \ P(\alpha, \kappa) \right) = \, & 0.
	\end{align}
\end{thm}

\begin{proof}
    We will make use of Gordon's inequality to analyze the linear program from \Cref{eq:LP_def}; this idea has been useful in a number of related contexts (e.g. \cite{stojnic2013another,miolane2021distribution}), including in analyzing the objective value of the same linear program over the unit sphere \cite{montanari2024tractability}.
    Gordon's inequality will allow us to replace the dependence on $X$ with a simpler dependence on $1$-dimensional Gaussian vectors, which will allow us to show that the objective value of the LP is asymptotically equivalent to a very well-concentrated random variable.
    Then, we again use Gordon's inequality to analyze the objective value of the LP over the restricted set of solutions not satisfying \Cref{eq:tight_fraction,eq:margin_conv}, and notice that the objective value over this set drops considerably.
    As a consequence, we conclude that the optimal solution must satisfy \Cref{eq:tight_fraction,eq:margin_conv} with high probability. 
    
	For notational convenience, we recast $(\rho(\alpha, \kappa), \gamma(\alpha, \kappa), t(\alpha, \kappa), P(\alpha, \kappa))$ as $(\rho_*, \gamma_*, t_*, P_*)$.
	\vspace{0.5em}
	
	\noindent {\bf Part 1. Proof of \cref{eq:tight_fraction}.}
	To begin with, note that the linear program~\eqref{eq:LP_def} can be rewritten as
	\begin{equation}\label{eq:LP1_primal}
		\mbox{maximize} \ v^\top \theta, \ \mbox{subject to} \ \frac{1}{\sqrt{N}} X \theta \ge \kappa 1, \ \theta \ge -1, \ 1 \ge \theta.
	\end{equation}
	Since its domain is compact, the above problem always has a solution.
By duality, we know that the following dual problem also has a solution and their values are equal:
	\begin{equation}\label{eq:LP1_dual}
		\mbox{minimize} \ \left( \alpha + \beta \right)^\top 1 - \kappa \lambda^\top 1, \ \mbox{subject to} \ \frac{1}{\sqrt{N}} X^\top \lambda = \beta - \alpha - v, \ \lambda, \alpha, \beta \ge 0.
	\end{equation}
	Denote by $(\hat{\lambda}, \hat{\alpha}, \hat{\beta})$ the solution to \cref{eq:LP1_dual}.
By complementary slackness, we know that
	\begin{equation}
		\hat{\theta}_i \in \{ \pm 1 \} \Longleftrightarrow \hat{\alpha}_i > 0 \ \mbox{or} \ \hat{\beta}_i > 0 \Longleftrightarrow \hat{\alpha}_i + \hat{\beta}_i > 0.
	\end{equation}
	Therefore, it suffices to show that for all $\veps > 0$:
	\begin{equation}
		\lim_{M/N \to \alpha} \P \left( \left\vert \left\{ i \in [N]: \hat{\alpha}_i + \hat{\beta}_i > 0 \right\} \right\vert \ge \left( 2 \Phi \left( -t_* \right) - \veps \right) N \right) = 1.
	\end{equation}
	To this end, define the set
	\begin{equation}
		N_0^{\veps} = \left\{ z \in \R^N: z \ge 0, \left\vert \left\{ i \in [N]: z_i > 0 \right\} \right\vert \ge \left( 2 \Phi \left( -t_* \right) - \veps \right) N \right\}.
	\end{equation}
	We will show that
	\begin{equation}
		\lim_{M/N \to \alpha} \P \left( \hat{\alpha} + \hat{\beta} \in N_0^{\veps} \right) = 1.
	\end{equation}
	Let us consider the following two optimization problems with the same objective function but different constraints:
	\begin{align}
		\xi_{M,N} = \, & \min \left\{ \left( \alpha + \beta \right)^\top 1 - \kappa \lambda^\top 1 \right\}, \ \mbox{s.t.} \ \frac{1}{\sqrt{N}} X^\top \lambda = \beta - \alpha - v, \ \lambda, \alpha, \beta \ge 0, \\
		\xi_{M,N}^{\veps} = \, & \min \left\{ \left( \alpha + \beta \right)^\top 1 - \kappa \lambda^\top 1 \right\}, \ \mbox{s.t.} \ \frac{1}{\sqrt{N}} X^\top \lambda = \beta - \alpha - v, \ \lambda, \alpha, \beta \ge 0, \ \alpha + \beta \notin N_0^{\veps}.
	\end{align}
	Since $\xi_{M,N}^{\veps} > \xi_{M,N}$ implies $\hat{\alpha} + \hat{\beta} \in N_0^{\veps}$, it suffices to show that
	\begin{equation}\label{eq:W2_conv_1}
		\lim_{M/N \to \alpha} \P \left( \xi_{M,N}^{\veps} > \xi_{M,N} \right) = 1.
	\end{equation}
	To establish \cref{eq:W2_conv_1}, we proceed with the following two steps:
	\vspace{0.5em}
	
	\noindent {\bf Step 1. Analysis of $\xi_{M, N}$.} By strong duality, we know that $\xi_{M, N}$ is just the value of the linear program \cref{eq:LP_def}, which can be recast into the following max-min (or min-max) form:
	\begin{align}
		\xi_{M, N} = \, & \max_{\norm{\theta}_2 \le \sqrt{N}} \min_{\lambda, \alpha, \beta \ge 0} \left\{ v^\top \theta + \lambda^\top \left( \frac{1}{\sqrt{N}} X \theta - \kappa 1 \right) + \alpha^\top \left( \theta + 1 \right) + \beta^\top \left( 1 - \theta \right) \right\} \\
		\stackrel{(i)}{=} \, & \min_{\lambda, \alpha, \beta \ge 0} \max_{\norm{\theta}_2 \le \sqrt{N}} \left\{ v^\top \theta + \lambda^\top \left( \frac{1}{\sqrt{N}} X \theta - \kappa 1 \right) + \alpha^\top \left( \theta + 1 \right) + \beta^\top \left( 1 - \theta \right) \right\},
	\end{align}
	where $(i)$ follows from Sion's minimax theorem.
Conditioning on $v$ (independent of $X$) and applying \cref{lem:gordon}, we deduce that
	\begin{equation}\label{eq:gordon_res_1}
		\P \left( \xi_{M, N} \le t \right) \le 2 \P \left( \xi_{M, N}^{(1)} \le t \right), \quad \P \left( \xi_{M, N} \ge t \right) \le 2 \P \left( \xi_{M, N}^{(1)} \ge t \right), \quad \forall t \in \R,
	\end{equation}
	where
	\begin{align}
		\xi_{M, N}^{(1)} = \, & \min_{\lambda, \alpha, \beta \ge 0} \max_{\norm{\theta}_2 \le \sqrt{N}} \left\{ v^\top \theta + \frac{1}{\sqrt{N}} \norm{\lambda}_2 \theta^\top g + \frac{1}{\sqrt{N}} \norm{\theta}_2 \lambda^\top h - \kappa \lambda^\top 1 + \alpha^\top \left( \theta + 1 \right) + \beta^\top \left( 1 - \theta \right) \right\}, \\
		\mbox{and} \, \, & h \sim \normal(0, I_M), \ g \sim \normal(0, I_N), \ (g, h) \perp\!\!\!\perp v.
	\end{align}
	Next, we note that for any fixed $\rho \in [0, 1]$ with $\norm{\theta}_2 = \rho \sqrt{N}$, the inner maximum over $\theta$ equals
	\begin{equation*}
		\rho \sqrt{N} \norm{v + \frac{1}{\sqrt{N}} \norm{\lambda}_2 g + \alpha - \beta}_2
		 + \lambda^\top \left( \rho h - \kappa 1 \right) + \alpha^\top 1 + \beta^\top 1.
	\end{equation*}
	Similarly, given $\norm{\lambda}_2 = \gamma \sqrt{N}$ and $\lambda \ge 0$, the outer minimum always equals
	\begin{equation*}
		\rho \sqrt{N} \norm{v + \gamma g + \alpha - \beta}_2
		 - \sqrt{N} \gamma \norm{\left( \kappa 1 - \rho h \right)_+}_2 + (\alpha + \beta)^\top 1.
	\end{equation*}
	As a consequence, we obtain that
	\begin{align}
		\xi_{M, N}^{(1)} = \, & \min_{\substack{ \gamma \ge 0 \\ \alpha, \beta \ge 0}} \max_{\rho \in [0, 1]} \left\{ \rho \sqrt{N} \norm{v + \gamma g + \alpha - \beta}_2
		 - \sqrt{N} \gamma \norm{\left( \kappa 1 - \rho h \right)_+}_2 + (\alpha + \beta)^\top 1 \right\} \\
		 \stackrel{d}{=} \, & N \min_{\substack{ \gamma \ge 0 \\ \alpha, \beta \ge 0}} \max_{\rho \in [0, 1]} \left\{ \frac{\rho}{\sqrt{N}} \norm{ \sqrt{1 + \gamma^2} g + \alpha - \beta}_2
		 - \frac{\gamma \sqrt{\alpha_{M, N}}}{\sqrt{M}} \norm{\left( \kappa 1 - \rho h \right)_+}_2 + \frac{1}{N}\sum_{i=1}^{N}(\alpha_i + \beta_i) \right\} \\
		 = \, & N \max_{\rho \in [0, 1]} \min_{\substack{ \gamma \ge 0 \\ \alpha, \beta \ge 0}} \left\{ \frac{\rho}{\sqrt{N}} \norm{ \sqrt{1 + \gamma^2} g + \alpha - \beta}_2
		 - \frac{\gamma \sqrt{\alpha_{M, N}}}{\sqrt{M}} \norm{\left( \kappa 1 - \rho h \right)_+}_2 + \frac{1}{N}\sum_{i=1}^{N}(\alpha_i + \beta_i) \right\},
	\end{align}
	where $\alpha_{M, N} = M/N$, and the last equality again follows from Sion's minimax theorem: one can easily verify that the objective function is convex-concave.
To further simplify the above expression for $\xi_{M, N}^{(1)}$, we note that for any $B \ge 0$, $B \in \R^N$,
	\begin{equation*}
		\min_{\substack{\alpha, \beta \ge 0 \\ \alpha + \beta = B}} \norm{ \sqrt{1 + \gamma^2} g + \alpha - \beta}_2 = \norm{\left( \sqrt{1 + \gamma^2} |g| - B \right)_+}_2,
	\end{equation*}
	thus leading to
	\begin{align}
		\frac{\xi_{M, N}^{(1)}}{N} = \, & \max_{\rho \in [0, 1]} \min_{\substack{ \gamma \ge 0 \\ B \ge 0}} \left\{ \frac{\rho}{\sqrt{N}} \norm{\left( \sqrt{1 + \gamma^2} |g| - B \right)_+}_2
		 - \frac{\gamma \sqrt{\alpha_{M, N}}}{\sqrt{M}} \norm{\left( \kappa 1 - \rho h \right)_+}_2 + \frac{1}{N}\sum_{i=1}^{N} B_i \right\} \\
		 \stackrel{(i)}{=} \, & \max_{\rho \in [0, 1]} \min_{\substack{ \gamma \ge 0 \\ A \ge 0}} \left\{
		 - \frac{\gamma \sqrt{\alpha_{M, N}}}{\sqrt{M}} \norm{\left( \kappa 1 - \rho h \right)_+}_2 + \frac{\rho \sqrt{1 + \gamma^2}}{\sqrt{N}} \norm{\left( |g| - A \right)_+}_2 + \frac{\sqrt{1 + \gamma^2}}{N}\sum_{i=1}^{N} A_i \right\},
	\end{align}
	where in $(i)$ we make the change of variables $A = B / \sqrt{1 + \gamma^2}$.
    We can now apply some calculus to remove the dependence on $A$; the following lemma summarizes the outcome of the calculation.

\begin{lem}\torestate{\label{lem:calculus_1}
	Let $\rho \ge 0$, $y \in \R^N, y \ge 0$.
Then, we have
	\begin{align}
		\min_{A \ge 0} \left\{ \frac{\rho}{\sqrt{N}} \norm{\left( y - A \right)_+}_2 + \frac{1}{N}\sum_{i=1}^{N} A_i \right\} = \, \min_{t \ge 0} \left\{ \frac{\rho^2 t}{2} + \frac{1}{2 t N} \norm{\min \left( y, t \right)}_2^2 + \frac{1}{N} \sum_{i=1}^{N} \left( y_i - t \right)_+ \right\}.
	\end{align}}
\end{lem}

We will give the proof of \Cref{lem:calculus_1} in \Cref{sec:aux_lemmas}.
Applying \cref{lem:calculus_1} yields that
	\begin{align*}
		\frac{\xi_{M, N}^{(1)}}{N}
		 = \, & \max_{\rho \in [0, 1]} \min_{\gamma \ge 0} \left\{
		 - \frac{\gamma \sqrt{\alpha_{M, N}}}{\sqrt{M}} \norm{\left( \kappa 1 - \rho h \right)_+}_2 + \sqrt{1+\gamma^2} \varphi_N (\rho; g) \right\},
	\end{align*}
	where we define
    \begin{equation}
        \varphi_N (\rho; g) = \min_{t \ge 0} \left( \frac{\rho^2 t}{2} + \frac{1}{2 t N} \sum_{i=1}^{N} \min \left( |g_i|, t \right)^2 + \frac{1}{N} \sum_{i=1}^{N} (|g_i| - t)_+ \right),
    \end{equation}
    (see also \cref{lem:ULLN_1}).

    We now wish to simplify the inner minimization; here, we re-use the calculation appearing in \cite[Lemma 3]{montanari2024tractability}:

\begin{lem}\label{lem:calculus_2}
Suppose $a, b \geq 0$.
Then we have
$$
\begin{aligned}
\inf _{x \geq 0}\left\{a \sqrt{1+x^2}-b x\right\}= \, & \begin{cases}\sqrt{a^2-b^2}, & \text { if } a \geq b \\
-\infty, & \text { if } a<b\end{cases}, \\
\underset{x \geq 0}{\operatorname{argmin}}\left\{a \sqrt{1+x^2}-b x\right\}= \, & \begin{cases}b / \sqrt{a^2-b^2}, & \text { if } a>b \\
\infty, & \text { if } a \leq b\end{cases}.
\end{aligned}
$$
\end{lem}
    
	According to \cref{lem:calculus_2}, the inner minimum (infimum) is finite if any only if
	\begin{equation*}
		\rho \in E_{M, N} (\rho) = \left\{ \rho \in [0, 1]: \varphi_N (\rho; g)^2 \ge \frac{\alpha_{M, N}}{M} \norm{\left( \kappa 1 - \rho h \right)_+}_2^2 \right\}.
	\end{equation*}
	By \cref{lem:ULLN_1}, \cref{lem:ULLN_2}, and our assumption $\alpha_{M, N} \to \alpha$ where
	\begin{equation*}
		\operatorname{int} E(\rho) \neq \varnothing, \quad E(\rho) = \left\{ \rho \in [0, 1]: \varphi (\rho)^2 \ge \alpha \E \left[ \left( \kappa - \rho G \right)_+^2 \right] \right\},
	\end{equation*}
	we know that with high probability, $\operatorname{int} E_{M, N} (\rho) \neq \varnothing$.
Consequently,
	\begin{align*}
		\frac{\xi_{M, N}^{(1)}}{N} = \, & \max_{\rho \in E_{M, N} (\rho)} \left\{ \sqrt{\varphi_N (\rho; g)^2 - \frac{\alpha_{M, N}}{M} \norm{\left( \kappa 1 - \rho h \right)_+}_2^2} \right\} \\
		= \, & \max_{\rho \in [0, 1]} \left\{ \sqrt{ \left( \varphi_N (\rho; g)^2 - \frac{\alpha_{M, N}}{M} \norm{\left( \kappa 1 - \rho h \right)_+}_2^2 \right)_+} \right\} \\
		\stackrel{p}{\to} \, & \max_{\rho \in [0, 1]} \left\{ \sqrt{ \left( \varphi (\rho)^2 - \alpha \E \left[ \left( \kappa - \rho G \right)_+^2 \right] \right)_+} \right\} \\
		= \, & \max_{\rho \in E(\rho)} \left\{ \sqrt{ \varphi (\rho)^2 - \alpha \E \left[ \left( \kappa - \rho G \right)_+^2 \right] } \right\} \\
		= \, & \max_{\rho \in [0, 1]} \min_{\gamma \ge 0} \left\{ - \gamma \sqrt{\alpha} \sqrt{\E \left[ \left( \kappa - \rho G \right)_+^2 \right]} + \sqrt{1 + \gamma^2} \varphi(\rho) \right\} \\
        = \, & \sqrt{ \varphi \left( \rho_* \right)^2 - \alpha \E \left[ \left( \kappa - \rho_* G \right)_+^2 \right] }
	\end{align*}
	as $M, N \to \infty$ with $M/N \to \alpha$.
Combining this fact with \cref{eq:gordon_res_1}, we deduce that
	\begin{equation}\label{eq:xi_lower_bd}
	\begin{split}
		\plim_{M/N \to \alpha} \frac{\xi_{M, N}}{N} = \plim_{M/N \to \alpha} \frac{\xi_{M, N}^{(1)}}{N} = \, & \max_{\rho \in [0, 1]} \min_{\gamma \ge 0} \left\{ - \gamma \sqrt{\alpha} \sqrt{\E \left[ \left( \kappa - \rho G \right)_+^2 \right]} + \sqrt{1 + \gamma^2} \varphi(\rho) \right\} \\
		= \, & \sqrt{ \varphi \left( \rho_* \right)^2 - \alpha \E \left[ \left( \kappa - \rho_* G \right)_+^2 \right] }.
	\end{split}
	\end{equation}
	\vspace{0.5em}
	
	\noindent{\bf Step 2. Analysis of $\xi_{M,N}^{\veps}$.} Similar to the previous step, we have
	\begin{align}
		\xi_{M, N}^{\veps} \ge \, \min_{\substack{\lambda, \alpha, \beta \ge 0 \\ \alpha + \beta \notin N_0^{\veps}}} \max_{\norm{\theta}_2 \le \sqrt{N}} \left\{ v^\top \theta + \lambda^\top \left( \frac{1}{\sqrt{N}} X \theta - \kappa 1 \right) + \alpha^\top \left( \theta + 1 \right) + \beta^\top \left( 1 - \theta \right) \right\}.
	\end{align}
	Define for $h \sim \normal(0, I_M)$, $g \sim \normal(0, I_N)$, $(g, h) \perp\!\!\!\perp v$:
	\begin{equation}
		\xi_{M,N}^{\veps, (1)} = \min_{\substack{\lambda, \alpha, \beta \ge 0 \\ \alpha + \beta \notin N_0^{\veps}}} \max_{\norm{\theta}_2 \le \sqrt{N}} \left\{ v^\top \theta + \frac{1}{\sqrt{N}} \norm{\lambda}_2 \theta^\top g + \frac{1}{\sqrt{N}} \norm{\theta}_2 \lambda^\top h - \kappa \lambda^\top 1 + \alpha^\top \left( \theta + 1 \right) + \beta^\top \left( 1 - \theta \right) \right\}.
	\end{equation}
	Then, we deduce similarly from \cref{lem:gordon} $(i)$ that
	\begin{equation*}
		\P \left( \xi_{M, N}^{\veps} \le t \right) \le 2 \P \left( \xi_{M,N}^{\veps, (1)} \le t \right), \quad \forall t \in \R.
	\end{equation*}
	It then suffices to obtain a lower bound for $\xi_{M,N}^{\veps, (1)}$.
Following a similar argument to the previous step, one gets that
	\begin{align*}
		\xi_{M,N}^{\veps, (1)} \ge \, & \min_{\substack{ \gamma \ge 0 \\ \alpha, \beta \ge 0 \\ \alpha + \beta \notin N_0^{\veps}}} \max_{\rho \in [0, 1]} \left\{ \rho \sqrt{N} \norm{v + \gamma g + \alpha - \beta}_2
		 - \sqrt{N} \gamma \norm{\left( \kappa 1 - \rho h \right)_+}_2 + (\alpha + \beta)^\top 1 \right\} \\
		 \stackrel{d}{=} \, & N \min_{\substack{ \gamma \ge 0 \\ \alpha, \beta \ge 0 \\ \alpha + \beta \notin N_0^{\veps}}} \max_{\rho \in [0, 1]} \left\{ \frac{\rho}{\sqrt{N}} \norm{ \sqrt{1 + \gamma^2} g + \alpha - \beta}_2
		 - \frac{\gamma \sqrt{\alpha_{M, N}}}{\sqrt{M}} \norm{\left( \kappa 1 - \rho h \right)_+}_2 + \frac{1}{N}\sum_{i=1}^{N}(\alpha_i + \beta_i) \right\} \\
		 \ge \, & N \max_{\rho \in [0, 1]} \min_{\substack{ \gamma \ge 0 \\ \alpha, \beta \ge 0 \\ \alpha + \beta \notin N_0^{\veps}}} \left\{ \frac{\rho}{\sqrt{N}} \norm{ \sqrt{1 + \gamma^2} g + \alpha - \beta}_2
		 - \frac{\gamma \sqrt{\alpha_{M, N}}}{\sqrt{M}} \norm{\left( \kappa 1 - \rho h \right)_+}_2 + \frac{1}{N}\sum_{i=1}^{N}(\alpha_i + \beta_i) \right\} \\
		 = \, & N \max_{\rho \in [0, 1]} \min_{\substack{ \gamma \ge 0 \\ B \ge 0, B \notin N_0^{\veps}}} \left\{ \frac{\rho}{\sqrt{N}} \norm{\left( \sqrt{1 + \gamma^2} |g| - B \right)_+}_2
		 - \frac{\gamma \sqrt{\alpha_{M, N}}}{\sqrt{M}} \norm{\left( \kappa 1 - \rho h \right)_+}_2 + \frac{1}{N}\sum_{i=1}^{N} B_i \right\} \\
		 = \, & N \max_{\rho \in [0, 1]} \min_{\substack{ \gamma \ge 0 \\ A \ge 0, A \notin N_0^{\veps} }} \left\{
		 - \frac{\gamma \sqrt{\alpha_{M, N}}}{\sqrt{M}} \norm{\left( \kappa 1 - \rho h \right)_+}_2 + \frac{\rho \sqrt{1 + \gamma^2}}{\sqrt{N}} \norm{\left( |g| - A \right)_+}_2 + \frac{\sqrt{1 + \gamma^2}}{N}\sum_{i=1}^{N} A_i \right\}.
	\end{align*}
	This calculation implies that
	\begin{equation}
		\frac{\xi_{M,N}^{\veps, (1)}}{N} \ge \max_{\rho \in [0, 1]} \min_{\gamma \ge 0} \left\{
		 - \frac{\gamma \sqrt{\alpha_{M, N}}}{\sqrt{M}} \norm{\left( \kappa 1 - \rho h \right)_+}_2 + \sqrt{1 + \gamma^2} \min_{A \ge 0, A \notin N_0^{\veps}} \left\{
		 \frac{\rho}{\sqrt{N}} \norm{\left( |g| - A \right)_+}_2 + \frac{1}{N}\sum_{i=1}^{N} A_i \right\} \right\}.
	\end{equation}
	For future convenience, we define
	\begin{equation}
		\varphi_{N, \veps} (\rho; g) = \,\min_{A \ge 0, A \notin N_0^{\veps}} \left\{
		 \frac{\rho}{\sqrt{N}} \norm{\left( |g| - A \right)_+}_2 + \frac{1}{N}\sum_{i=1}^{N} A_i \right\}.
	\end{equation}
	By definition, we always have $\varphi_{N, \veps} (\rho; g) \ge \varphi_{N} (\rho; g)$, so there exists $\rho \in [0, 1]$ such that
	\begin{equation*}
		\varphi_{N, \veps} (\rho; g)^2 \ge \frac{\alpha_{M, N}}{M} \norm{\left( \kappa 1 - \rho h \right)_+}_2^2.
	\end{equation*}
	Using a similar argument as in the previous step, we deduce that
	\begin{align*}
		\frac{\xi_{M,N}^{\veps, (1)}}{N} \ge \, & \max_{\rho \in [0, 1]} \min_{\gamma \ge 0} \left\{
		 - \frac{\gamma \sqrt{\alpha_{M, N}}}{\sqrt{M}} \norm{\left( \kappa 1 - \rho h \right)_+}_2 + \sqrt{1 + \gamma^2} \varphi_{N, \veps} (\rho; g) \right\} \\
		 = \, & \max_{\rho \in [0, 1]} \left\{ \sqrt{ \left( \varphi_{N, \veps} (\rho; g)^2 - \frac{\alpha_{M, N}}{M} \norm{\left( \kappa 1 - \rho h \right)_+}_2^2 \right)_+} \right\} \\
		 \ge \, & \sqrt{ \varphi_{N, \veps} \left( \rho_*; g \right)^2 - \frac{\alpha_{M, N}}{M} \norm{\left( \kappa 1 - \rho_* h \right)_+}_2^2 }
	\end{align*} 
	with high probability.
In what follows we compute the limit of the right hand side of the above equation.
Using a similar argument as in the proof of \cref{lem:calculus_1}, we get
	\begin{align}
		\varphi_{N, \veps} (\rho; g) = \, & \min_{A \ge 0, A \notin N_0^{\veps}} \left\{
		 \frac{\rho}{\sqrt{N}} \norm{\left( |g| - A \right)_+}_2 + \frac{1}{N}\sum_{i=1}^{N} A_i \right\} \\
		 = \, & \min_{A \ge 0, A \notin N_0^{\veps}} \min_{t \ge 0} \left\{ \frac{\rho^2 t}{2} + \frac{1}{2 t N} \sum_{i=1}^{N} \left( \vert g_i \vert - A_i \right)_+^2 + \frac{1}{N}\sum_{i=1}^{N} A_i \right\} \\
		= \, & \min_{t \ge 0} \left\{ \frac{\rho^2 t}{2} + \min_{A \ge 0, A \notin N_0^{\veps}} \left\{ \frac{1}{2 t N} \sum_{i=1}^{N} \left( \vert g_i \vert - A_i \right)_+^2 + \frac{1}{N}\sum_{i=1}^{N} A_i \right\} \right\}.
	\end{align}
	Recall that 
	\begin{equation*}
		N_0^{\veps} = \left\{ A \in \R^N: A \ge 0, \left\vert \left\{ i \in [N]: A_i > 0 \right\} \right\vert \ge \left( 2 \Phi \left( -t_* \right) - \veps \right) N \right\}.
	\end{equation*}
	Denote $N(t_*, \veps) = \lceil ( 1 - 2 \Phi ( -t_* ) + \veps ) N \rceil$, then
	\begin{equation}
		A \notin N_0^{\veps} \Longleftrightarrow \left\vert \left\{ i \in [N]: A_i = 0 \right\} \right\vert \ge N(t_*, \veps).
	\end{equation}
	It is straightforward to see that subject to $A \ge 0, A \notin N_0^{\veps}$,  the minimum of
	\begin{equation*}
		\frac{1}{2 t N} \sum_{i=1}^{N} \left( \vert g_i \vert - A_i \right)_+^2 + \frac{1}{N}\sum_{i=1}^{N} A_i
	\end{equation*}
	must be achieved at some $A \ge 0$ satisfying
	\begin{equation}
		A_{(1)} = A_{(2)} = \cdots = A_{(N(t_*, \veps))} = 0,
	\end{equation}
	where $\vert g_{(1)} \vert \le \cdots \le \vert g_{(N)} \vert$ are the order statistics of $\{ \vert g_i \vert \}_{i \in [N]}$.
In other words, we must have $A_j = 0$ if $\vert g_j \vert$ is among the $N(t_*, \veps)$ smallest $\vert g_i \vert$'s.
This observation yields
	\begin{align}
		& \min_{A \ge 0, A \notin N_0^{\veps}} \left\{ \frac{1}{2 t N} \sum_{i=1}^{N} \left( \vert g_i \vert - A_i \right)_+^2 + \frac{1}{N}\sum_{i=1}^{N} A_i \right\} \\
		= \, & \frac{1}{2 t N} \sum_{i=1}^{N(t_*, \veps)} \left\vert g_{(i)} \right\vert^2 + \min_{A \ge 0} \left\{ \frac{1}{2 t N} \sum_{i=N(t_*, \veps)+1}^{N} \left( \left\vert g_{(i)} \right\vert - A_i \right)_+^2 + \frac{1}{N}\sum_{i=N(t_*, \veps) + 1}^{N} A_i \right\} \\
		\stackrel{(i)}{=} \, & \frac{1}{2 t N} \sum_{i=1}^{N(t_*, \veps)} \left\vert g_{(i)} \right\vert^2 + \frac{1}{2 t N} \sum_{i=N(t_*, \veps)+1}^{N} \min \left( \left\vert g_{(i)} \right\vert, t \right)^2 + \frac{1}{N}\sum_{i=N(t_*, \veps) + 1}^{N} \left( \left\vert g_{(i)} \right\vert - t \right)_+,
	\end{align}
	where in $(i)$ we use \cref{eq:varphi_minimizer,eq:varphi_minimum}.
It finally follows that
	\begin{align*}
		\varphi_{N, \veps} (\rho; g) = \, & \min_{t \ge 0} \left\{ \frac{\rho^2 t}{2} + \frac{1}{2 t N} \sum_{i=1}^{N(t_*, \veps)} \left\vert g_{(i)} \right\vert^2 + \frac{1}{2 t N} \sum_{i=N(t_*, \veps)+1}^{N} \min \left( \left\vert g_{(i)} \right\vert, t \right)^2 + \frac{1}{N}\sum_{i=N(t_*, \veps) + 1}^{N} \left( \left\vert g_{(i)} \right\vert - t \right)_+ \right\}.
	\end{align*}
	By definition, we have $\varphi(0) = 0$, which implies $\rho_* > 0$.
Using a similar argument as in the proof of \cref{lem:ULLN_1}, we can show that
	\begin{equation*}
		\plim_{N \to \infty} \varphi_{N, \veps} (\rho_*; g) = \varphi_{\veps} (\rho_*),
	\end{equation*}
	where
	\begin{align}
		\varphi_{\veps} (\rho_*) = \min_{t \ge 0} \left\{ \frac{\rho_*^2 t}{2} + \frac{1}{2 t} \E \left[ G^2 1_{ \vert G \vert \le q(t_*, \veps)} \right] + \frac{1}{2 t} \E \left[ \min \left( \vert G \vert, t \right)^2 1_{ \vert G \vert > q(t_*, \veps)} \right] + \E \left[ \left( \vert G \vert - t \right)_+ 1_{ \vert G \vert > q(t_*, \veps)} \right] \right\},
	\end{align}
	and $q(t_*, \veps)$ satisfies
	\begin{equation*}
		1 - 2 \Phi(-q(t_*, \veps)) = \lim_{N \to \infty} \frac{N(t_*, \veps)}{N} = 1 - 2 \Phi(-t_*) + \veps \implies q(t_*, \veps) > t_*.
	\end{equation*}
	According to \cref{lem:ULLN_2}, we have
	\begin{equation*}
		\plim_{M \to \infty} \frac{\alpha_{M, N}}{M} \norm{\left( \kappa 1 - \rho_* h \right)_+}_2^2 = \alpha \E \left[ \left( \kappa - \rho_* G \right)_+^2 \right].
	\end{equation*}
	Therefore,
	\begin{equation}
		\plim_{M/N \to \alpha} \frac{\xi_{M,N}^{\veps, (1)}}{N} \ge \, \sqrt{\varphi_{\veps} (\rho_*)^2 - \alpha \E \left[ \left( \kappa - \rho_* G \right)_+^2 \right]}.
	\end{equation}
	Since $\P ( \xi_{M, N}^{\veps} \le t ) \le 2 \P ( \xi_{M,N}^{\veps, (1)} \le t )$ for all $t \in \R$, we finally deduce that
	\begin{equation}\label{eq:xi_delta_lower_bd}
		\plim_{M/N \to \alpha} \frac{\xi_{M, N}^{\veps}}{N} \ge \, \sqrt{\varphi_{\veps} (\rho_*)^2 - \alpha \E \left[ \left( \kappa - \rho_* G \right)_+^2 \right]}.
	\end{equation}
	\vspace{0.5em}
	
	\noindent {\bf Step 3. Prove \cref{eq:W2_conv_1}.} Combining \cref{eq:xi_lower_bd,eq:xi_delta_lower_bd}, it suffices to show that $\varphi_{\veps} (\rho_*) > \varphi (\rho_*)$.
To this end, we first recall their definitions:
	\begin{align*}
		\varphi (\rho_*) = \, & \min_{t \ge 0} \left\{ \frac{\rho_*^2 t}{2} + \frac{1}{2 t} \E \left[ \min \left( \vert G \vert, t \right)^2 \right] + \E \left[ \left( \vert G \vert - t \right)_+ \right] \right\} \\
        = \, & \frac{\rho_*^2 t_*}{2} + \frac{1}{2 t_*} \E \left[ \min \left( \vert G \vert, t_* \right)^2 \right] + \E \left[ \left( \vert G \vert - t_* \right)_+ \right], \\
		\varphi_{\veps} (\rho_*) = \, & \min_{t \ge 0} \left\{ \frac{\rho_*^2 t}{2} + \frac{1}{2 t} \E \left[ G^2 1_{ \vert G \vert \le q(t_*, \veps)} \right] + \frac{1}{2 t} \E \left[ \min \left( \vert G \vert, t \right)^2 1_{ \vert G \vert > q(t_*, \veps)} \right] + \E \left[ \left( \vert G \vert - t \right)_+ 1_{ \vert G \vert > q(t_*, \veps)} \right] \right\},
	\end{align*}
	where $q(t_{*}, \veps) > t_*$.
Note that for any $t \ge 0$,
	\begin{align}
		& \frac{\rho_*^2 t}{2} + \frac{1}{2 t} \E \left[ G^2 1_{ \vert G \vert \le q(t_*, \veps)} \right] + \frac{1}{2 t} \E \left[ \min \left( \vert G \vert, t \right)^2 1_{ \vert G \vert > q(t_*, \veps)} \right] + \E \left[ \left( \vert G \vert - t \right)_+ 1_{ \vert G \vert > q(t_*, \veps)} \right] \\
		& - \left( \frac{\rho_*^2 t}{2} + \frac{1}{2 t} \E \left[ \min \left( \vert G \vert, t \right)^2 \right] + \E \left[ \left( \vert G \vert - t \right)_+ \right] \right) \\
		= \, & \frac{1}{2 t} \left( \E \left[ G^2 1_{ \vert G \vert \le q(t_*, \veps)} \right] - \E \left[ \min \left( \vert G \vert, t \right)^2 1_{ \vert G \vert \le q(t_*, \veps)} \right] \right) - \E \left[ \left( \vert G \vert - t \right)_+ 1_{ \vert G \vert \le q(t_*, \veps)} \right].
	\end{align}
	If $t \ge q(t_*, \veps)$, the last line above equals $0$.
Otherwise, it equals
	\begin{align*}
		\frac{1}{2 t} \E \left[ (G^2 - t^2) 1_{t \le \vert G \vert \le q(t_*, \veps)} \right] - \E \left[ \left( \vert G \vert - t \right) 1_{t \le \vert G \vert \le q(t_*, \veps)} \right]
		= \, \frac{1}{2 t} \E \left[ (\vert G \vert - t)^2 1_{t \le \vert G \vert \le q(t_*, \veps)} \right],
	\end{align*}
	which is strictly positive since $t < q(t_*, \veps)$.
Now, since the optimization problem that defines $\varphi(\rho_*)$ is uniquely minimized at $t_* < q(t_*, \veps)$, and its value diverges to $+ \infty$ as $t \to + \infty$, we finally deduce that $\varphi_{\veps} (\rho_*) > \varphi (\rho_*)$.
This proves \cref{eq:W2_conv_1}, and in turn completes the proof of \cref{eq:tight_fraction}.
	
	\vspace{0.5em}
	\noindent {\bf Part 2. Proof of \cref{eq:margin_conv}.}
	Note that the original linear program \cref{eq:LP_def} can be recast as
	\begin{equation}\label{eq:LP2_primal}
		\mbox{maximize} \ v^\top \theta, \ \mbox{subject to} \ u \ge \kappa 1, \ \frac{1}{\sqrt{N}} X \theta = u, \ \theta \ge -1, \ 1 \ge \theta.
	\end{equation}
	Since $X_{ij} \iidsim \normal(0, 1)$, we know that there exists a constant $C_{\alpha}$ that only depends on $\alpha$, such that $\P (\norm{X}_{\op} > C_{\alpha} \sqrt{N}) \to 0$ as $M, N \to \infty$, $M / N \to \alpha$.
Denote by $(\hat{\theta}, \hat{u})$ the solution to \cref{eq:LP2_primal}, then we always have $\norm{\hat{\theta}}_2 \le \sqrt{N}$, and with high probability $\norm{\hat{u}}_2 \le C_{\alpha} \sqrt{N}$.
Recall $\xi_{M, N}$ from the first part of the proof, then we know that
	\begin{equation}
		\xi_{M, N} = \max_{\substack{\norm{\theta}_2 \le \sqrt{N} \\ \norm{u}_2 \le C_{\alpha} \sqrt{N}} } \min_{\substack{\lambda, \alpha, \beta \ge 0 \\ \mu \in \R^M}} \left\{ v^\top \theta + \lambda^\top \left( u - \kappa 1 \right) + \mu^\top \left( \frac{1}{\sqrt{N}} X \theta - u \right) + \alpha^\top \left( \theta + 1 \right) + \beta^\top \left( 1 - \theta \right) \right\}
	\end{equation}
	with high probability.
Now, for any $\veps > 0$, define
	\begin{equation}
		M_2^{\veps} = \left\{ z \in \R^M: W_2 \left( \frac{1}{M} \sum_{i=1}^{M} \delta_{z_i}, P_* \right) \le \veps \right\},
	\end{equation}
	it suffices to show that $\hat{u} \in M_2^{\veps}$ with high probability.
To this end, define
	\begin{equation}
		\eta_{M, N}^{\veps} = \max_{\substack{\norm{\theta}_2 \le \sqrt{N} \\ \norm{u}_2 \le C_{\alpha} \sqrt{N}} } \left\{ v^\top \theta \right\}, \ \mbox{s.t.} \ u \ge \kappa 1, \ u \notin M_2^{\veps}, \ \frac{1}{\sqrt{N}} X \theta = u, \ \theta \ge -1, \ 1 \ge \theta.
	\end{equation}
	Then, it remains to show that
	\begin{equation}\label{eq:W2_conv_2}
		\lim_{M/N \to \alpha} \P \left( \eta_{M, N}^{\veps} < \xi_{M, N} \right) = 1.
	\end{equation}
	From the first part of the proof, we already know that
	\begin{equation}
		\plim_{M/N \to \alpha} \frac{\xi_{M, N}}{N}
		= \, \sqrt{ \varphi \left( \rho_* \right)^2 - \alpha \E \left[ \left( \kappa - \rho_* G \right)_+^2 \right] }.
	\end{equation}
	In what follows we will derive an upper bound for $\eta_{M, N}^{\veps}$.
First, note that
	\begin{align*}
		\eta_{M, N}^{\veps} = \, & \max_{\substack{\norm{\theta}_2 \le \sqrt{N} \\ \norm{u}_2 \le C_{\alpha} \sqrt{N}, u \notin M_2^{\veps}} } \min_{\substack{\lambda, \alpha, \beta \ge 0 \\ \mu \in \R^M}} \left\{ v^\top \theta + \lambda^\top \left( u - \kappa 1 \right) + \mu^\top \left( \frac{1}{\sqrt{N}} X \theta - u \right) + \alpha^\top \left( \theta + 1 \right) + \beta^\top \left( 1 - \theta \right) \right\} \\
		\le \, & \max_{\substack{\norm{\theta}_2 \le \sqrt{N} \\ \norm{u}_2 \le C_{\alpha} \sqrt{N}, u \notin M_2^{\veps}} } \min_{\substack{(\lambda, \alpha, \beta, \mu) \in B_C \\ \lambda, \alpha, \beta \ge 0}} \left\{ v^\top \theta + \lambda^\top \left( u - \kappa 1 \right) + \mu^\top \left( \frac{1}{\sqrt{N}} X \theta - u \right) + \alpha^\top \left( \theta + 1 \right) + \beta^\top \left( 1 - \theta \right) \right\} \\
		:= \, & \eta_{M, N}^{\veps} (C), \quad \forall C = (C_1, C_2, C_3) \in \R_{> 0}^3,
	\end{align*}
	where 
 	\begin{equation*}
		B_C = \left\{ (\lambda, \alpha, \beta, \mu): \norm{\lambda}_2 \le C_1 \sqrt{N}, \norm{\mu}_2 \le C_2 \sqrt{N}, \norm{\alpha + \beta}_2 \le C_3 \sqrt{N} \right\}.	
	\end{equation*}	
	Now we define for $h \sim \normal(0, I_M)$, $g \sim \normal(0, I_N)$, $(g, h) \perp\!\!\!\perp v$:
	\begin{align}
		\eta_{M,N}^{\veps, (1)} (C) = \, & \max_{\substack{\norm{\theta}_2 \le \sqrt{N} \\ \norm{u}_2 \le C_{\alpha} \sqrt{N}, u \notin M_2^{\veps}} } \min_{\substack{(\lambda, \alpha, \beta, \mu) \in B_C \\ \lambda, \alpha, \beta \ge 0}} \bigg\{ v^\top \theta + \lambda^\top \left( u - \kappa 1 \right) + \frac{1}{\sqrt{N}} \norm{\mu}_2 \theta^\top g + \frac{1}{\sqrt{N}} \norm{\theta}_2 \mu^\top h \\
		& \quad \quad \quad \quad \quad \quad \quad \quad \quad \quad \quad \quad \quad \,\, - \mu^\top u + \alpha^\top \left( \theta + 1 \right) + \beta^\top \left( 1 - \theta \right) \bigg\}.
	\end{align}
	Then, applying \cref{lem:gordon} $(i)$ to $- \eta_{M, N}^{\veps} (C)$ and $- \eta_{M,N}^{\veps, (1)} (C)$, we obtain that
	\begin{equation*}
		\P \left( \eta_{M, N}^{\veps} \ge t \right) \le \P \left( \eta_{M, N}^{\veps} (C) \ge t \right) \le 2 \P \left( \eta_{M,N}^{\veps, (1)} (C) \ge t \right), \quad \forall t \in \R.
	\end{equation*}
	The rest of the proof will be devoted to the analysis of $\eta_{M,N}^{\veps, (1)} (C)$.
Similar to part 1, we obtain the following simplified expression for $\eta_{M,N}^{\veps, (1)} (C)$:
	\begin{align*}
		\eta_{M,N}^{\veps, (1)} (C) \stackrel{d}{=} \, & N \max_{\substack{\rho \in [0, 1] \\ \norm{u}_2 \le C_{\alpha} \sqrt{N}, u \notin M_2^{\veps}} } \min_{\substack{\norm{B}_2 \le C_3 \sqrt{N}, B \ge 0 \\ \gamma \in [0, C_2]}} \bigg\{ \frac{\rho}{\sqrt{N}} \norm{ \left( \sqrt{1+\gamma^2} \vert g \vert - B \right)_+}_2 - \frac{C_1}{\sqrt{N}} \norm{(\kappa 1 - u)_+}_2 \\
		& \quad \quad \quad \quad \quad \quad \quad \quad \quad \quad \quad \quad \quad \quad \quad \,\, - \frac{\gamma}{\sqrt{N}} \norm{u - \rho h}_2 + \frac{1}{N} \sum_{i=1}^{N} B_i \bigg\}.
	\end{align*}
	Fixing $C_2$ and $C_3$ and letting $C_1 \to \infty$, we obtain that
	\begin{align*}
		\limsup_{C_1 \to \infty} \eta_{M,N}^{\veps, (1)} (C) \le \, & N \max_{\substack{\rho \in [0, 1] \\ u \ge \kappa 1, u \notin M_2^{\veps}} } \min_{\substack{\norm{B}_2 \le C_3 \sqrt{N}, B \ge 0 \\ \gamma \in [0, C_2]}} \bigg\{ \frac{\rho}{\sqrt{N}} \norm{ \left( \sqrt{1+\gamma^2} \vert g \vert - B \right)_+}_2 \\
		& \quad \quad \quad \quad \quad \quad \quad \quad \quad \quad \quad \quad \quad - \frac{\gamma}{\sqrt{N}} \norm{u - \rho h}_2 + \frac{1}{N} \sum_{i=1}^{N} B_i \bigg\} \\
		\le \, & N \max_{\substack{\rho \in [0, 1] \\ u \ge \kappa 1, u \notin M_2^{\veps}} } \min_{\substack{\norm{A}_2 \le \sqrt{N}C_3 / \sqrt{1 + C_2^2}, A \ge 0 \\ \gamma \in [0, C_2]}} \bigg\{ \frac{\rho \sqrt{1 + \gamma^2}}{\sqrt{N}} \norm{ \left( \vert g \vert - A \right)_+}_2 \\
		& \quad \quad \quad \quad \quad \quad \quad \quad \quad \quad \quad \quad \quad \quad \quad \quad - \frac{\gamma}{\sqrt{N}} \norm{u - \rho h}_2 + \frac{\sqrt{1 + \gamma^2}}{N} \sum_{i=1}^{N} A_i \bigg\}.
	\end{align*}
	According to the proof of \cref{lem:calculus_1}, the inner minimum is achieved at some $A$ such that
	\begin{equation}
		A_i \le \vert g_i \vert , \ \forall i \in [N] \implies \norm{A}_2 \le \norm{g}_2 \le 2 \sqrt{N} \ \mbox{with high probability}.
	\end{equation}
	As a consequence, we can send $C_3 \to \infty$ so that
	\begin{align*}
		\limsup_{C_3 \to \infty} \limsup_{C_1 \to \infty} \eta_{M,N}^{\veps, (1)} (C) \le \, & N \max_{\substack{\rho \in [0, 1] \\ u \ge \kappa 1, u \notin M_2^{\veps}} } \min_{\substack{ A \ge 0 \\ \gamma \in [0, C_2]}} \bigg\{ \frac{\rho \sqrt{1 + \gamma^2}}{\sqrt{N}} \norm{ \left( \vert g \vert - A \right)_+}_2 \\
		& \,\, \quad \quad \quad \quad \quad \quad \quad \quad \quad - \frac{\gamma}{\sqrt{N}} \norm{u - \rho h}_2 + \frac{\sqrt{1 + \gamma^2}}{N} \sum_{i=1}^{N} A_i \bigg\} \\
		= \, & N \max_{\substack{\rho \in [0, 1] \\ u \ge \kappa 1, u \notin M_2^{\veps}} } \min_{ \gamma \in [0, C_2]} \left\{ - \frac{\gamma}{\sqrt{N}} \norm{u - \rho h}_2 + \sqrt{1 + \gamma^2} \varphi_N (\rho; g) \right\},
	\end{align*}
	where $\varphi_N (\rho; g)$ is defined in \cref{lem:ULLN_1}.
Note that for any $u \ge \kappa 1$, we have
	\begin{equation*}
		\left\vert u_i - \rho h_i \right\vert \ge \left\vert u_i - \max(\kappa, \rho h_i) \right\vert + \left( \kappa - \rho h_i \right)_+, \ \forall i \in [M],
	\end{equation*}
	which leads to
	\begin{equation}
		\norm{u - \rho h}_2 \ge \sqrt{\norm{u - \max(\kappa 1, \rho h)}_2^2 + \norm{\left( \kappa 1 - \rho h \right)_+}_2^2}.
	\end{equation}
	Now we recast $C_2$ as $C$ and redefine
	\begin{equation}
		\eta_{M,N}^{\veps, (1)} (C) = \, N \max_{\substack{\rho \in [0, 1] \\ u \ge \kappa 1, u \notin M_2^{\veps}} } \min_{ \gamma \in [0, C]} \left\{ - \frac{\gamma}{\sqrt{N}} \norm{u - \rho h}_2 + \sqrt{1 + \gamma^2} \varphi_N (\rho; g) \right\},
	\end{equation}
	we then get that
	\begin{equation}\label{eq:gordon_res_2}
		\P \left( \eta_{M, N}^{\veps} \ge t \right) \le 2 \P \left( \eta_{M,N}^{\veps, (1)} (C) \ge t \right), \quad \forall t \in \R.
	\end{equation}
	In what follows, we show that $\eta_{M,N}^{\veps, (1)} (C) < \xi_{M, N}$ with high probability.
For any $\delta > 0$, we have
	\begin{align*}
		\frac{\eta_{M,N}^{\veps, (1)} (C)}{N} = \, & \max_{\substack{\rho \in [0, 1] \\ u \ge \kappa 1, u \notin M_2^{\veps}} } \min_{ \gamma \in [0, C]} \left\{ - \frac{\gamma}{\sqrt{N}} \norm{u - \rho h}_2 + \sqrt{1 + \gamma^2} \varphi_N (\rho; g) \right\} \\
		= \, & \max \Bigg\{ \max_{\substack{\rho \in [0, 1], \vert \rho - \rho_* \vert \ge \delta \\ u \ge \kappa 1, u \notin M_2^{\veps}} } \min_{ \gamma \in [0, C]} \left\{ - \frac{\gamma}{\sqrt{N}} \norm{u - \rho h}_2 + \sqrt{1 + \gamma^2} \varphi_N (\rho; g) \right\}, \\
		& \quad \quad \quad \max_{\substack{\rho \in [0, 1], \vert \rho - \rho_* \vert \le \delta \\ u \ge \kappa 1, u \notin M_2^{\veps}} } \min_{ \gamma \in [0, C]} \left\{ - \frac{\gamma}{\sqrt{N}} \norm{u - \rho h}_2 + \sqrt{1 + \gamma^2} \varphi_N (\rho; g) \right\}
		 \Bigg\}.
	\end{align*}
	Regarding the two inner maximums, we will estimate them separately:
	First,
	\begin{align}
		& \plim_{M/N \to \alpha} \max_{\substack{\rho \in [0, 1], \vert \rho - \rho_* \vert \ge \delta \\ u \ge \kappa 1, u \notin M_2^{\veps}} } \min_{ \gamma \in [0, C]} \left\{ - \frac{\gamma}{\sqrt{N}} \norm{u - \rho h}_2 + \sqrt{1 + \gamma^2} \varphi_N (\rho; g) \right\} \\
		\le \, & \plim_{M/N \to \alpha} \max_{\rho \in [0, 1], \vert \rho - \rho_* \vert \ge \delta } \min_{ \gamma \in [0, C]} \left\{ - \frac{\gamma}{\sqrt{N}} \norm{\left( \kappa 1 - \rho h \right)_+}_2 + \sqrt{1 + \gamma^2} \varphi_N (\rho; g) \right\} \\
		= \, & \max_{\rho \in [0, 1], \vert \rho - \rho_* \vert \ge \delta } \min_{ \gamma \in [0, C]} \left\{ - \gamma \sqrt{\alpha} \sqrt{\E \left[ \left( \kappa - \rho G \right)_+^2 \right]} + \sqrt{1 + \gamma^2} \varphi(\rho) \right\} \\
		< \, & \sqrt{ \varphi \left( \rho_* \right)^2 - \alpha \E \left[ \left( \kappa - \rho_* G \right)_+^2 \right] }
	\end{align}
	for sufficiently large $C$.
Second, we can choose $\delta$ to be sufficiently small such that
	\begin{equation*}
		\rho \in [\rho_* - \delta, \rho_* + \delta] \implies W_2 \left( \operatorname{Law}\left( \max(\rho G, \kappa) \right), P_* \right) \le \frac{\veps}{2}.
	\end{equation*}
	Consequently, if $u \ge \kappa 1$ and $u \notin M_2^{\veps}$, with high probability we have
	\begin{equation*}
		\frac{1}{\sqrt{M}} \norm{u - \max(\kappa 1, \rho h)}_2 \ge \frac{\veps}{4},
	\end{equation*}
	thus leading to
	\begin{equation*}
		\frac{1}{\sqrt{M}} \norm{u - \rho h}_2 \ge \sqrt{\frac{\veps^2}{16} + \frac{1}{M} \norm{\left( \kappa 1 - \rho h \right)_+}_2^2}.
	\end{equation*}
	It then follows that
	\begin{align}
		& \plim_{M/N \to \alpha} \max_{\substack{\rho \in [0, 1], \vert \rho - \rho_* \vert \le \delta \\ u \ge \kappa 1, u \notin M_2^{\veps}} } \min_{ \gamma \in [0, C]} \left\{ - \frac{\gamma}{\sqrt{N}} \norm{u - \rho h}_2 + \sqrt{1 + \gamma^2} \varphi_N (\rho; g) \right\} \\
		\le \, & \plim_{M/N \to \alpha} \max_{\rho \in [0, 1], \vert \rho - \rho_* \vert \le \delta } \min_{ \gamma \in [0, C]} \left\{ - \gamma \sqrt{\alpha_{M, N}} \sqrt{\frac{\veps^2}{16} + \frac{1}{M} \norm{\left( \kappa 1 - \rho h \right)_+}_2^2} + \sqrt{1 + \gamma^2} \varphi_N (\rho; g) \right\} \\
		= \, & \max_{\rho \in [0, 1], \vert \rho - \rho_* \vert \le \delta } \min_{ \gamma \in [0, C]} \left\{ - \gamma \sqrt{\alpha} \sqrt{\E \left[ \left( \kappa - \rho G \right)_+^2 \right] + \frac{\veps^2}{16}} + \sqrt{1 + \gamma^2} \varphi(\rho) \right\} \\
		< \, & \sqrt{ \varphi \left( \rho_* \right)^2 - \alpha \E \left[ \left( \kappa - \rho_* G \right)_+^2 \right] }
	\end{align}
	for sufficiently large $C$.
Combining these two estimates, we finally deduce that for $C$ large enough,
	\begin{equation}
		\plim_{M/N \to \infty} \frac{\eta_{M,N}^{\veps, (1)} (C)}{N} < \sqrt{ \varphi \left( \rho_* \right)^2 - \alpha \E \left[ \left( \kappa - \rho_* G \right)_+^2 \right] } = \plim_{M/N \to \infty} \frac{\xi_{M,N}}{N}.
	\end{equation}
	Combining the above inequality with \cref{eq:gordon_res_2} establishes \cref{eq:W2_conv_2} and consequently completes the proof of part 2.
\end{proof}

\begin{rem}
    \cref{eq:tight_fraction} precisely characterizes the proportion of binary coordinates after the first LP stage of our algorithm. In the second stage, we adapt the Edge-Walk algorithm from \cite{lovett2015constructive} to round the remaining coordinates to $\pm 1$, where \cref{eq:margin_conv} plays a crucial role in analyzing the asymptotic behavior of the margin constraints.
\end{rem}

\subsection{Stage 2: Rounding}\label{sec:rounding}
The rounding step is based on the discrepancy minimization algorithm of Lovett-Meka in \cite{lovett2015constructive}.
The idea is to take the output of the previous stage $\hat \theta$ as a starting place and run a random walk until either $\theta_i$ becomes $\{\pm 1\}$ (the corresponding variable constraint becomes tight, up to arbitrarily small error) or $\langle X_i, \theta \rangle$ becomes $\kappa + \veps$ (the inequality constraint almost becomes tight, up to margin $\veps$).
Then, we continue the random walk, but only in the subspace orthogonal to the previous tight constraints.
When the random walk evetually becomes stuck, with very high probability the number of tight variable constraints is significantly increased (there are more $\theta_i$'s in $\{\pm 1\}$).
Iterating this procedure in an appropriate way (with slightly reduced margin $\veps$ each time), we will get a solution in $\{\pm 1\}^N$ eventually.

\subsubsection{The Edge-Walk Algorithm and Partial Coloring Lemma}\label{sec:rounding_alg}
In the second step, we round the remaining non-binary coordinates of $\hat \theta$ to $\pm 1$ with \Cref{alg:edge-walk}, a variant of the edge-walk algorithm of \cite{lovett2015constructive}.
Our implementation of the algorithm differs in that the constraints are one-sided; further we have been more careful to optimize the constants.

\begin{algorithm}\caption{The Edge-Walk Algorithm.}\label{alg:edge-walk}

\textbf{Inputs:} 
Step size $\gamma > 0$, 
approximation parameter $\delta > 0$, 
constraint matrix $X \in \R^{M \times N}$, 
slack vector $c \in \R^M_+$,
and feasible starting point $\theta_0 \in [-1,1]^N \cap \{\theta : X\theta \ge c\}$.

\medskip

Define $T = \lceil 16/(3\gamma^2)\rceil$ the number of steps. For $t = 1,\ldots,T$ do
\begin{itemize}
\item Let $\cC^{\mathrm{var}}_t := \cC^{\mathrm{var}}_t(\theta_{t-1}) = \{i \in [N]: |(\theta_{t-1})_i| \ge 1-\delta\}$ be the set of variable constraints `nearly hit' so far.
\item Let $\cC^{\mathrm{disc}}_t := \cC^{\mathrm{disc}}_t(\theta_{t-1}) = \{j \in [M]: \<\theta_{t-1}-\theta_0,X_j\> \le (-c_j+\delta)\|X_j\|_2\}$ be the set of discrepancy constraints `nearly hit' so far.
\item Let
$ \cV_t := \cV(\theta_{t-1}) = \{u \in \R^n: u_i=0\; \forall i \in \cC^{\mathrm{var}}_t,\quad \<u,X_j\>=0 \;\forall j \in \cC^{\mathrm{disc}}_t\}$ be the linear subspace orthogonal to the `nearly hit' variables and discrepancy constraints.
\item Set $\theta_t := \theta_{t-1} + \gamma U_t$, where $U_t \sim \normal(\cV_t)$, the multi-dimensional standard Gaussian distribution supported on $\cV_t$.
\end{itemize}
\smallskip
\textbf{Output:} $\theta = \theta_T \in [-1,1]^N$ satisfying $\langle X_j,\theta - \theta_0\rangle \ge -c_j \|X_j\|_2$ for all $j \in [M]$.
\smallskip
\end{algorithm}

We now recall the main partial coloring lemma of \cite{lovett2015constructive}.
We state the lemma generalized and optimize it for the case of one-sided constraints.
The proof is identical to that of \cite[Theorem 4]{lovett2015constructive}, so we omit it here for brevity.

\begin{thm}[Partial Coloring Lemma \cite{lovett2015constructive}]\label{thm:mainpartialcolor}
Suppose $X,c,\theta_0,\gamma,\delta$ are defined as in \Cref{alg:edge-walk},
and further assume that $\delta = O(\gamma \sqrt{\log MN/\gamma})=o(1/\log N)$, and that the slack vector $c\ge 0$ satisfies
$\sum_{j=1}^M \exp(-c_j^2/K_1) \le N/K_2$, where $K_1 = 16$ and $K_2 = 8$.
Then with probability at least $0.1$, \Cref{alg:edge-walk} (the Edge-Walk algorithm) finds a point $\theta \in [-1,1]^N$ such that
\begin{enumerate}
\item[(i)] $\<\theta-\theta_0, X_j\> \ge -c_j \|X_j\|_2$.
\item[(ii)] $|\theta_i| \ge 1-\delta$ for at least $N/2$ indices $i \in [N]$.
\end{enumerate}
Moreover, the algorithm runs in time $O((M+N)^3 \cdot \delta^{-2} \cdot \log(NM/\delta))$.
\end{thm}

By default, we repeat the edge-walk algorithm $O(\log N)$ times in order to boost the success probability from $0.1$ to $1-o_N(1)$. 

\subsubsection{The Complete Algorithm and Its Performance Guarantee}
We will iterate the edge-walk algorithm multiple times until we have an almost-fully colored vector, and in the final step we apply randomized rounding to produce a genuine $\kappa$-margin solution. 
The details are summarized in \Cref{alg:full}.
\medskip

\begin{algorithm}\caption{The complete algorithm}\label{alg:full}

\noindent \textbf{Input:} 
Constraint matrix $X \in \R^{M \times N}$, initial intercept parameter $\kappa_0$, slack vectors $c^{(1)},\ldots,c^{(K)}$ for $K = \lceil 2 \log N\rceil$, and $\delta = o(1/\log N)$.

\begin{itemize} 
\item Compute the starting point $\hat \theta ^{(0)} = \hat \theta$ as the output of the linear program \eqref{eq:LP_def} with constraint matrix $X \in \R^{M \times N}$ and intercept $\kappa_0$ ($\kappa_0 > \kappa$, to allow for some slackness for the subsequent rounding steps).

\item For $k = 1,\ldots,K = \lceil2\log N\rceil$, do
\begin{itemize}
\item Let $I_{k-1}:=\{ i: \vert \hat \theta^{(k-1)}_i \vert < 1-\delta \}$ be the coordinates not `fixed' in the $k$-th iteration, and set $N_{k-1} = |I_{k-1}|$.
\item Run the edge-walk algorithm with constraints given by $X$ restricted to the columns of $X$ corresponding to $I_{k-1}$, slack vector $c^{(k)}$, starting point $\hat \theta^{(k-1)}$, and step size/approximation parameters $\gamma=\delta/\sqrt{\log N}$.
\item Set $\hat \theta ^{(k)}= \hat \theta ^{(k-1)}_T$, the output of the edge-walk algorithm.
\end{itemize}

\item \textbf{Randomized Rounding:}
For $i \in [N]$, take $\chi_i = \mathrm{sign}(\hat{\theta}^{(K)}_i)$ with probability $(1+|\hat{\theta}^{(K)}_i|)/2$ and $-\mathrm{sign}(\hat{\theta}^{(K)}_i)$ with probability $(1 - |\hat{\theta}^{(K)}_i|)/2$.
\end{itemize}

\noindent\textbf{Output:} 
$\chi \in \{\pm 1\}^N$.
\smallskip
\end{algorithm}

The following theorem gives guarantees for our full algorithm. 
\begin{thm}\label{thm:fullcolor}
    Let $X \in \R^{M \times N}$ where $X_{ij} \iidsim \normal (0, 1)$.
Denote $X_1, \cdots, X_M$ as the rows of $X$. Let $\hat{\theta}$ be the output of the first stage of our algorithm. Recall $K_1$ and $K_2$ from \cref{thm:mainpartialcolor}.
For $j \in [M] $ and $k \in [K]$, let $c_j^{(k)}\ge 0$ and $\delta>0$ be such that 
    \begin{align}\label{eq:cjrequirement}
        &\sum_{k=1}^K c_j^{(k)} \|(X_j)_{I_{k-1}}\|_2 \leq \< \hat \theta, X_j\> - \kappa\sqrt{N} - 4 \sqrt{\delta N \log M }, \quad \forall j \in [M],\\ \label{eq:cjrequirement2}
        &\sum_{j=1}^M \exp(- (c_j^{(k)})^2 /K_1 ) \leq \frac{N_{k-1}}{K_2}, \quad \forall k \in [K].
    \end{align}
    Then with high probability, the output of \Cref{alg:full} satisfies $\chi \in \{\pm 1\}^N$ and $X \chi \geq \kappa \sqrt{N}$.
\end{thm}

\begin{proof}[Proof of \cref{thm:fullcolor}]
    By iterating the partial coloring lemma, under our assumptions on the slack vectors $c^{(k)}$ we have $\< \hat \theta^{(K)}, X_j \> \geq \kappa+4 \sqrt{\delta N \log M }$ for all $j \in [M]$.

    We claim that the difference between $\hat{\theta}^{(K)}$ and $\chi$ is small.
Define $Y=\chi-\hat{\theta}^{(K)}$.
We are going to show that with high probability, the inner product satisfies
    \begin{equation}\label{eq:randomsignrounding}
        \<Y,X_j\> = \<\chi,X_j\> - \<\hat{\theta}^{(K)},X_j\> \geq -4 \sqrt{\delta N \log M }.
    \end{equation}
    Note that $\E[\chi_i]=\hat{\theta}^{(K)}_i$.
We have that $|Y_i| \le 2$, $\E[Y_i]=0$ and $\var(Y_i) \le \delta$.
Fix some $j \in [M]$, with probability larger than $1-N^{-2}$, we have that $\|X_j\|_2 \le \sqrt{2N}$ (by the Laurent-Massart bound \cite{laurent2000adaptive}) and $\|X_j\|_{\infty} \le 2\sqrt{\log N}$.
Thus, by a standard Chernoff bound,
    \begin{equation}
    \P\left[ \<Y,X_j\> <  -2 \sqrt{2\log M} \cdot \sqrt{2 N \delta} \right] \leq \exp(-2 \log M) = M^{-2}.
    \end{equation}
    Therefore, our claim holds by the union bound, and $\< \chi, X_j \> \geq \kappa \sqrt{N}$ for all $j \in [M]$.
\end{proof}

The remaining task is to choose a sequence of slack vectors $c_j^{(k)}$ to use in \cref{thm:fullcolor}.
We first present a lemma that describes the behavior of $\|(X_j)_{I_k}\|_2 $.
\begin{lem}\label{lem:normbound}
    For any $k \geq 0$ with $N_k\geq 1$, we have that with high probability,
    \begin{equation}\label{eq:normbound}
         W_2 \left( \frac{1}{M} \sum_{j=1}^{M} \delta_{\|(X_j)_{I_k}\|^2_2/N_k}, \delta_1 \right) \leq \, \frac{\sqrt{2}}{\sqrt{N_k}}+ o_N(1).
    \end{equation}
    In particular, unless $N_k = O_N(1)$, the right hand size is $o_N(1)$.
\end{lem}
\begin{proof}
    Note that if $I_k$ is a fixed subset of $[N]$, then the bound is rather standard. Instead, $I_k$ is random and it depends on the matrix $X$ in a complicated way. So we use a slightly more involved argument and apply the union bound. For any fixed $J \subset [N]$ with $|J| = \ell\geq 1$, we note that $\|(X_j)_{J}\|^2_2$ follows the $\chi^2$-distribution with degree of freedom $\ell$.
Therefore, by the Laurent-Massart bound \cite{laurent2000adaptive}, we have that
    \begin{equation}\label{lem:LMbound}
        \P \left( \left| \|(X_j)_{J}\|^2_2 - \ell \right| \geq 2\sqrt{\ell x} + 2x  \right) \leq 2\exp(-x).
    \end{equation}
    For simplicity of notations, we define $h_\ell = \log \binom{N}{\ell} + \log N$. For $x_0 = N^{1/6}$, define a set of \textit{bad} indices $B=\{j\in [M]: |\|(X_j)_{J}\|^2_2 - \ell| \geq 2 \sqrt{ \ell x_0} + 2 x_0\}$.
Then by \cref{lem:LMbound}, $|B|$ is stochastically dominated by a Bin$(M, 2\exp(-x_0))$ random variable.
Therefore, by the Chernoff bound, with probability larger than $1-\exp(-2h_\ell)$, $|B| = O(h_\ell/ x_0)$.
Now for $j \in B^c$, we have uniform control of $|\|(X_j)_{J}\|^2_2 - \ell|$ and the random variables are independent. Thus we can apply the Azuma-Hoeffding inequality, and get that
    \begin{align*}
        \P \left(\sum_{j \in B^c} \left[\left( \|(X_j)_J\|^2_2  - \ell \right)^2  -2 \ell\right]\geq t \right) \leq \exp \left ( -\frac{2t^2}{|B^c| (2\sqrt{\ell x_0} + 2x_0)^4} \right).
    \end{align*}
    Rearranging the terms, we get that 
    \begin{align*}
        \P \left(\frac{1}{M} \sum_{j \in B^c} \left( \frac{\|(X_j)_J\|^2_2}{\ell} -1 \right)^2 \geq \frac{2 |B^c|}{\ell M} + \frac{t}{\ell^2 M} \right) \leq \exp\left ( -\frac{2t^2}{|B^c| (2\sqrt{\ell x_0} + 2x_0)^4} \right) \leq \exp(-2 h_\ell),
    \end{align*}
    if we take $t = c\sqrt{h_\ell M} (\sqrt{\ell x_0}+x_0)^4$, for constant $c$ large enough. In this case, the term
    \begin{align*}
        \frac{t}{\ell^2 M} \leq    \frac{c\sqrt{h_\ell M} (\sqrt{\ell x_0}+x_0)^2}{\ell^2 M} = O \left( \max \left( \frac{\sqrt{h_\ell}x_0}{\sqrt{M} \ell}, \frac{\sqrt{h_\ell} x_0^2}{\ell^2 \sqrt{M}}\right) \right) \leq O \left (\frac{\sqrt{\log N} }{N^{1/6}} \right).
    \end{align*}
    Now for $j \in B$, again by \cref{lem:LMbound}, $\|(X_j)_J\|^2_2 = O(\sqrt{\ell h_\ell} + h_\ell)$ with probability larger than $1-\exp(-2 h_\ell)$. This implies that there exists a constant $c$ such that 
    \begin{align*}
        \P \left(\frac{1}{M} \sum_{j \in B} \left( \frac{\|(X_j)_J\|^2_2}{\ell} -1 \right)^2 \geq \frac{c|B|}{M} \max\left( 1,\frac{h_\ell^2}{\ell^2}\right) \right) \leq \exp(-2 h_\ell).
    \end{align*}
    Since $|B| = O(h_\ell/x_0)$ with probability $1-\exp(-2h_\ell)$, we have that 
    \begin{align*}
        \frac{c|B|}{M} \max\left( 1,\frac{h_\ell^2}{\ell^2}\right) = O \left( \frac{\log^3 N}{ N^{1/6}} \right).
    \end{align*}
    Combining the above two estimates for $B$ and $B^c$, we have that
\begin{align*}
    \P \left(\frac{1}{M} \sum_{j \in [M]} \left( \frac{\|(X_j)_J\|^2_2}{\ell} -1 \right)^2 \geq \frac{2}{\ell} + \frac{c \log^3 N}{N^{1/6}} \right) \leq \exp(-2h_\ell).
\end{align*}
    By a union bound over all $J\subset [N]$ with $|J| = \ell$, we have the lemma.
\end{proof}
We also state another lemma that bounds $\|(X_j)_{I_k}\|_2 $ for all $j$.
\begin{lem}\label{lem:LMboundall}
    With high probability for all $k = 0, \cdots, K-1$ and all $j \in [M]$, 
    \begin{align*}
        \|(X_j)_{I_{k}}\|^2_2 \leq 9 ( N_k\log(N/N_k)+ N_k).
    \end{align*}
\end{lem}
\begin{proof}
    For any $J \subset [N]$ with $|J| = \ell\geq 1$, we note that $\|(X_j)_{J}\|^2_2$ follows the $\chi^2$-distribution with degree of freedom $\ell$.
Therefore, by the Laurent-Massart bound \cite{laurent2000adaptive}, we have that
    \begin{equation}\label{lem:LMbound2}
        \P \left( \|(X_j)_{J}\|^2_2 \ge \ell+ 2\sqrt{\ell x} +2x \right) \leq \exp(-x).
    \end{equation}
     By Stirling's formula, $\log \binom{N}{\ell} \leq N \Ent(\ell/N)\leq \ell(\log (N/\ell)+1)$, where $\Ent(\cdot)$ is the entropy function.
If we take $x = \ell (\log(N/\ell)+1) + 2\log (N)$, then $\ell+ 2\sqrt{\ell x} +2x \leq 2\ell+3x = 3\ell \log(N/\ell) + 6 \log(N)+5\ell \leq 9(\ell\log(N/\ell)+\ell)$.
Therefore, we have
    \begin{equation}\label{lem:LMbound3}
        \P \left( \|(X_j)_{J}\|^2_2 \ge 9 ( \ell\log(N/\ell)+ \ell) \right) \leq \exp(-(\ell (\log(N/\ell)+1) + 2\log (N)))\leq \frac{1}{N^2 \binom{N}{\ell}}.
    \end{equation}
    By a union bound over all $k \in [K]$, $j \in [M]$ and $I_{k}\subset [N]$ with $|I_{k}| = N_{k}$, we have the lemma.
\end{proof}

Now we state and prove the main result of the present section.
Recall the definition of $\rho(\alpha, \kappa)$, $t(\alpha, \kappa)$ and $P(\alpha, \kappa)$ in \cref{def:LP_auxiliary} and \cref{def:tP}.
Define $R_0 = 1-2 \Phi( -t(\alpha, \kappa_0))$.
For $1\leq k \leq K$, define $\hat N_k = N_k + \sqrt{2N_k}$.
Further, define $R_k = 2^{-k+1} R_0$, $\hat R_k = 11R_k(\log(1/R_k)+1)$ and $\tilde R_k = 9R_k(\log(1/R_k)+1)^2$.
\begin{thm}[Proportional choice of slack vectors]\label{thm:differentc}
    For any $\kappa_0 > \kappa$ and $\beta_1, \cdots, \beta_K\geq 0$ satisfying $\sum_{k=1}^K \beta_k < 1$, if there exists a constant $K_p\geq 0$ such that
    \begin{align}
        &\label{eq:condition1}
		\alpha < \sup_{\rho \in [0, 1]} \left\{ \frac{\varphi(\rho)^2}{\E \left[ \left( \kappa_0 - \rho G \right)_+^2 \right]} \right\}, \\ \label{eq:condition2}
        &\alpha < \left( \E_{Y \sim P(\alpha, \kappa_0)}\exp\left( -\frac{\beta_k^2 (Y-\kappa)^2 }{ K_1 R_k} \right) \right)^{-1}\frac{R_k}{K_2} , \qquad \forall 1\leq k \leq K_p,\\
        \label{eq:condition3}
        &\alpha < \exp\left( \frac{\beta_k^2 (\kappa_0-\kappa)^2 }{ K_1 \hat R_k } \right)\frac{R_k}{ K_2} , \qquad \forall K_p < k \leq K,\\
        \label{eq:condition4}
        &\frac{\beta_k^2 (\kappa_0-\kappa)^2}{K_1}  > 3R_k, \quad \forall 1 \leq k \leq K_p, \qquad \frac{\beta_k^2 (\kappa_0-\kappa)^2}{2 K_1}  >\tilde R_k, \quad \forall K_p < k \leq K,
    \end{align}
    then with the choice
    \begin{equation}
        c_j^{(k)}  =  \beta_k\frac{ \< \hat \theta, X_j\> - \kappa \sqrt{N} - 4 \sqrt{\delta N \log M }}{\|(X_j)_{I_{k-1}}\|_2}, \quad \forall k \in [K], \quad \text{and} \quad \delta = o(1/\log M),
    \end{equation}
    \cref{eq:cjrequirement} and \cref{eq:cjrequirement2} hold with high probability.
    Therefore, the algorithm finds a solution with high probability.
\end{thm}
\begin{rem}
    The first condition (\cref{eq:condition1}) is to make sure that the algorithm in Stage 1 finds a solution, with margin $\kappa_0$.
The second condition (\cref{eq:condition2}) guarantees that in the first $K_p$ iterations of the Edge-Walk algorithm, the conditions in \cref{thm:fullcolor} are satisfied.
The third condition (\cref{eq:condition3}) focuses on later rounds, when $N_k$ might be sublinear in $N$.
And finally the fourth condition (\cref{eq:condition4}) provides a nice monotonicity property.
\end{rem}
\begin{proof}[Proof of \cref{thm:differentc}]
    By our choice of $c_j^{(k)}$'s, \cref{eq:cjrequirement} holds automatically. Since $\delta = o(1/\log M)$, the $c_j^{(k)}$'s are guaranteed to be non-negative since $\kappa_0>\kappa$.
To verify \cref{eq:cjrequirement2}, we note that 
    \begin{align*}
        \frac{1}{M}\sum_{j=1}^M \exp(- (c_j^{(k)})^2 /K_1 )  &= \frac{1}{M}\sum_{j=1}^M \exp \left(- \beta_k^2 \frac{ (\< \hat \theta, X_j\> - \kappa\sqrt{N} - 4 \sqrt{\delta N \log M })^2  }{ K_1 \|(X_j)_{I_{k-1}}\|_2^2 } \right) \\
        &= \frac{1}{M}\sum_{j=1}^M \exp \left(- \beta_k^2 \frac{ (\< \hat \theta, X_j\>/\sqrt{N} - \kappa - 4 \sqrt{\delta  \log M })^2 }{ K_1 \|(X_j)_{I_{k-1}}\|_2^2/N } \right).
    \end{align*}
    By \cref{eq:margin_conv} and \cref{lem:normbound}, we have that for $1 \leq k \leq K_p$,
    \begin{align*}
        &\, \plim_{M/N \to \alpha} \frac{1}{M}\sum_{j=1}^M \exp \left(- \beta_k^2 \frac{ (\< \hat \theta, X_j\> - \kappa\sqrt{N} - 4 \sqrt{\delta N \log M })^2 }{ K_1 \|(X_j)_{I_{k-1}}\|_2^2 } \right) \\
        &\leq \plim_{N \to \infty} \E_{Y \sim P(\alpha, \kappa_0)}\exp\left( -\frac{\beta_k^2 (Y-\kappa)^2 }{K_1 (\hat N_{k-1}/N)  } \right).
    \end{align*}
    Note that for each $Y \geq \kappa_0$, the function $f(x):=x \exp(\beta_k^2(Y-\kappa)^2 N/(K_1 (x+\sqrt{2x})))$ is decreasing in $x$ whenever $x+\sqrt{2x} \leq \beta^2(\kappa_0-\kappa)^2/K_1$.
By \cref{eq:tight_fraction} and \cref{thm:mainpartialcolor}, we know that 
    $$
    N_{k-1} \leq (1+o_N(1))R_0 2^{-k+1}N = (1+o_N(1))R_k N.
    $$
Thus by \cref{eq:condition4}, $(N_{k-1}+\sqrt{2N_{k-1}})/N < 3 R_k < \beta^2(\kappa_0-\kappa)^2/K_1$.
Therefore 
$$
\E_{Y \sim P(\alpha, \kappa_0)}\exp\left( -\frac{\beta_k^2 (Y-\kappa)^2 }{K_1 (\hat N_{k-1}/N)  } \right)/N_{k-1}
$$ 
is increasing in $N_{k-1}$.
In particular, this implies that
    \begin{align*}
        \E_{Y \sim P(\alpha, \kappa_0)}\exp\left( -\frac{\beta_k^2 (Y-\kappa)^2 }{K_1 (\hat N_{k-1}/N)  } \right)/N_{k-1} \leq (1+o_N(1))\E_{Y \sim P(\alpha, \kappa_0)}\exp\left( -\frac{\beta_k^2 (Y-\kappa)^2 }{K_1 \hat R_k  } \right)/(R_k N),
    \end{align*}
    since $N_{k-1} \leq (1+o_N(1))R_k N$.
Therefore, \cref{eq:condition2} implies that for $1 \leq k \leq K_p$,
    \begin{align*}
        \plim_{M/N \to \alpha} \frac{1}{M}\sum_{j=1}^M \exp(- (c_j^{(k)})^2 /K_1 ) \leq \plim_{N \to \infty} \E_{Y \sim P(\alpha, \kappa_0)}\exp\left( -\frac{\beta_k^2 (Y-\kappa)^2 }{K_1 (\hat N_{k-1}/N)  } \right) \leq \frac{N_{k-1}}{\alpha K_2 N}.
    \end{align*}
    Since $M/N \to \alpha$, we have \cref{eq:cjrequirement2}.
For $K_p<k\leq K$, by \cref{lem:LMboundall}, with high probability,
    \begin{align*}
        \sum_{j=1}^M \exp(- (c_j^{(k)})^2 /K_1 )  
        &= \sum_{j=1}^M \exp \left(- \beta_k^2 \frac{ (\< \hat \theta, X_j\>/\sqrt{N} - \kappa - 4 \sqrt{\delta  \log M })^2 }{ K_1 \|(X_j)_{I_{k-1}}\|_2^2/N } \right)\\
        & \leq \sum_{j=1}^M \exp \left(- \beta_k^2 \frac{(\kappa_0 - \kappa - 4 \sqrt{\delta  \log M })^2 }{ 9 K_1 N_{k-1}(\log(N/N_{k-1})+1) /N } \right)\\
        & \stackrel{(i)}{\leq} \sum_{j=1}^M \exp \left(- \beta_k^2 \frac{(\kappa_0 - \kappa)^2 }{ (10 K_1 N_{k-1}(\log(N/N_{k-1})+1) /N } \right)\\
        & \stackrel{(ii)}{\leq} \frac{N_{k-1}}{K_2}.
    \end{align*}
    Here in $(i)$ we used $\delta = o(1/\log M)$.
In $(ii)$, we used monotonicity from \cref{eq:condition4}.
Indeed, the function $x\exp(a/(10x(\log(1/x)+1)))$ is decreasing if $10x(\log(1/x)+1)^2\leq a$.
So \cref{eq:condition4}, \cref{eq:condition3} and $N_{k-1}/N \leq (1+o_N(1))R_k$ together imply that 
    \begin{align*}
        \alpha \exp\left( -\frac{\beta_k^2 (\kappa_0-\kappa)^2 }{ 11 K_1 N_{k-1}/N (\log(N/N_{k-1})+1) } \right) < \frac{N_{k-1}}{K_2 N} , \qquad \forall K_p < k \leq K.
    \end{align*}
        This implies $(ii)$ since $M/N \to \alpha$.
\end{proof}

\subsection{Thresholds for Our Algorithm as a Function of the Margin}\label{sec:algthreshold}
In this section, we customize the analysis of our algorithm in the cases $\kappa \to -\infty$, $\kappa = 0$, and $\kappa \to \infty$, deriving a lower bound on the algorithmic threshold for each case.

\subsubsection{Large Negative Margin}
We first deal with the special case when $\kappa \to -\infty$.
\begin{thm}\label{thm:kappaverynegative}
    There exists an absolute constant $C_\mathrm{ALG} > 0$, such that with high probability, our two-step algorithm (with properly tuned parameters) finds a $\kappa$-margin solution when $\alpha < \frac{C_\mathrm{ALG}}{\Phi(\kappa) \kappa^2}$ for sufficiently negative $\kappa$.
\end{thm}
\begin{proof}[Proof of \cref{thm:kappaverynegative}]
    To establish the theorem, we need to propose appropriate $\kappa_0$ and $\beta_k$ and verify the conditions in \cref{thm:differentc}.
For a constant $c_a$, we fix $\alpha_0 = c_a/\kappa^2 \Phi(\kappa)$.
Recall the definition of $R_k$, $\hat R_k$ and $\tilde R_k$ right above \cref{thm:differentc}.
We use $R_k(\alpha)$ to emphasize its dependence on $\alpha$.
Since $R_k(\alpha)$ is increasing in $\alpha$, if we replace all the $R_k(\alpha)$, $\tilde R_k(\alpha)$ and $\hat R_k(\alpha)$ in \cref{eq:condition2}, \cref{eq:condition3}, and \cref{eq:condition4} by $R_k(\alpha_0)$, $\tilde R_k(\alpha_0)$ and $\hat R_k(\alpha_0)$, we will get more strict inequalities if $\alpha<\alpha_0$.
Indeed this is the case for \cref{eq:condition4} because $R_k(\alpha)$ and $\tilde R_k(\alpha)$ are both increasing in $\alpha$.
And if \cref{eq:condition4} holds for $R_k(\alpha_0)$, then the right hand sides of \cref{eq:condition2} and \cref{eq:condition3} are decreasing in $\alpha$.
So for the purpose of the proof, it suffices to propose $\kappa_0$ and $\beta_k$ and verify the conditions in \cref{thm:differentc} with $R_k(\alpha_0)$, $\tilde R_k(\alpha_0)$ and $\hat R_k(\alpha_0)$.

    For constants $c_0>0$ and $0<\beta_0<1/6$ to be determined later, we take $\kappa_0 = \kappa + c_0/|\kappa|$ and $\beta_k = \beta_0 2^{-k/4}$.
Indeed $\sum_{k=1}^K \beta_k < 1$.
Clearly $\kappa_0 < 0$, so by \cref{lem:feasibility_LP}, \cref{eq:abs_bound_alpha} (thus \cref{eq:condition1}) holds.
Furthermore, by \cref{eq:tight_fraction} and \cref{lem:large_neg_kappa},
    \begin{equation}\label{eq:bound_n1}
        R_0 = - 2\Phi \left( -\left( 3 \sqrt{\frac{\pi}{2}} b + o_{\kappa} (1) \right) \frac{1}{\kappa^2} \right) = \frac{3b}{\kappa^2} (1+o_{\kappa} (1)).
    \end{equation}
    And therefore, for any $k \in [K]$, 
    \begin{equation}
        R_k = 2^{-(k-1)} R_0 = \frac{6b}{\kappa^2 2^{k}} (1+o_{\kappa} (1)).
    \end{equation}
    We note that since $\Phi(\kappa_0) = \Phi(\kappa) \exp(c_0)(1+o_\kappa(1))$, \cref{eq:asymp_rho_t_neg_kappa} implies that $3be^b = c_a e^{c_0} (1+o_\kappa(1))$.
In particular, when $c_0$ is large, $b\leq c_0$.
Now we look at the requirements \cref{eq:condition2}.
For any constant $a>0$, by \cref{def:tP}, we compute
    \begin{align*}
        &\,\, \E_{Y \sim P(\alpha_0, \kappa_0)}\exp\left( - a \kappa^2(Y-\kappa)^2 \right) \\
        &= \int_{\kappa_0}^\infty \frac{1}{\sqrt{2\pi \rho^2}}\exp(-a \kappa^2(y-\kappa)^2) \exp(-y^2/2\rho^2) dy + \Phi(\kappa_0)\exp(-a \kappa^2 (\kappa_0-\kappa)^2)\\
        & = \Phi \left( - \frac{\kappa_0-\mu_1}{\mu_2} \right) \frac{\mu_2}{\rho} \exp(\mu_3) +  \Phi\left(\kappa+ \frac{c_0}{|\kappa|}\right)\exp(-a \kappa^2 (\kappa_0-\kappa)^2),
    \end{align*}
    where
    \begin{equation}
        \mu_1 = \frac{2a\kappa^3 \rho^2}{ 2a\kappa^2 \rho^2 + 1}, \qquad \mu_2 = \sqrt{\frac{\rho^2}{2a\kappa^2 \rho^2 +1}}, \qquad \mu_3 = \frac{-a\kappa^4}{ 2a\kappa^2 \rho^2 +1 }.
    \end{equation}
    Further simplify, by \cref{lem:large_neg_kappa}, we have that
    \begin{align}
        &\,\, \E_{Y \sim P(\alpha, \kappa_0)}\exp\left( - a \kappa^2(Y-\kappa)^2 \right) \\
        &= (1+ o_\kappa(1))\Phi \left(\frac{1}{\sqrt{2a} \rho^2} - c_0 \sqrt{2a}\right) \frac{1}{\sqrt{2a\rho^2} |\kappa| }\exp\left( -\frac{\kappa^2}{2\rho^2} + \frac{1}{4a \rho^4} \right)  +  \Phi\left(\kappa+ \frac{c_0}{|\kappa|}\right)\exp(-a c_0^2)\\ \label{eq:expectation_comp}
        & = (1+ o_\kappa(1)) \mu_4(a) \Phi(\kappa) ,
    \end{align}
    where
    \begin{align*}
        \mu_4(a) = \Phi \left( \frac{1}{\sqrt{2a}} - c_0 \sqrt{2a}\right) \frac{1}{\sqrt{2a}} \exp\left( -b + \frac{1}{4a} \right) + \exp\left(c_0-a c_0^2 \right).
    \end{align*}
    So to verify \cref{eq:condition2}, it suffices to have 
    \begin{align*}
        \alpha  < \min_{1\leq k \leq K_p} \frac{R_k(1-o_\kappa(1))}{K_2 \mu_4\left( \frac{\beta_k^2}{K_1 R_k \kappa^2} \right) \Phi(\kappa)} = \frac{(1-o_\kappa(1))}{\Phi(\kappa) \kappa^2} \min_{1\leq k \leq K_p} \frac{6b}{2^k K_2 \mu_4\left( \frac{\beta_k^2 2^k}{6b K_1 } \right) }.
    \end{align*}
    We note that $\mu_4(a)\leq C_1 \Phi (- C_2\sqrt{a})+\exp(c_0 -ac_0^2) \leq \exp(-C_3a )$ when $a = \Omega(1)$, for some positive constants $C_1$, $C_2$ and $C_3$.
Thus
    \begin{align*}
        \min_{1\leq k \leq K_p} \frac{6b}{2^k K_2 \mu_4\left( \frac{\beta_k^2 2^k}{6b K_1 } \right) } \geq \min_{k \geq 1} \frac{6b}{K_2} 2^{-k} \exp\left( \frac{C_3 \beta_0^2 2^{k/2}}{6bK_1} \right) \geq C_4,
    \end{align*}
    for some positive constant $C_4$.
So $\alpha < C_4(1-o_\kappa(1))/(\Phi(\kappa) \kappa^2)$ implies \cref{eq:condition2}.
Now for \cref{eq:condition4}, we check that
    \begin{align*}
        \frac{\beta_k^2 (\kappa_0-\kappa)^2}{3 K_1 R_k} = \frac{2^{k/2} \beta_0^2 c_0^2}{3K_1 R_0 \kappa^2} = (1+ o_\kappa(1)) \frac{2^{k/2} \beta_0^2 c_0^2}{ 18 b K_1} > 1, 
    \end{align*}
    for any $k\geq 1$ and any $\beta_0^2 c_0^2/b$ sufficiently large.
Further,
    \begin{align*}
        \frac{\beta_k^2 (\kappa_0-\kappa)^2}{2 K_1 \tilde R_k} = \frac{2^{k/2} \beta_0^2 c_0^2}{18 K_1 R_0 \kappa^2 (\log(2^k/R_0)+1)^2} = (1+ o_\kappa(1))\frac{\beta_0^2 c_0^2}{108 b K_1} \frac{2^{k/2}}{(\log(2^k \kappa^2 /6b)+1)^2} >1,
    \end{align*}
    whenever $k > 10 \log \log (|\kappa|)$ and $\beta_0^2 c_0^2/b$ sufficiently large.
Take $K_p = 10 \log \log (|\kappa|)$.
Finally, we check \cref{eq:condition3}.
Similar to the argument for \cref{eq:condition2}, we compute
    \begin{align*}
        \min_{k > K_p} \exp\left( \frac{\beta_k^2 (\kappa_0-\kappa)^2 }{ K_1 \hat R_k } \right)\frac{R_k}{ K_2}  = \min_{k > K_p} \exp\left( \frac{2^{k/2} \beta_0^2 c_0^2}{ 6b K_1 (\log(2^k \kappa^2 /6b)+1) } \right)\frac{6b}{ 2^k K_2} > 1, 
    \end{align*}
    for $k > 10 \log \log (|\kappa|)$ and $\beta_0^2 c_0^2/b$ sufficiently large, which verifies \cref{eq:condition3}.
So if we take $\beta_0 = 0.1$ and $c_0$ sufficiently large, then all the requirements are satisfied since $c_0 \geq b$ when $c_0$ is sufficiently large.
\end{proof}

\subsubsection{Margin Zero}
In this section, we look at another special case when $\kappa = 0$.
This case is more involved; we will first give a simple analysis that works for $\alpha \le 0.05$, and later in \Cref{app:improved} we give an even more specialized  analysis of the partial coloring lemma that works for $\alpha \le 0.1$.
To improve the constants in \cite{lovett2015constructive}, we firstly present a tail estimate for the gaussian random walk.
\begin{lem}\label{lem:gaussiantail}
    Let $Z_0 = 0, Z_1, \cdots, Z_T \in \mathbb{R}^n$ be a martingale with increments $Y_i = Z_i-Z_{i-1}$.
Let $v$ be a fixed vector in $\mathbb{R}^n$ such that $\|v\|_2 = 1$.
Assume that for all $1\leq i \leq T$, and $z_1, \cdots, z_{i-1} \in \mathbb{R}^n$, $Y_i\mid (Z_1 = z_1, \cdots, Z_{i-1} = z_{i-1}) \sim \normal(0, \Sigma_i)$, with $\|\Sigma_i\|_{\mathrm{op}} \leq 1$ (note that $\Sigma_i$ might depend on $Z_1, \cdots Z_{i-1}$).
Then for any $\lambda>0$, 
    \begin{align*}
        \mathbb{P}[\<v, Z_T\> \geq \lambda \sqrt{T}] \leq 2 \Phi (- \lambda).
    \end{align*}
\end{lem}
\begin{proof}
For $i = 0 , \cdots, T$, we write $X_i = \<v, Z_i\>$.
Then $X_T$ can be written as $X_T = \<v, Y_1\>+\<v, Y_2\>+\cdots+\<v, Y_T\> = V_1 N_1 + V_2 N_2 + \cdots + V_T N_T$, where $N_1, \cdots, N_T$ are $ \iidsim \normal(0,1)$, and each $V_i$ satisfies: (1) $V_i$ might depend on $Z_1, \cdots, Z_{i-1}$, (2) $V_i$ is independent of $N_{i+1}, \cdots, N_T$ and (3) $|V_i| \leq 1$.
Note that (3) follows from the fact that $|V_i|^2 = v^\top \Sigma_i v \leq 1$.
This implies that we can construct a standard Brownian motion $B$ and a sequence of stopping times $\tau_i$ such that
\begin{align*}
    \{B_{\tau_i}\}_{i \in [T]} \overset{d}{=} \{X_i\}_{i \in [T]}, \qquad \tau_i\overset{d}{=} \sum_{j=1}^i V_i^2 \quad \forall i \in [T].
\end{align*}
In particular, this implies that $\tau_T \leq T$ and therefore,
\begin{align*}
    \mathbb{P}[\<v, Z_T\> \geq \lambda \sqrt{T}] = \mathbb{P}[X_T \geq \lambda \sqrt{T}] = \mathbb{P}[B_{\tau_T} \geq \lambda \sqrt{T}] \leq \mathbb{P}[\max_{0\leq t \leq T } B_t \geq \lambda \sqrt{T}] = 2 \Phi (- \lambda),
\end{align*}
completing the proof.
\end{proof}
We next state and prove a partial coloring lemma with improved constants for the case $\kappa = 0$:
\begin{thm}[Partial Coloring Lemma for $\kappa = 0$]\label{thm:mainpartialcolor_zero}
Let positive constants $K_1, K_2, K_3$ be such that $1/K_1 + 1/K_2 + 1/K_3 \le 0.999$.
Let $X_1,\ldots,X_M \in \R^N$ be vectors, and $\theta_0 \in [-1,1]^N$ be a ``starting'' point.
Let $c_1,\ldots,c_M \ge 0$ be thresholds such that $2 \sum_{j=1}^M \Phi( - c_j/\sqrt{K_1}) \le N/K_2$.
Let $\delta>0$ be a small approximation parameter.
Then with probability at least $0.0008/(1 - 1/K_3)$, the Edge-Walk algorithm (\cref{alg:edge-walk}) in \cref{sec:rounding_alg} finds a point $\theta \in [-1,1]^N$ such that
\begin{enumerate}
\item[(i)] $\<\theta-\theta_0, X_j\> \ge -c_j \|X_j\|_2$.
\item[(ii)] $|\theta_i| \ge 1-\delta$ for at least $N/K_3$ indices $i \in [N]$.
\end{enumerate}
Moreover, the algorithm runs in time $O((M+N)^3 \cdot \delta^{-2} \cdot \log(NM/\delta))$.
\end{thm}
\begin{proof}
    The proof strategy closely follows that of \cite[Theorem 4]{lovett2015constructive}, so we will mainly focus on their differences here. Set $T = K_1 / \gamma^2$, and let $\theta = \theta_T$ be the output of the Edge-Walk algorithm in \cref{sec:rounding_alg}. It suffices to show that
    \begin{enumerate}
        \item [(a)] $\theta_T \in [-1, 1]^N$ and $\langle \theta_T - \theta_0, X_j \rangle \ge - c_j \norm{X_j}_2$ for all $j \in [M]$.
        \item [(b)] $|(\theta_T)_i| \ge 1-\delta$ for at least $N/K_3$ indices $i \in [N]$.
    \end{enumerate}
    The proof of part (a) is completely analogous to that of \cite[Claim 13]{lovett2015constructive}, and we omit it for brevity. To prove part (b), recall $\cC^{\mathrm{var}}_t$, $\cC^{\mathrm{disc}}_t$, $\cV_t$, and $U_t \sim \normal(\cV_t)$ from the Edge-Walk algorithm. Similar to \cite[Claim 14]{lovett2015constructive}, we first establish an a priori upper bound on $\vert \cC^{\mathrm{disc}}_T \vert$. By definition, $j \in \cC^{\mathrm{disc}}_T$ implies that $\langle \theta_T - \theta_0, X_j \rangle \le ( - c_j + \delta ) \norm{X_j}_2$. Since $\theta_T - \theta_0 = \gamma (U_1 + \cdots + U_T)$ where $U_t \sim \normal (\cV_t)$, we know that
    \begin{align*}
        \P \left( j \in \cC^{\mathrm{disc}}_T \right) \le \P \left( \left\langle \frac{X_j}{\norm{X_j}_2}, \ \sum_{t=1}^{T} U_t \right\rangle \le - \frac{c_j - \delta}{\gamma} \right) \stackrel{(i)}{\le} 2 \Phi \left( - \frac{c_j - \delta}{\gamma \sqrt{T}} \right) \stackrel{(ii)}{=} 2 \Phi \left( - \frac{c_j - \delta}{\sqrt{K_1}} \right),
    \end{align*}
    where $(i)$ follows from \cref{lem:gaussiantail}, and in $(ii)$ we use the fact $\gamma^2 T = K_1$. Consequently,
    \begin{align}
        \E \left[ \left\vert \cC^{\mathrm{disc}}_T \right\vert \right] = \, & \sum_{j=1}^{M} \P \left( j \in \cC^{\mathrm{disc}}_T \right) \le 2 \sum_{j=1}^{M} \Phi \left( - \frac{c_j - \delta}{\sqrt{K_1}} \right) \\
        \le \, & 2 \sum_{j=1}^{M} \Phi \left( - \frac{c_j}{\sqrt{K_1}} \right) + O(\delta M) \le \frac{N}{K_2} + o(N).
    \end{align}
    Following the proof of \cite[Claim 16]{lovett2015constructive}, we can show that
    \begin{align}
        \E \left[ \left\vert \cC^{\mathrm{var}}_T \right\vert \right] \ge \, & N \left( 1 - \frac{1}{K_1} \right) - \E \left[ \left\vert \cC^{\mathrm{disc}}_T \right\vert \right] \ge \left( 1 - \frac{1}{K_1} - \frac{1}{K_2} - o(1) \right) N \\
        \ge \, & \left( \frac{1}{K_3} + 0.01 - o(1) \right) N \ge \left( \frac{1}{K_3} + 0.009 \right) N,
    \end{align}
    since $1/K_1 + 1/K_2 + 1/K_3 \le 0.999$ by our assumption. Using Markov's inequality, we deduce that $\left\vert \cC^{\mathrm{var}}_T \right\vert \ge N / K_3$ with probability at least $0.0009/(1 - 1/K_3)$. Since we already know that part (a) holds with probability at least $1 - 1 / \poly(M, N)$, it follows that $\theta_T$ satisfies both (a) and (b) with probability at least $0.0008/(1 - 1/K_3)$, completing the proof.
\end{proof}

\paragraph{Choice of Parameters for $\alpha = 0.05$.} 
In the case $\alpha \le 0.05$, we have chosen the following parameters via computer grid search:
\begin{align*}
    &\kappa_0 = 3.42, \; K_1 = 4.2, \; K_2 = 30, \; K_3 = 1.3745,\\
    &  c^{(1)} = 6.440850, \; c^{(2)} = 7.184010, \; c^{(3)} = 7.864960, \; c^{(4)} = 8.496780, \; c^{(5)} = 9.088580, \; \\
    &c^{(6)} = 9.646970, \; c^{(7)} = 10.176940, \; c^{(8)} = 10.682360, \; c^{(9)} = 11.166300, \; c^{(10)} = 11.631250, \;\\ 
    &c^{(11)} = 12.079250, \; c^{(12)} = 12.512010, \; c^{(13)} = 12.930960, \; c^{(14)} = 13.337330, \; c^{(15)} = 13.732170, \; \\
    &c^{(16)} = 14.116410, \; c^{(17)} = 14.490850, \; c^{(18)} = 14.856160, \; c^{(19)} = 15.212970, \; c^{(20)} = 15.562490, \;\\
    & c^{(k)} = k^2/20, \; \forall k \geq 21.
\end{align*}

\begin{thm}\label{thm:fullcoloring0}
    For $\kappa = 0$, our two-step algorithm (with properly tuned parameters) finds a solution when $\alpha \leq 0.05$ with high probability.
\end{thm}
\begin{proof}[Proof of \cref{thm:fullcoloring0}]
    By monotonicity, we restrict our attention to the case when $\alpha = 0.05$. We choose our parameters as above. We check that $0<\alpha<\frac{2}{\pi} \mathbb{E}\left[(\kappa_0-G)_{+}^2\right]^{-1}$. Therefore, by \cref{lem:feasibility_LP}, \cref{eq:abs_bound_alpha} is fulfilled. Next, we use \cref{def:LP_auxiliary} and \cref{lem:property_maxmin} to solve for $t$ and $\rho$ and compute that the right hand side of $\cref{eq:tight_fraction}$ satisfies $2\Phi(-t(\alpha,\kappa_0)) > 0.9498$. Define $R^{(0)}_{\text{tight}} = 0.9498$. 

    We next apply the partial coloring lemma \cref{thm:mainpartialcolor_zero}. More specifically, for round $k$, we choose 
    \begin{align*}
        c_j^{(k)} = c^{(k)} \frac{\sqrt{|I_{k-1}|}}{ \| (X_j)_{I_{k-1}} \|_2}.
    \end{align*}
    Note that the fraction is a technicality to deal with the fluctuation of the norm of $(X_j)_{I_{k-1}}$. For each round $k$, we check that
    \begin{equation}\label{eq:checkpartialcolor}
        2 \alpha \Phi\left (- \frac{c^{(k)}}{\sqrt{K_1}} \right ) < (1-R^{(0)}_{\text{tight}}) (1-1/K_3)^{k-1} /K_2.
    \end{equation}
    The first $20$ rounds is checked numerically and the later rounds hold because the left hand side decays much faster than the right hand side. When $k > 0.1 \log N$, we actually have a stronger control,
    \begin{equation}\label{eq:checkpartialcolor2}
        2 \alpha \Phi \left ( -\frac{c^{(k)}}{100 (k+1)\sqrt{K_1} } \right) < (1-R^{(0)}_{\text{tight}}) (1-1/K_3)^{k-1} /K_2.
    \end{equation}
    These would imply that the conditions of the partial coloring lemma hold. Indeed, we can assume with loss of generality that $N_{k-1} = \lceil N (1-R^{(0)}_{\text{tight}}) (1-1/K_3)^{k-1} \rceil$. Now by \cref{lem:normbound} and \cref{eq:checkpartialcolor},
    \begin{align*}
        2\sum_{j=1}^M \Phi(-c_j^{(k)}/\sqrt{K_1}) = 2\sum_{j=1}^M \Phi \left ( -\frac{c^{(k)}}{\sqrt{K_1}} \frac{\sqrt{|I_{k-1}|}}{ \| (X_j)_{I_{k-1}} \|_2} \right ) \leq 2 (1+o_N(1)) \alpha N \Phi \left ( -\frac{c^{(k)}}{\sqrt{K_1}} \right ) < \frac{N_{k-1}}{K_2},
    \end{align*}
    with high probability for all $k \leq 0.1 \log N$. When $k > 0.1 \log N$, by \cref{lem:LMboundall} and \cref{eq:checkpartialcolor2},
    \begin{align*}
        2\sum_{j=1}^M \Phi(-c_j^{(k)}/\sqrt{K_1}) = 2\sum_{j=1}^M \Phi \left ( -\frac{c^{(k)}}{\sqrt{K_1}} \frac{\sqrt{|I_{k-1}|}}{ \| (X_j)_{I_{k-1}} \|_2} \right ) \leq 2 \alpha N\Phi \left ( -\frac{c^{(k)}}{100 (k+1)\sqrt{K_1} } \right) < \frac{N_{k-1}}{K_2},
    \end{align*}
    with high probability. Now apply \cref{thm:mainpartialcolor_zero}, we have that the margin for each row is bounded by 
    \begin{align*}
        \frac{1}{\sqrt{N}}|\<\hat \theta^{(K)}-\hat \theta^{(0)},X_j\>|
        & \leq \frac{1}{\sqrt{N}}\sum_{k = 1}^\infty c_j^{(k)} \|(X_j)_{I_{k-1}}\|_2 \leq \sum_{k = 1}^\infty c^{(k)} \sqrt{(1-R^{(0)}_{\text{tight}}) (1-1/K_3)^{k-1}}\\
        &< \sum_{k = 1}^{20} c^{(k)} \sqrt{(1-R^{(0)}_{\text{tight}}) (1-1/K_3)^{k-1}} + \sum_{k = 21}^\infty c^{(k)} \sqrt{(1-R^{(0)}_{\text{tight}}) (1-1/K_3)^{k-1}}\\
        &<3.38+0.01 = 3.39.
    \end{align*}
    This means that before iterations of the edge walk algorithm, all rows have a margin $\frac{1}{\sqrt{N}}\<\hat \theta^{(0)}, X_j\> \geq \kappa_0 = 3.42$, and after the edge walk steps, each inner product decreases by at most $3.39$. Therefore, every inner product remains above $0.03$ after the iterations of edge walk steps. Finally, we study the randomized sign rounding step. As $\delta = o(1/\log N)$, by \cref{eq:randomsignrounding} we have that the change in the inner product $\<\chi,X_j\> - \<\hat{\theta}^{(K)},X_j\> \geq -o(\sqrt{N})$. Thus, we have that $\<\chi,X_j\> > 0$ for all $j \in [M]$, which finishes the proof.
\end{proof}

\subsubsection{Large Positive Margin}
In this section, we investigate the case $\kappa \to + \infty$, and show that the algorithmic threshold is asymptotic to $(2/\pi) \kappa^{-2}$.
\begin{thm}\label{thm:alg_kappa_positive}
	For any constant $c < 2 / \pi$ and sufficiently large $\kappa > 0$, as long as $\alpha \le c / \kappa^2$, our two-step algorithm (with properly tuned parameters) finds a $\kappa-$margin solution with high probability.
\end{thm}
\begin{proof}
	We only provide a sketch here as the proof idea is mostly similar to that of \cref{thm:kappaverynegative}.
By monotonicity, it suffices to consider the case $\alpha = c/\kappa^2$.
Define $\kappa_0 = \kappa + 1$.
Then, we have
	\begin{equation}
		\alpha = \frac{c}{\kappa^2} = \frac{c + o_{\kappa} (1)}{\kappa_0^2}.
	\end{equation}
	In the first stage of our algorithm, we solve the linear program \eqref{eq:LP_def} with $\kappa$ replaced by $\kappa_0$.
Applying \cref{thm:LP_behavior} and \cref{lem:large_neg_kappa}, we know that the solution $\hat{\theta}$ satisfies $\langle X_i, \hat{\theta} \rangle \ge \kappa_0 \sqrt{N}$ for all $i \in [M]$, and that with high probability,
	\begin{align}\label{eq:pos_kappa_guarantee}
		\frac{1}{N} \left\vert \left\{ i \in [N] : \hat{\theta}_i \in \{ \pm 1 \} \right\} \right\vert \ge \, & 2 \Phi \left( -t \right) \ge 1 - \sqrt{\frac{2}{\pi}} t \ge 1 - \frac{c}{2 \kappa^2}
	\end{align}
	for large enough $\kappa$.
Now for the second stage of iterative Lovett-Meka, the initialization $\hat{\theta}_0 = \hat{\theta}$ satisfies $\langle X_i, \hat{\theta}_0 \rangle \ge (\kappa + 1) \sqrt{N}$ for all $i \in [M]$, so we only need to achieve a margin of $- \sqrt{N}$ on the sub-matrix $(X_{ij})_{(i, j) \in [M] \times J_0}$, where $\vert J_0 \vert = [(c/2\kappa^2)N]$ and
	\begin{equation}
		(\hat{\theta}_0)_j = \pm 1, \quad \forall j \in [N] \backslash J_0.
	\end{equation}
	The existence of such $J_0$ is guaranteed by \cref{eq:pos_kappa_guarantee}.
Note that for the second stage of our algorithm, the new asymptotic aspect ratio equals
	\begin{equation*}
		\alpha_0 = \lim_{M, N \to \infty} \frac{M}{\vert J_0 \vert} = \lim_{N \to \infty} \frac{[(c/\kappa^2) N]}{[(c / 2 \kappa^2) N]} = 2,
	\end{equation*}
	a constant that does not depend on $\kappa$.
Then, \cref{thm:fullcolor} and \cref{lem:LMboundall} together imply that with high probability, iteratively running Lovett-Meka yields a margin asymptotic to
	\begin{equation*}
		- C_0 \sqrt{\vert J_0 \vert} \ge - \frac{C_0}{\kappa} \sqrt{\frac{c}{2}} \sqrt{N} \ge - \sqrt{N}
	\end{equation*}
	for sufficiently large $\kappa$, where $C_0$ is some absolute constant.
Therefore, the final output of our two-stage algorithm is a $\kappa$-margin solution, completing the proof of \cref{thm:alg_kappa_positive}.
\end{proof}

\subsection{Auxiliary Lemmas}\label{sec:aux_lemmas}
In this section, we present some auxiliary lemmas required for the proof of our main results in \cref{sec:alg}. We first collect below a few properties of the functions defined in \cref{def:LP_auxiliary}.
\begin{lem}\label{lem:property_phi}
	Let $\varphi(\rho)$ and $t(\rho)$ be as defined in \cref{def:LP_auxiliary}.
Then, $t(\rho)$ exists uniquely, $\varphi (\rho)$ is differentiable, and $\varphi'(\rho) = \rho t(\rho)$.
Further, for any $\rho \in [0, 1]$, $t(\rho)$ is the unique solution to the following equation:
	\begin{equation}\label{eq:foc_t_rho}
		\rho^2 = \E \left[ \min \left( \frac{\vert G \vert}{t}, 1 \right)^2 \right], \quad t \in [0, +\infty].
	\end{equation} 
\end{lem}
\begin{proof}
	Define for $\rho \in [0, 1]$ and $t \ge 0$,
	\begin{equation}
		f_{\rho} (t) = \frac{\rho^2 t}{2} + \frac{1}{2 t} \E \left[ \min \left( \vert G \vert, t \right)^2 \right] + \E \left[ \left( \vert G \vert - t \right)_+ \right].
	\end{equation}
	By direct calculation, we have
	\begin{equation}
		f_{\rho}'(t) = \frac{\rho^2}{2} - \frac{1}{2} \E \left[ \min \left( \frac{\vert G \vert}{t}, 1 \right)^2 \right],
	\end{equation}
	which is strictly increasing in $t$.
Therefore, $t(\rho)$ exists uniquely and satisfies \cref{eq:foc_t_rho}.
Further, $t(1) = 0$ and $t(0) = +\infty$.
Finally, by envelope theorem we know that $\varphi' (\rho) = \rho t(\rho)$.
\end{proof}

\begin{lem}\label{lem:equivalent_phi}
	$\varphi(\rho)$ has the following equivalent formulation:
	\begin{equation}\label{eq:equivalent_phi}
		\varphi(\rho) = \, \inf_{t \ge 0} \left\{ \rho \sqrt{\E \left[ \min \left( \vert G \vert, t \right)^2 \right]} + \E \left[ \left( \vert G \vert - t \right)_+ \right] \right\}.
	\end{equation}
	Consequently, $\varphi(\rho)$ is concave in $\rho$.
\end{lem}
\begin{proof}
	By \cref{lem:property_phi}, we know that
	\begin{equation*}
		\frac{\rho^2 t(\rho)}{2} = \frac{1}{2 t(\rho)} \E \left[ \min \left( \vert G \vert, t(\rho) \right)^2 \right] = \frac{\rho}{2} \sqrt{\E \left[ \min \left( \vert G \vert, t(\rho) \right)^2 \right]}.
	\end{equation*}
	Therefore,
	\begin{align*}
		\varphi(\rho) = f(\rho, t(\rho)) = \, & \rho \sqrt{\E \left[ \min \left( \vert G \vert, t(\rho) \right)^2 \right]} + \E \left[ \left( \vert G \vert - t(\rho) \right)_+ \right] \\
		\ge \, & \inf_{t \ge 0} \left\{ \rho \sqrt{\E \left[ \min \left( \vert G \vert, t \right)^2 \right]} + \E \left[ \left( \vert G \vert - t \right)_+ \right] \right\}.
	\end{align*}
	To prove the ``$\le$" direction, recall that
	\begin{equation}
		\varphi(\rho) = \min_{A \ge 0} \left\{ \rho \sqrt{\E \left[ \left( \vert G \vert - A \right)_+^2 \right]} + \E[A] \right\}.
	\end{equation}
	The desired inequality then follows by taking $A = (\vert G \vert - t)_+$ and minimizing over $t \ge 0$.
Further, $\varphi$ is concave since it is the pointwise minimum of a family of affine functions.
\end{proof}

\begin{lem}\label{lem:property_maxmin}
	For any $\alpha > 0$ satisfies \cref{eq:abs_bound_alpha}, the maximin problem \cref{eq:key_max_min} has a unique solution $(\rho(\alpha, \kappa), \gamma(\alpha, \kappa))$.
Further, $\rho(\alpha, \kappa)$ solves the following equation:
	\begin{equation}\label{eq:foc_rho}
		\alpha \Phi \left( \frac{\kappa}{\rho} \right) = \frac{\varphi(\rho) \varphi'(\rho)}{\rho}.
	\end{equation}
\end{lem}
\begin{proof} 
	Note that \cref{eq:abs_bound_alpha} implies the existence of $\rho \in [0, 1]$ satisfying
	\begin{equation*}
		\varphi(\rho)^2 > \alpha \E \left[ \left( \kappa - \rho G \right)_+^2 \right].
	\end{equation*}
	For any such $\rho$, the inner minimum of \cref{eq:key_max_min} equals
	\begin{equation*}
		\sqrt{\varphi(\rho)^2 - \alpha \E \left[ \left( \kappa - \rho G \right)_+^2 \right]},
	\end{equation*}
	and is achieved at a unique $\gamma = \gamma(\rho)$.
Otherwise, the minimum is $-\infty$.
We further note that this minimum is a concave function of $\rho$ since it is the pointwise minimum of a family of concave functions ($\rho \mapsto \E [ ( \kappa - \rho G )_+^2 ]^{1/2}$ is obviously convex, and $\varphi(\rho)$ is concave due to \cref{lem:equivalent_phi}).
Therefore, its maximizer is unique.
This proves the uniqueness of $\rho(\alpha, \kappa)$ and $\gamma(\alpha, \kappa)$.
	
	Further, \cref{eq:key_max_min} is equivalent to maximizing
	\begin{equation}
		F(\rho) = \varphi(\rho)^2 - \alpha \E \left[ \left( \kappa - \rho G \right)_+^2 \right].
	\end{equation}
	By direct calculation, we get
	\begin{align*}
		F'(\rho) = \, & 2 \varphi(\rho) \varphi'(\rho) + 2 \alpha \E \left[ G(\kappa - \rho G)_+ \right] \\
		= \, & 2 \left( \varphi(\rho) \varphi'(\rho) - \rho \alpha \P \left( \rho G \le \kappa \right) \right).
	\end{align*}
	Setting $F'(\rho) = 0$ yields \cref{eq:foc_rho}.
\end{proof}

\begin{lem}\label{lem:feasibility_LP}
	The followings are true regarding \cref{eq:abs_bound_alpha}:
	\begin{itemize}
		\item [(a)] If $\kappa < 0$, then \cref{eq:abs_bound_alpha} holds for all $\alpha \in (0, \infty)$.
		\item [(b)] If $\kappa = 0$, then \cref{eq:abs_bound_alpha} holds for all $\alpha \in (0, 2)$.
		\item [(b)] If $\kappa > 0$, then \cref{eq:abs_bound_alpha} holds for all $\alpha$ such that
			\begin{equation}
				0 < \alpha < \frac{2}{\pi} \E \left[ ( \kappa - G )_+^2 \right]^{-1}.
			\end{equation}
	\end{itemize}
\end{lem}
\begin{proof}
	We first prove part (a).
It suffices to show that
	\begin{equation}
		\lim_{\rho \to 0^+} \frac{\varphi(\rho)^2}{\E \left[ \left( \kappa - \rho G \right)_+^2 \right]} = +\infty.
	\end{equation}
	To this end, note that
	\begin{equation}
		\frac{\varphi(\rho)^2}{\E \left[ \left( \kappa - \rho G \right)_+^2 \right]} = \frac{ (\varphi(\rho) / \rho)^2}{\E \left[ \left( \kappa / \rho - G \right)_+^2 \right]}.
	\end{equation}
	According to \cref{lem:equivalent_phi}, $\varphi (\rho)$ is a concave function.
Therefore, $\varphi(\rho) / \rho \ge \varphi(1)$ for all $\rho \in [0, 1]$.
Consequently, we have
	\begin{align}
		& \liminf_{\rho \to 0^+} \frac{\varphi(\rho)^2}{\E \left[ \left( \kappa - \rho G \right)_+^2 \right]} = \, \liminf_{\rho \to 0^+} \frac{ (\varphi(\rho) / \rho)^2}{\E \left[ \left( \kappa / \rho - G \right)_+^2 \right]} \\
		\ge \, & \liminf_{\rho \to 0^+} \frac{\varphi(1)^2}{\E \left[ \left( \kappa / \rho - G \right)_+^2 \right]} = \liminf_{\rho \to 0^+} \frac{2 / \pi}{\E \left[ \left( \kappa / \rho - G \right)_+^2 \right]} = + \infty,
	\end{align}
	since $\kappa / \rho \to -\infty$ as $\rho \to 0^+$.
In the above display we also use the fact $\varphi (1) = \sqrt{2 / \pi}$, which follows from direct calculation.
This proves part (a).
	
	As for part (b), note that when $\kappa = 0$, we have
	\begin{align*}
		\sup_{\rho \in [0, 1]} \left\{ \frac{\varphi(\rho)^2}{\E \left[ \left( \kappa - \rho G \right)_+^2 \right]} \right\} = \, \sup_{\rho \in [0, 1]} \left\{ \frac{ (\varphi(\rho) / \rho)^2}{\E [ G_+^2 ]} \right\} = 2 \sup_{\rho \in [0, 1]} \frac{\varphi(\rho)^2}{\rho^2} = 2 \varphi'(0)^2.
	\end{align*}
	By \cref{lem:equivalent_phi} and envelope theorem, we have
	\begin{equation}
		\varphi' (0) = \sqrt{\E \left[ \min \left( \vert G \vert, t(0) \right)^2 \right]} = \sqrt{\E \left[ \vert G \vert^2 \right]} = 1
	\end{equation}
	since $t(0) = + \infty$.
Therefore,
	\begin{equation}
		\sup_{\rho \in [0, 1]} \left\{ \frac{\varphi(\rho)^2}{\E \left[ \left( \kappa - \rho G \right)_+^2 \right]} \right\} = \, 2 \varphi'(0)^2 = 2,
	\end{equation}
	which completes the proof of part (b).
	
	Finally, the conclusion of part (c) is a direct consequence of the following calculation:
	\begin{align}
		\sup_{\rho \in [0, 1]} \left\{ \frac{\varphi(\rho)^2}{\E \left[ \left( \kappa - \rho G \right)_+^2 \right]} \right\} \ge \, \frac{\varphi(1)^2}{\E \left[ \left( \kappa - G \right)_+^2 \right]} = \frac{2}{\pi} \E \left[ ( \kappa - G )_+^2 \right]^{-1}.
	\end{align}
	This concludes the proof of \cref{lem:feasibility_LP}.
\end{proof}

\begin{lem}\label{lem:rho_and_t}
	Recall that $t(\rho)$ is the unique solution to \cref{eq:foc_t_rho}.
As $\rho \to 1^{-}$, we have
	\begin{equation}\label{eq:rho_and_t}
		\lim_{\rho \to 1^{-}} \frac{t(\rho)}{1 - \rho} = 3 \sqrt{\frac{\pi}{2}}.
	\end{equation}
\end{lem}
\begin{proof}
	First, we note that according to \cref{eq:foc_t_rho}, $t(\rho) \to 0$ as $\rho \to 1$.
Further, we have
	\begin{align*}
		(1 + \rho) (1 - \rho) = \, & 1 - \rho^2 = \E \left[ 1 - \min \left( \frac{\vert G \vert}{t}, 1 \right)^2 \right] = \, \E \left[ \left( 1 - \frac{\vert G \vert^2}{t^2} \right) 1_{\vert G \vert \le t} \right] \\
		= \, & \int_{-t}^{t} \left(1 - \frac{x^2}{t^2} \right) \phi(x) \d x = t \int_{-1}^{1} \left( 1 - u^2 \right) \phi(tu) \d u \\
		= \, & t \phi(0) \int_{-1}^{1} (1 - u^2) \d u + O(t^2) = \frac{2}{3} \sqrt{\frac{2}{\pi}} t + O(t^2),
	\end{align*}
	which implies that
	\begin{align}
		& 2 (1 - \rho) + O \left( (1 - \rho)^2 \right) = \frac{2}{3} \sqrt{\frac{2}{\pi}} t(\rho) + O \left( t(\rho)^2 \right) \\
		\implies & t(\rho) = 3 \sqrt{\frac{\pi}{2}} (1 - \rho) + O \left( (1 - \rho)^2 \right).
 	\end{align}
 	This proves \cref{eq:rho_and_t}.
\end{proof}

\begin{lem}\label{lem:large_neg_kappa}
	Assume for some constant $c > 0$ that does not depend on $\kappa$, $\alpha = (c + o_{\kappa}(1))/\kappa^2 \Phi(\kappa)$ as $\kappa \to - \infty$.
Then, \cref{eq:abs_bound_alpha} holds and we have
	\begin{equation}\label{eq:asymp_rho_t_neg_kappa}
		\rho = 1 - \frac{b + o_{\kappa} (1)}{\kappa^2}, \quad t = \left( 3 \sqrt{\frac{\pi}{2}} b + o_{\kappa} (1) \right) \frac{1}{\kappa^2},
	\end{equation}
	where $b = b(c)$ solves the equation $3 b e^b = c$.
\end{lem}
\begin{proof}
	First, according to \cref{lem:feasibility_LP} (a) we know that \cref{eq:abs_bound_alpha} holds.
We now prove \cref{eq:asymp_rho_t_neg_kappa}.
To this end, note that \cref{eq:foc_rho} implies that
	\begin{align}
		\frac{c + o_{\kappa}(1)}{\kappa^2} = \, \alpha \Phi(\kappa) \stackrel{(i)}{\ge} \alpha \Phi \left( \frac{\kappa}{\rho} \right) = \frac{\varphi(\rho) \varphi'(\rho)}{\rho} \stackrel{(ii)}{\ge} \varphi(1) \varphi' (\rho) = \sqrt{\frac{2}{\pi}} \varphi'(\rho),
	\end{align}
	where $(i)$ is due to $\rho \in [0, 1]$, and $(ii)$ follows from the fact that $\varphi$ is concave.
Since $\varphi'$ decreases from $1$ to $0$ as $\rho$ increases from $0$ to $1$, we must have $\rho = 1 - o_{\kappa} (1)$.
Further, from \cref{lem:property_phi} and \cref{lem:rho_and_t} we know that $\varphi'(\rho) = \rho t(\rho) \sim 1 - \rho$, which further implies that $1 - \rho = O_{\kappa} (1 / \kappa^2)$.
We denote $\rho = 1 - b_{\kappa} / \kappa^2$, it suffices to show that $b_{\kappa} = b + o_{\kappa} (1)$ where $3 b e^b = c$.
Note that on the one hand, we have
	\begin{align*}
		\alpha \Phi \left( \frac{\kappa}{\rho} \right) = \, & \frac{c + o_{\kappa} (1)}{\kappa^2 \Phi(\kappa)} \Phi \left( \frac{\kappa}{\rho} \right) = \frac{c + o_{\kappa} (1)}{\kappa^2 \Phi(\kappa)} \Phi \left( \kappa + \frac{b_{\kappa} + o_{\kappa} (1)}{\kappa} \right) \\
		\stackrel{(i)}{=} \, & \left( 1 + o_{\kappa} (1) \right) \frac{c}{\kappa^2} \exp \left( \frac{\kappa^2}{2} - \frac{1}{2} \left( \kappa + \frac{b_{\kappa} + o_{\kappa} (1)}{\kappa} \right)^2 \right) \\
		= \, & \left( 1 + o_{\kappa} (1) \right) \frac{c}{\kappa^2} \exp \left( -b_{\kappa} \right),
	\end{align*}
	where in $(i)$ we use Mills ratio $\Phi (-t) \sim \phi(t)/t$ as $t \to + \infty$.
On the other hand,
	\begin{align*}
		\frac{\varphi(\rho) \varphi'(\rho)}{\rho} = \varphi(\rho) t(\rho) = \left( 1 + o_{\kappa} (1) \right) \sqrt{\frac{2}{\pi}} \times 3 \sqrt{\frac{\pi}{2}} (1 - \rho) = \left( 1 + o_{\kappa} (1) \right) \frac{3 b_{\kappa}}{\kappa^2}.
	\end{align*}
	Combining the above estimates, we obtain that
	\begin{equation*}
		3 b_{\kappa} \exp(b_{\kappa}) = (1 + o_{\kappa} (1)) c,
	\end{equation*}
	which implies that $b_{\kappa} = b + o_{\kappa} (1)$ with $3 b e^b = c$.
Applying \cref{lem:rho_and_t}, we get
	\begin{equation*}
		t = \left( 3 \sqrt{\frac{\pi}{2}} b_{\kappa} + o_{\kappa} (1) \right) \frac{1}{\kappa^2} = \left( 3 \sqrt{\frac{\pi}{2}} b + o_{\kappa} (1) \right) \frac{1}{\kappa^2}.
	\end{equation*}
	This completes the proof.
\end{proof}

\begin{lem}\label{lem:large_pos_kappa}
	Assume for some constant $c \in (0, 2/\pi)$ that does not depend on $\kappa$, $\alpha = (c + o_{\kappa}(1))/\kappa^2$ as $\kappa \to + \infty$.
Then, \cref{eq:abs_bound_alpha} holds and we have
	\begin{equation}\label{eq:asymp_rho_t_pos_kappa}
		\rho = 1 - \frac{c + o_{\kappa} (1)}{3 \kappa^2}, \quad t = \left( \sqrt{\frac{\pi}{2}} c + o_{\kappa} (1) \right) \frac{1}{\kappa^2}.
	\end{equation}
\end{lem}
\begin{proof}
	First, note that as $\kappa \to + \infty$, we have $\alpha \Phi (\kappa / \rho) \le \alpha = O_{\kappa} (1 / \kappa^2)$.
Using a similar argument as in the proof of \cref{lem:large_neg_kappa}, we know that $1 - \rho = O_{\kappa} (1 / \kappa^2)$.
We thus deduce that
	\begin{align*}
		\frac{c + o_{\kappa} (1)}{\kappa^2} = \alpha = \left( 1 + o_{\kappa} (1) \right) \alpha \Phi \left( \frac{\kappa}{\rho} \right) = \frac{\varphi(\rho) \varphi'(\rho)}{\rho} = (3 +o_{\kappa} (1)) (1 - \rho).
	\end{align*}
	This immediately implies \cref{eq:asymp_rho_t_pos_kappa} and completes the proof of the lemma.
\end{proof}

We need the following variant of Gordon's comparison theorem for Gaussian processes \cite{gordon1985some, thrampoulidis2015regularized, miolane2021distribution}:
\begin{lem}\label{lem:gordon}
	Let $\cD_{u} \subset \mathbb{R}^{n_{1}+n_{2}}$ and $\cD_{v} \subset \mathbb{R}^{m_{1}+m_{2}}$ be compact sets and let $Q: \cD_{u} \times \cD_{v} \rightarrow \mathbb{R}$ be a continuous function.
Let $G = (G_{i, j} ) \stackrel{\iid}{\sim} \normal (0,1)$, $g \sim \sN (0, I_{n_{1}} )$ and $h \sim \sN (0, I_{m_{1}} )$ be independent standard Gaussian vectors.
For $u \in \mathbb{R}^{n_{1}+n_{2}}$ and $v \in \mathbb{R}^{m_{1}+m_{2}}$, we define $\tilde{u}=\left(u_{1}, \ldots, u_{n_{1}}\right)$ and $\tilde{v}=\left(v_{1}, \ldots, v_{m_{1}}\right)$.
Define
	\begin{align}
		C^{*}(G) = \, & \min_{u \in \cD_{u}} \max_{v \in \cD_{v}} \left\{ \tilde{v}^{\top} G \tilde{u}+Q(u, v) \right\}, \\
		L^{*}(g, h) = \, & \min_{u \in \cD_{u}} \max_{v \in \cD_{v}} \left\{ \norm{\tilde{v}}_2 g^{\top} \tilde{u} + \norm{\tilde{u}}_2 h^{\top} \tilde{v} + Q(u, v) \right\}.
	\end{align}
	Then we have:
	\begin{itemize}
		\item [(i)] For all $t \in \mathbb{R}$,
		\begin{equation*}
			\mathbb{P}\left(C^{*}(G) \leq t \right) \leq 2 \mathbb{P}\left(L^{*}(g, h) \leq t \right).
		\end{equation*}
		\item [(ii)] If $\cD_{u}$ and $\cD_{v}$ are convex and if $Q$ is convex-concave, then for all $t \in \mathbb{R}$,
		\begin{equation*}
			\mathbb{P}\left(C^{*}(G) \geq t \right) \leq 2 \mathbb{P}\left(L^{*}(g, h) \geq t \right) .
		\end{equation*} 
	\end{itemize}
	Further, $(i)$ and $(ii)$ hold for general non-compact $\cD_{u}$ if the minimum over $\cD_{u}$ is achieved.
\end{lem}

\begin{proof}
	Points $(i)$ and $(ii)$ are just \cite[Corollary G.1]{miolane2021distribution}.
Generalization to a non-compact $\cD_{u}$ is achieved by exhausting $\cD_{u}$ with a sequence of compact domains and noting that the minimum is attained within all sufficient large compact exhaustions.
\end{proof}

\restatelemma{lem:calculus_1}
\begin{proof}
	By direct calculation,
	\begin{align*}
		& \min_{A \ge 0} \left\{ \frac{\rho}{\sqrt{N}} \norm{\left( y - A \right)_+}_2 + \frac{1}{N}\sum_{i=1}^{N} A_i \right\} \\
		= \, & \min_{A \ge 0} \min_{t \ge 0} \left\{ \frac{\rho t}{2} + \frac{\rho}{2 t N} \sum_{i=1}^{N} \left( y_i - A_i \right)_+^2 + \frac{1}{N}\sum_{i=1}^{N} A_i \right\} \\
		= \, & \min_{t \ge 0} \left\{ \frac{\rho t}{2} + \min_{A \ge 0} \left\{ \frac{\rho}{2 t N} \sum_{i=1}^{N} \left( y_i - A_i \right)_+^2 + \frac{1}{N}\sum_{i=1}^{N} A_i \right\} \right\} \\
		= \, & \min_{t \ge 0} \left\{ \frac{\rho t}{2} + \frac{1}{N} \sum_{i=1}^{N} \min_{A \ge 0} \left\{ \frac{\rho}{2 t} \left( y_i - A \right)_+^2 + A \right\} \right\}.
	\end{align*}
	It is straightforward to check that for any $y \ge 0$:
	\begin{align}\label{eq:varphi_minimizer}
		\arg\min_{A \ge 0} \left\{ \frac{\rho}{2 t} \left( y - A \right)_+^2 + A \right\} = \, & \left( y - \frac{t}{\rho} \right)_+, \\
		\label{eq:varphi_minimum}
		\min_{A \ge 0} \left\{ \frac{\rho}{2 t} \left( y - A \right)_+^2 + A \right\} = \, & \frac{\rho}{2 t} \min \left( y, \frac{t}{\rho} \right)^2 + \left( y - \frac{t}{\rho} \right)_+.
	\end{align}
	Consequently, we have
	\begin{align*}
		& \min_{A \ge 0} \left\{ \frac{\rho}{\sqrt{N}} \norm{\left( y - A \right)_+}_2 + \frac{1}{N}\sum_{i=1}^{N} A_i \right\} \\
		= \, & \min_{t \ge 0} \left\{ \frac{\rho t}{2} + \frac{1}{N} \sum_{i=1}^{N} \left( \frac{\rho}{2 t} \min \left( y_i, \frac{t}{\rho} \right)^2 + \left( y_i - \frac{t}{\rho} \right)_+ \right) \right\} \\
		= \, & \min_{t \ge 0} \left\{ \frac{\rho t}{2} + \frac{\rho}{2 t N} \norm{\min \left( y, \frac{t}{\rho} \right)}_2^2 + \frac{1}{N} \sum_{i=1}^{N} \left( y_i - \frac{t}{\rho} \right)_+ \right\} \\
		= \, & \min_{t \ge 0} \left\{ \frac{\rho^2 t}{2} + \frac{1}{2 t N} \norm{\min \left( y, t \right)}_2^2 + \frac{1}{N} \sum_{i=1}^{N} \left( y_i - t \right)_+ \right\},
	\end{align*}
	which completes the proof.
\end{proof}

\begin{lem}\label{lem:ULLN_1}
	Recall $\varphi(\rho)$ from \cref{def:LP_auxiliary}.
For $g \sim \normal(0, I_N)$, define
	\begin{equation}\label{eq:phi_N}
		\varphi_N (\rho; g) = \min_{t \ge 0} \left( \frac{\rho^2 t}{2} + \frac{1}{2 t N} \sum_{i=1}^{N} \min \left( |g_i|, t \right)^2 + \frac{1}{N} \sum_{i=1}^{N} (|g_i| - t)_+ \right).
	\end{equation}
	Then, we have
	\begin{equation}
		\plim_{N \to \infty} \sup_{\rho \in [0, 1]} \left\vert \varphi_N (\rho; g) - \varphi(\rho) \right\vert = 0.
	\end{equation}
\end{lem}
\begin{proof}
	For any $\veps > 0$ and $\rho \in [\veps, 1]$, denote by $t_N (\rho)$ the minimizer of the right hand side of \cref{eq:phi_N}.
Then, we know that
	\begin{equation*}
		\frac{\rho^2 t_N (\rho)}{2} \le \frac{1}{N} \sum_{i=1}^{N} \vert g_i \vert \implies t_N (\rho) \le \frac{2}{\rho^2 N} \sum_{i=1}^{N} \vert g_i \vert \le \frac{2}{\veps^2 N} \sum_{i=1}^{N} \vert g_i \vert.
	\end{equation*}
	By the law of large number, $(1/N) \sum_{i=1}^{N} \vert g_i \vert \to \sqrt{2 / \pi}$ almost surely.
Hence, there exists a constant $C = C_{\veps}$ such that
	\begin{equation*}
		\lim_{N \to \infty} \P \left( \sup_{\rho \in [\veps, 1]}t_N (\rho) \le C_{\veps} \right) \to 1.
	\end{equation*}
	As a consequence, with high probability, we have for all $\rho \in [\veps, 1]$:
	\begin{equation}
		\varphi_N (\rho; g) = \min_{t \in [0, C_{\veps}]} \left( \frac{\rho^2 t}{2} + \frac{1}{2 t N} \sum_{i=1}^{N} \min \left( |g_i|, t \right)^2 + \frac{1}{N} \sum_{i=1}^{N} (|g_i| - t)_+ \right).
	\end{equation}
	Further, for all $\rho \in [\veps, 1]$ we have
	\begin{equation}
		\varphi (\rho) = \min_{t \in [0, C_{\veps}]} \left\{ \frac{\rho^2 t}{2} + \frac{1}{2 t} \E \left[ \min \left( \vert G \vert, t \right)^2 \right] + \E \left[ \left( \vert G \vert - t \right)_+ \right] \right\}.
	\end{equation}
	Since as $t \to 0$,
	\begin{equation*}
		\max \left\{ \frac{1}{2 t N} \sum_{i=1}^{N} \min \left( |g_i|, t \right)^2, \ \frac{1}{2 t} \E \left[ \min \left( \vert G \vert, t \right)^2 \right] \right\} = O(t),
	\end{equation*}
	applying the uniform law of large numbers yields that
	\begin{align}
		\plim_{N \to \infty} \sup_{t \in [0, C_{\veps}]} \left\vert \frac{1}{2 t N} \sum_{i=1}^{N} \min \left( |g_i|, t \right)^2 - \frac{1}{2 t} \E \left[ \min \left( \vert G \vert, t \right)^2 \right] \right\vert = \, & 0, \\
		\plim_{N \to \infty} \sup_{t \in [0, C_{\veps}]} \left\vert \frac{1}{N} \sum_{i=1}^{N} (|g_i| - t)_+ - \E \left[ \left( \vert G \vert - t \right)_+ \right] \right\vert = \, 0.
	\end{align}
	Consequently, we obtain that
	\begin{equation*}
		\plim_{N \to \infty} \sup_{\rho \in [\veps, 1]} \left\vert \varphi_N (\rho; g) - \varphi(\rho) \right\vert = 0.
	\end{equation*}
	Finally, note that if $\rho \in [0, \veps]$, we have (taking $t = 1 / \sqrt{\rho}$)
	\begin{align*}
		\varphi_N (\rho; g) \le \, \frac{\rho}{2} + \frac{\sqrt{\rho}}{2 N} \sum_{i=1}^{N} |g_i|^2 + \frac{1}{N} \sum_{i=1}^{N} \left( |g_i| - \frac{1}{\sqrt{\rho}} \right)_+ = O_p (\sqrt{\rho}) = O_p (\sqrt{\veps})
	\end{align*}
	as $\veps \to 0$, and similarly $\varphi (\rho) = O(\sqrt{\rho}) = O (\sqrt{\veps})$.
This immediately implies that
	\begin{equation}
		\plim_{N \to \infty} \sup_{\rho \in [0, \veps]} \left\vert \varphi_N (\rho; g) - \varphi(\rho) \right\vert = O(\sqrt{\veps}).
	\end{equation}
	We thus deduce that for all $\veps \in [0, 1]$:
	\begin{equation}
		\plim_{N \to \infty} \sup_{\rho \in [0, 1]} \left\vert \varphi_N (\rho; g) - \varphi(\rho) \right\vert = O(\sqrt{\veps}).
	\end{equation}
	Sending $\veps \to 0^+$ concludes the proof.

	\end{proof}

\begin{lem}\label{lem:ULLN_2}
	Let $h \sim \normal(0, I_M)$, we have
	\begin{equation}
		\plim_{M \to \infty} \sup_{\rho \in [0, 1]} \left\vert \frac{1}{\sqrt{M}} \norm{\left( \kappa 1 - \rho h \right)_+}_2 - \E \left[ \left( \kappa - \rho G \right)_+^2 \right]^{1/2} \right\vert = 0.
	\end{equation}
\end{lem}
\begin{proof}
	See the proof of \cite[Lemma 4]{montanari2024tractability}.
\end{proof}

\section{Multi-Overlap Gap Property}\label{sec:mogp}
In this section, we investigate the overlap gap property (OGP) of the Ising Perceptron model. We demonstrate that the model exhibits a multi-Overlap Gap Property (m-OGP), an intricate geometric structure of the solution space, within certain ranges of $\alpha$. Leveraging this property, we can effectively rule out stable algorithms for these parameter ranges.

The first subsection is dedicated to establishing the overlap gap property, while the second subsection focuses on eliminating stable algorithms based on the m-OGP framework.

\subsection{Multi-OGP Structure}\label{sec:ogpstructure}
In this section, we show that the Ising Perceptron model exhibits m-OGP. The proofs follows the same strategy as in \cite{gamarnik2022algorithms}, so we omit some details and only highlight the differences.

We start by recalling an important definition for sets of overlapped solutions.
\begin{defn}[Definition 2.1 in \cite{gamarnik2022algorithms}]
Let $0<\eta<\beta<1, m \in \mathbb{Z}_+$ and $\cI \subset [0,\pi/2]$ be a set of indices. Let $\alpha < \asat (\kappa)$. Let $X_0, \cdots, X_m$ to be a sequences of i.i.d. $M \times N$ matrices with i.i.d. $\normal(0,1)$ entries, where $M/N \to \alpha$. We define $\cS_\kappa(\beta, \eta, m, \alpha, \cI)$ to be the set of all $m$-tuples $\theta^{(i)}$ where $1\leq i \leq m$, $\theta^{(i)} \in \{\pm 1\}^N$ such that the following holds:
\begin{itemize}
    \item (Pairwise Overlap Condition) For any $1\leq i < j \leq m$, we have 
\begin{align*}
    \beta-\eta < N^{-1}\< \theta^{(i)}, \theta^{(j)} \> < \beta.
\end{align*}
    \item (One-sided Margin Constraints) There exists $\tau_i \in \cI$, where $1\leq i \leq m$, such that 
    \begin{align*}
        X_i(\tau_i) \theta^{(i)} \geq \kappa \sqrt{N}, \quad \forall \, 1\leq i \leq m,
    \end{align*}
    where
    \begin{align*}
        X_i(\tau_i) = \cos(\tau_i) X_0 + \sin(\tau_i) X_i.
    \end{align*}
\end{itemize}
\end{defn}

\begin{thm}\label{thm:mogp}
    There exist constants $C_1>0$ and $\kappa_1 < 0$ such that for any $\kappa<\kappa_1$ and $\alpha > \frac{C_1\log^2|\kappa|}{\kappa^2 \Phi(\kappa)}$, there exist $0 < \eta < \beta <1$, $c>0$, and $m \in \mathbb{N}$ such that the following holds. Fix any $\cI \subset [0,\pi/2]$ with $|\cI| \leq \exp(cN)$, we have
    \begin{align*}
        \mathbb{P}[\cS_\kappa(\beta, \eta, m, \alpha, \cI) \neq \varnothing] = \exp(-\Theta(N)).
    \end{align*}
\end{thm}

\begin{proof}
    By monotonicity, it suffices to prove the theorem for $\alpha = \frac{C_1\log^2|\kappa|}{\kappa^2 \Phi(\kappa)}$. We use the first moment method to compute the size of $\cS_\kappa(\beta, \eta, m, \alpha, \cI)$. Note that by Markov's inequality,
    \begin{align*}
        \mathbb{P}[\cS_\kappa(\beta, \eta, m, \alpha, \cI) \neq \varnothing] = \mathbb{P}[\left |\cS_\kappa(\beta, \eta, m, \alpha, \cI) \right | \geq 1] \leq \mathbb{E}[\left |\cS_\kappa(\beta, \eta, m, \alpha, \cI) \right |].
    \end{align*}
    Now $\mathbb{E}[|\cS_\kappa(\beta, \eta, m, \alpha, \cI) |]$ can be decomposed as the product of the number of $\theta^{(i)}$ that satisfy the overlap conditions, and the probability that they are solutions (condition on the overlap structure).

    Define $M(m,\beta, \eta)$ to be the number of $m$-tuples of $\theta^{(i)} \in \{\pm 1\}^N$ where $1\leq i \leq m$ such that $\beta-\eta < N^{-1}\inp{\theta^{(i)}, \theta^{(j)}} < \beta$ for $ 1\leq i < j \leq m$. Then by Stirling's formula,
    \begin{align*}
        M(m,\beta, \eta) &\leq 2^N \left( \sum_{\substack{\rho: \frac{1+\beta-\eta}{2} \leq \rho \leq \frac{1+\beta}{2} \\ \rho N \in \mathbb{N}}} \binom{N}{N \rho} \right)^{m-1} \\
        &\leq \exp \left (N \log 2 + (m-1) N \Ent\left(\frac{1+\beta-\eta}{2} \right) + O(\log N) \right),
    \end{align*}
    where $\Ent(\cdot)$ is the binary entropy function. 

    Now we count the probability term. We note that
    \begin{align*}
        &\;\;\mathbb{P}[\exists \tau_1, \cdots, \tau_m\in \cI: X_i(\tau_i) \theta^{(i)} \geq \kappa \sqrt{N}, \, \forall \, 1\leq i \leq m] \\
        &\leq |\cI|^m \max_{\tau_i \in \cI, 1\leq i \leq m} \mathbb{P}[X_i(\tau_i) \theta^{(i)} \geq \kappa \sqrt{N}, \, \forall \, 1\leq i \leq m]\\
        & \leq |\cI|^m \max_{\tau_i \in \cI, 1\leq i \leq m} \mathbb{P}[\texttt{R}_i(\tau_i) \theta^{(i)} \geq \kappa \sqrt{N}, \, \forall \, 1\leq i \leq m]^M,
    \end{align*}
    where we denote $\texttt{R}_i(\tau_i)$ as the first row of the matrix $X_i(\tau_i)$. We note that $\{N^{-1/2}\texttt{R}_i(\tau_i) \theta^{(i)}: 1\leq i \leq m\}$ is a multivariate normal with each entry having zero mean and unit variance. Furthermore, for $i\neq j$, we have that
    \begin{align*}
        \cov(N^{-1/2}\texttt{R}_i(\tau_i) \theta^{(i)}, N^{-1/2}\texttt{R}_j(\tau_j) \theta^{(j)}) = N^{-1} (\theta^{(i)})^\top \mathbb{E}(\texttt{R}_i \texttt{R}_j^\top) \theta^{(j)} = \cos(\tau_i) \cos(\tau_j) N^{-1} \inp{\theta^{(i)}, \theta^{(j)}}.
    \end{align*}
Note that since $N^{-1}\inp{\theta^{(i)}, \theta^{(j)}} < \beta$, we have that by the inclusion–exclusion principle and \cref{lem:jointprob},
    \begin{align*} 
    \mathbb{P}[\texttt{R}_i(\tau_i) \theta^{(i)} \geq \kappa, \, \forall \, 1\leq i \leq m] &\leq 1 - \sum_{i=1}^m \mathbb{P}[\texttt{R}_i(\tau_i) \theta^{(i)} < \kappa] +\sum_{i\neq j} \mathbb{P}[\texttt{R}_i(\tau_i) \theta^{(i)} < \kappa , \texttt{R}_j(\tau_j) \theta^{(j)} < \kappa ]\\
    & \leq 1- m\Phi(\kappa) + \frac{m(m-1)}{2} 4\Phi(\kappa) \Phi\left(-|\kappa|\sqrt{\frac{1-\beta}{2}}\right).
    \end{align*}
Now, putting together the counting and probability terms, we have that,
\begin{align*}
    &\mathbb{E}[|\cS_\kappa(\beta, \eta, m, \alpha, \cI) |] \leq \exp \left (N \log 2 + (m-1) N \Ent\left(\frac{1+\beta-\eta}{2} \right) + O(\log N) \right) \\
    &\hspace{5cm} \cdot |\cI|^m \left( 1- m\Phi(\kappa) + \frac{m(m-1)}{2} 4\Phi(\kappa) \Phi\left(-|\kappa|\sqrt{\frac{1-\beta}{2}}\right) \right)^M.
\end{align*}
Taking the logarithm, we have that
\begin{align*}
&\log \mathbb{E}[|\cS_\kappa(\beta, \eta, m, \alpha, \cI) |] \leq N \log 2 + m N \Ent\left(\frac{1+\beta-\eta}{2} \right) + m \log|\cI| + o(N)\\
&\hspace{5cm}+ \alpha N \log\left( 1- m\Phi(\kappa) + \frac{m(m-1)}{2} 4\Phi(\kappa) \Phi\left(-|\kappa|\sqrt{\frac{1-\beta}{2}}\right) \right) .
\end{align*}
Note that with $\beta = 1-b_1 \log|\kappa|/\kappa^2$ and $\eta = b_2 \log|\kappa|/\kappa^2$, we have that
\begin{align*}
&\Ent\left(\frac{1+\beta-\eta}{2} \right) = \frac{1+\beta-\eta}{2} \log \left(\frac{2}{1+\beta-\eta} \right) + \frac{1-\beta+\eta}{2} \log \left(\frac{2}{1-\beta+\eta} \right)\\
&= \frac{(b_1+b_2) \log|\kappa|}{2\kappa^2}\left( \frac{2\kappa^2}{(b_1+b_2) \log|\kappa|} \right) + \left( 1- \frac{(b_1+b_2) \log|\kappa|}{2\kappa^2} \right) \log\left( \frac{1}{1-(b_1+b_2) \log|\kappa|/(2\kappa^2)} \right)\\
& \leq \frac{(b_1+b_2) \log|\kappa|}{2\kappa^2}\left( \log(2 \kappa^2) +1  \right) \leq \frac{(b_1+b_2)\log^2|\kappa|}{\kappa^2} (1+O(\frac{1}{\log |\kappa|})).
\end{align*}
Furthermore, we have for large enough $|\kappa|$,
    \begin{align*}
        \Phi\left(-|\kappa|\sqrt{\frac{1-\beta}{2}}\right)= \Phi\left(-\sqrt{\frac{b_1 \log |\kappa|}{2}}\right) \leq \exp(-b_1 \log |\kappa|/4) = |\kappa|^{-b_1/4}.
    \end{align*}
Now we take $m = \lfloor\kappa^2 \rfloor$. Then with $|\cI| \leq \exp(cN)$, we have 
\begin{align*}
    N \log 2 + m N \Ent\left(\frac{1+\beta-\eta}{2} \right) + m \log|\cI| + o(N) &\leq m N \Ent\left(\frac{1+\beta-\eta}{2} \right) (1+o_\kappa(1)) \\
    &\leq mN \frac{(b_1+b_2)\log^2|\kappa|}{\kappa^2} (1+o_\kappa(1)).
\end{align*}
Furthermore, taking $b_1 > 8$, we have
\begin{align*}
     &\log\left( 1- m\Phi(\kappa) + \frac{m(m-1)}{2} 4\Phi(\kappa) \Phi\left(-|\kappa|\sqrt{\frac{1-\beta}{2}}\right) \right) \\
     &\leq \log\left( 1- m \Phi(\kappa) (1+o_\kappa(1)) \right) \leq -m \Phi(\kappa) (1+o_\kappa(1)).
\end{align*}
So as long as
\begin{align*}
    \alpha > \frac{(b_1+b_2) \log^2|\kappa|}{\kappa^2 \Phi(\kappa)} (1+o_\kappa(1)),
\end{align*}
we have 
\begin{align*}
    \log \mathbb{E}[|\cS_\kappa(\beta, \eta, m, \alpha, \cI) |] \leq -N \Theta(\log^2|\kappa|) = -\Theta(N),
\end{align*}
which completes the proof.
\end{proof}

Now we record a useful technical lemma.
\begin{lem}\label{lem:jointprob}
	Assume $(G, G') \sim \normal (0, \begin{bmatrix}
		1 & q \\
		q & 1
	\end{bmatrix})$. Then, for sufficiently negative $A < 0$, we have
    \begin{equation}\label{eq:joint_prob_kappa}
    	\P \left( G \le A, \ G' \le A \right) \le 4 \Phi(A) \left( 1 - \Phi \left( \vert A \vert \sqrt{\frac{1-q}{2}} \right) \right).
    \end{equation}
    Note that the above inequality holds uniformly for all $q \in [0, 1]$.
\end{lem}

\begin{proof}
	For $q \in [0, 1]$, define
	\begin{equation*}
		f_A (q) = \P_q \left( G \le A, \ G' \le A \right) = \E_q \left[ \bone_{G \le A} \bone_{G' \le A} \right],
	\end{equation*}
    where the probability and expectation are taken under $(G, G') \sim \normal (0, \begin{bmatrix}
    	1 & q \\
    	q & 1
    \end{bmatrix})$. Using Gaussian interpolation (see, e.g., \cite[Proposition 1.3.2]{talagrand2010mean}), we obtain that
    \begin{equation*}
    	f_A'(q) = \E_q \left[ \delta (G - A) \delta (G' - A) \right] = \phi_q (A, A) = \frac{1}{2 \pi \sqrt{1 - q^2}} \exp \left( - \frac{A^2}{1 + q} \right),
    \end{equation*}
    where $\delta$ refers to Dirac's delta measure and $\phi_q$ denotes the density of $(G, G')$. Since $f_A (0) = \P (G \le A)^2 = \Phi(A)^2$, it follows that
    \begin{align}
    	f_A (q) - \Phi(A)^2 =\, & \int_{0}^{q} \frac{1}{2 \pi \sqrt{1 - s^2}} \exp \left( - \frac{A^2}{1 + s} \right) \d s \\
    	=\, & \frac{1}{2 \pi} \exp \left( - \frac{A^2}{2} \right) \int_{0}^{q} \frac{1}{\sqrt{1 - s^2}} \exp \left( - \frac{(1-s) A^2}{2 (1 + s)} \right) \d s \\
    	\le\, & \frac{1}{2 \pi} \exp \left( - \frac{A^2}{2} \right) \int_{0}^{q} \frac{1}{\sqrt{1 - s}} \exp \left( - \frac{(1-s) A^2}{4} \right) \d s \\    	
    	\stackrel{(i)}{=}\, & \frac{1}{2 \pi \vert A \vert} \exp \left( - \frac{A^2}{2} \right) \int_{A^2 (1-q)}^{A^2} \frac{1}{\sqrt{u}} \exp \left( - \frac{u}{4} \right) \d u \\
    	\le\, & \frac{1}{2 \pi \vert A \vert} \exp \left( - \frac{A^2}{2} \right) \int_{A^2 (1-q)}^{\infty} \frac{1}{\sqrt{u}} \exp \left( - \frac{u}{4} \right) \d u \\
    	\stackrel{(ii)}{=}\, & \frac{2 \sqrt{2}}{\sqrt{2 \pi} \vert A \vert} \exp \left( - \frac{A^2}{2} \right) \int_{\vert A \vert \sqrt{\frac{1-q}{2}} }^{\infty} \frac{1}{\sqrt{2 \pi}} \exp \left( - \frac{t^2}{2} \right)  \d t \\
    	\le\, & 3 \Phi(A) \left( 1 - \Phi \left( \vert A \vert \sqrt{\frac{1-q}{2}} \right) \right)
    \end{align}
    for sufficiently negative $A$, where $(i)$ follows from the change of variable $u = A^2 (1-s)$, $(ii)$ follows from the change of variable $t = \sqrt{u/2}$, and the last line is due to Mills' ratio. Note that
    \begin{equation*}
    	1 - \Phi \left( \vert A \vert \sqrt{\frac{1-q}{2}} \right) = \Phi \left( A \sqrt{\frac{1-q}{2}} \right) \ge \Phi \left( \frac{A}{\sqrt{2}} \right) \gg \Phi(A),
    \end{equation*}
    which further implies that
    \begin{equation}
    	f_A (q) \le \Phi(A)^2 + 3 \Phi(A) \left( 1 - \Phi \left( \vert A \vert \sqrt{\frac{1-q}{2}} \right) \right) \le 4 \Phi(A) \left( 1 - \Phi \left( \vert A \vert \sqrt{\frac{1-q}{2}} \right) \right).
    \end{equation}
    This completes the proof of \cref{lem:jointprob}.
\end{proof}

\subsection{Failure of Stable Algorithms}\label{sec:fail_stable_alg}
In this subsection, we show that the existence of multi-OGP structure implies failure of stable algorithms for the Binary Perceptron model.

We adopt the notation of stability from \cite{gamarnik2022algorithms}. An algorithm $\mathcal{A}$ is a mapping from $\mathbb{R}^{M \times N}$ to the binary cube $\{\pm 1\}^N$. We allow $\mathcal{A}$ to be potentially randomized: we assume there exists an underlying probability space $\left(\Omega, \mathbb{P}_\omega\right)$ such that $\mathcal{A}: \mathbb{R}^{M \times N} \times \Omega \rightarrow \{\pm 1\}^N$. That is, for any $\omega \in \Omega$ and data matrix $X \in \mathbb{R}^{M \times N} ; \mathcal{A}(\cdot, \omega)$ returns a $\theta_{\mathrm{ALG}} \triangleq \mathcal{A}(X, \omega) \in \{\pm 1\}^N$; and we want $\theta_{\mathrm{ALG}}$ to satisfy $X \theta_{\mathrm{ALG}} \geq \kappa \sqrt{N}$.

\begin{defn}[Stable Algorithm]\label{def:stable}
    Fix $\kappa \in \mathbb{R}$, $\alpha<\asat(\kappa)$ and let $M/N \to \alpha$. An algorithm $\mathcal{A}: \mathbb{R}^{M \times N} \times \Omega \rightarrow$ $\{\pm 1\}^N$ is called $\left(\rho, p_f, p_{\mathrm{st}}, f, L\right)$-stable for the Ising Perceptron model, if it satisfies the following for all sufficiently large $N$.
\begin{itemize}
    \item (Success) Let $X \in \mathbb{R}^{M \times N}$ be a random matrix with i.i.d. $\normal(0,1)$ entries. Then,
\[
\mathbb{P}_{(X, \omega)}\left[X \mathcal{A}(X, \omega) \geq \kappa \sqrt{N}\right] \geq 1-p_f.
\]
\item (Stability) Let $X, \overline{X} \in \mathbb{R}^{M \times N}$ be random matrices, each with i.i.d. $\normal(0,1)$ entries such that $\mathbb{E}\left[X_{i j} \overline{X}_{i j}\right]=\rho$ for $1 \leq i \leq M$ and $1 \leq j \leq N$. Then,
\[
\mathbb{P}_{(X, \overline{X}, \omega)}\left[d_H(\mathcal{A}(X, \omega), \mathcal{A}(\overline{X}, \omega)) \leq f+L\|X-\overline{X}\|_F\right] \geq 1-p_{\mathrm{st}}.
\]
\end{itemize}
\end{defn}

Now we state our algorithmic hardness result.
\begin{thm}\label{thm:hardness}
    There exists constants $C_1>0$, $\kappa_1<0$ and $N_0 \in \mathbb{N}$ such that for any $\kappa < \kappa_1$, $\alpha > \frac{C_1 \log^2|\kappa|}{\kappa^2 \Phi(\kappa)}$, $N > N_0$ and $L>0$, there exists constants $\rho, p_f, p_{\mathrm{st}}, C$ such that the following holds. There exists no randomized algorithm $\mathcal{A}: \mathbb{R}^{M \times N} \rightarrow \{\pm 1 \}^N$ that is $(\rho, p_f, p_{\mathrm{st}}, CN, L)$-stable for the Ising Perceptron model in the sense of \cref{def:stable}.
\end{thm}

Similar to \cite{gamarnik2022algorithms}, we do not impose any constraints on the running time of the stable algorithm. As long as the algorithm satisfies the criteria outlined in \cref{def:stable}, it falls within the scope of our arguments. Moreover, both the success probability and stability probability are assumed to be constant, meaning the algorithm is not required to guarantee a high probability of success. The definition of a stable algorithm is broad. The algorithm $\mathcal{A}$ can produce solutions that differ by $\Theta(N)$ even when the inputs $X$ and $\overline{X}$ are nearly identical.

Most of our proofs are similar to those presented in \cite{gamarnik2022algorithms}, with necessary changes in notation. For the sake of completeness, we outline the key steps and state a few useful lemmas and theorems. The majority of the detailed proofs will be omitted.

The general strategy involves assuming the existence of a stable algorithm $\cA$. This assumption implies the presence of a set of solutions exhibiting the prescribed multi-overlap structure, leading to a contradiction.
\begin{itemize}
    \item \cref{lem:alg_determinimistic} reduces the randomized stable algorithm $\cA$ to a deterministic stable algorithm $\cA^*$.
    \item \cref{lem:chaos} focuses on the \emph{chaos} event. The \emph{chaos} event says that, for any $m$-tuples $\theta^{(i)} \in \{\pm 1 \}^N$, $1\leq i \leq m$, where $X_i \theta^{(i)} \geq \kappa \sqrt{N}$ for i.i.d. random matrices $X_i \in \mathbb{R}^{M \times N}$ each with i.i.d. $\normal (0,1)$ coordinates, there exists $\beta'$ such that with high probability for any such $m$-tuple, there exists a pair $1\leq i < j \leq m$ such that $\<\theta^{(i)},\theta^{(j)}\> \leq \beta' N$, where $\beta'< \beta-\eta$. We establish that the \emph{chaos} event happens with exponentially small probability below the m-OGP threshold.
    \item The next step is to construct interpolation trajectories and show that they are smooth, for some integers $T$ and $Q$. We construct ($T+1$) i.i.d. random matrices $X_i \in \mathbb{R}^{M \times N}$ where $0\leq i \leq T$. Then for an equal partition of the interval $[0,\pi/2]$ such that $0 = \tau_0 < \tau_1 < \cdots < \tau_Q = \pi/2$, we define
    \begin{align*}
        X_i(\tau_k) = \cos(\tau_k) X_0 + \sin(\tau_k) X_i, \quad 1\leq i \leq T, \quad 0 \leq k \leq Q.
    \end{align*}
    And further define $\theta_i(\tau_k) = \cA^*(X_i(\tau_k)).$ \cref{prop:trajectory_smooth} shows that since $\cA^*$ is stable, the overlaps for each pair of solutions evolve smoothly.
    \item \cref{lem:alg_success} shows that the algorithm is successful along each trajectory and time stamps using a union bound.
    \item Combining smoothness of the trajectories and the chaos property, \cref{prop:combine} shows that for every $A \subset [T]$ with $|A| = m$, there exists $1\leq  i_A < j_A \leq T$ and a time $\tau_A$ such that 
    \begin{align*}
        \beta-\eta < N^{-1}\inp{\theta_{i_A}(\tau_A),\theta_{j_A}(\tau_A)} < \beta.
    \end{align*}
    \item We then construct a graph $G = ([T],E)$. Each vertex $i $ of $G$ corresponds to the $i$-th interpolation trajectory. For any $1 \leq i < j \leq T$, $(i,j) \in E$ if and only if there is a time $t \in [Q]$ such that $\beta-\eta < N^{-1}\inp{\theta_{i}(\tau_t),\theta_{j}(\tau_t)} < \beta$. We then color each edge $(i,j) \in E$ with the first time $t \in [Q]$ such that $\beta-\eta < N^{-1}\inp{\theta_{i}(\tau_t),\theta_{j}(\tau_t)} < \beta$.
    \item Using \cref{thm:ramsey1} and \cref{thm:ramsey2}, we apply two-colored and multi-colored Ramsey theory to get a monochromatic clique $K_m$, as shown in \cref{prop:monochromatic}.
    \item Finally, we combine all the arguments and construct the desired forbidden configuration, which had been ruled out by our multi-OGP computation. This implies failure of the stable algorithm.
\end{itemize}

\begin{proof}[Proof of \cref{thm:hardness}]

Let $\kappa< \kappa_1$ and $\frac{C_1\log^2|\kappa|}{\kappa^2 \Phi(\kappa)}<\alpha<\asat$, where the constants $\kappa_1$ and $C_1$ are the same as in the proof of \cref{thm:mogp}. Now suppose a randomized algorithm $\mathcal{A}: \mathbb{R}^{M \times N} \rightarrow \{\pm 1\}^N$ exists that is $(\rho, p_f, p_{\mathrm{st}}, CN, L)$-stable. Following the above outline, we will derive a contradiction in the remainder of the proof.

\paragraph{Parameter Choice.} We set $0 < \eta< \beta = 1-\frac{9 \log^2|\kappa|}{\kappa^2}<1$ and assume $\beta-\eta > 1-\frac{10 \log^2|\kappa|}{\kappa^2}$. We can check that the m-OGP statement in \cref{thm:mogp} holds. As in Equations (61), (62) and (63) in \cite{gamarnik2022algorithms}, we set
\begin{align}\label{eq:parameter1}
    &f = CN, \quad C= \frac{\eta^2}{1600}, \quad Q = \frac{4800L\pi}{\eta^2}\sqrt{\alpha}, \quad T = \exp_2(2^{4mQ\log_2 Q}), \\ \label{eq:parameter2}
    &p_f = \frac{1}{9(Q+1)T}, \quad p_{\mathrm{st}} = \frac{1}{9Q(T+1)}, \quad \rho = \cos\left( \frac{\pi}{2Q} \right). 
\end{align}

\paragraph{Reduction to Deterministic Algorithms.}
\begin{lem}[Lemma 6.11 in \cite{gamarnik2022algorithms}]\label{lem:alg_determinimistic}
    Let $\kappa \in \mathbb{R}, \alpha<\asat(\kappa)$, and $M/N \to \alpha$. Suppose that $\mathcal{A}: \mathbb{R}^{M \times N} \times \Omega \rightarrow \{\pm 1\}^N$ is a randomized algorithm that is $\left(\rho, p_f, p_{\mathrm{st}}, f, L\right)$-stable for the Ising perceptron model. Then, there exists a deterministic algorithm $\mathcal{A}^*: \mathbb{R}^{M \times N} \rightarrow \{\pm 1\}^N$ that is $\left(\rho, 3 p_f, 3 p_{\mathrm{st}}, f, L\right)-$ stable.
\end{lem}
For the remainder of this discussion, we focus on the deterministic algorithm $\mathcal{A}^*: \mathbb{R}^{M \times N} \rightarrow \{\pm 1\}^N$ as described in \cref{lem:alg_determinimistic}, which is $\left(\rho, 3 p_f, 3 p_{\mathrm{st}}, f, L\right)-$ stable.

\paragraph{Chaos Event.} The chaos event describes the situation when $m$-tuples $\theta^{(i)}$ are solutions to $m$ i.i.d. random matrices $X_i$. The lemma below shows that it does not happen with high probability.
\begin{lem}\label{lem:chaos}
    There exist constants $C_1>0$ and $\kappa_1$ such that for any $\kappa<\kappa_1$ and $\alpha > \frac{C_1\log^2|\kappa|}{\kappa^2 \Phi(\kappa)}$, there exists $m \in \mathbb{N}$ such that the following holds. We have
    \begin{align*}
        \mathbb{P}\left [\cS_\kappa\left (1-\frac{10\log^2 |\kappa|}{\kappa^2}, \frac{\log^2 |\kappa|}{\kappa^2}, m, \alpha, \{\pi/2\}\right ) \neq \varnothing\right ] = \exp(-\Theta(N)).
    \end{align*}
\end{lem}
\begin{proof}
    The proof is the same as that of \cref{thm:mogp} so we omit it here.
\end{proof}
In particular, this lemma implies that for any $m$-tuples $\theta^{(i)} \in \{\pm 1 \}^N$ where $1\leq i \leq m$ with $X_i \theta^{(i)} \geq \kappa \sqrt{N}$, there exists $1\leq i < j \leq m$ such that $N^{-1} \inp{\theta^{(i)},\theta^{(j)}} \leq 1-\frac{10\log^2 |\kappa|}{\kappa^2}$.

\paragraph{Interpolation Paths are Smooth.} As above, we define $0 = \tau_0 < \tau_1< \cdots < \tau_Q = \pi/2$ as the equal partition of the interval $[0,\pi/2]$ and define $X_i(\tau) = \cos(\tau) X_0 + \sin(\tau) X_i$ for any $1\leq i \leq T$ and $\tau \in [0, \pi/2]$, where $X_0, \cdots, X_T$ are i.i.d. random matrices with i.i.d. $\normal(0,1)$ entries. Define $\theta_i(\tau_k) = \cA^*(X_i(\tau_k))$.
\begin{prop}[Proposition 6.13 in \cite{gamarnik2022algorithms}]\label{prop:trajectory_smooth}
We define 
\begin{align*}
    \cE_\mathrm{St} = \bigcap_{1\leq i < j \leq T} \bigcap_{0\leq k \leq Q-1} \left\{ \left|\inp{\theta_i(\tau_k), \theta_j(\tau_k)}-\inp{\theta_i(\tau_{k+1}), \theta_j(\tau_{k+1})} \right| \leq 4\sqrt{C}+4 \sqrt{3L\pi Q^{-1}} \alpha^{1/4}\right \},
\end{align*}
for $C$ defined in \cref{eq:parameter1}. Then
\begin{align*}
    \mathbb{P}[\cE_\mathrm{St}] \geq 1-3(T+1)Qp_{\mathrm{st}} -(T+1) \exp(-\Theta(N^2)).
\end{align*}
\end{prop}

\paragraph{The Algorithm is Successful Along Each Trajectory.}
\begin{lem}[Lemma 6.14 in \cite{gamarnik2022algorithms}]\label{lem:alg_success}
    We define
    \begin{align*}
        \cE_{\mathrm{Suc}} = \bigcap_{1\leq i \leq T} \bigcap_{0\leq k \leq Q}\left\{ X_i(\tau_k) \theta_i(\tau_k) \geq \kappa \sqrt{N} \right\}.
    \end{align*}
    Then we have 
    \begin{align*}
        \mathbb{P}[\cE_{\mathrm{Suc}}] \geq 1-3T(Q+1)p_f.   \end{align*}
\end{lem}

\paragraph{Combine the Arguments.}
For any subset $A \subset [T]$ with $|A| = m$, we define
\begin{align*}
    \cE_A = \left\{ \exists \theta^{(i)} \in \{\pm 1\}^N , i \in A: X_i(\pi / 2)\theta^{(i)}\geq \kappa \sqrt{N}, \beta-\eta \leq N^{-1} \inp{\theta^{(i)}, \theta^{(j)}} \leq \beta, i,j \in A, i \neq j\right\}.
\end{align*}
In other words, $\cE_A$ is exactly the chaos event in \cref{lem:chaos} restricted to $A \subset [T]$. So \cref{lem:chaos} implies that $\mathbb{P}[\cE_A]\geq \exp(-\Theta(N))$. Taking a union bound over $A \subset [T]$, we have that 
\begin{align*}
    \mathbb{P}[\cE_{\mathrm{Ch}}]:= \mathbb{P} \left[ \bigcap_{A \subset [T], |A| = m}  \cE_A^c \right ] \geq 1- \binom{T}{m} \exp(-\Theta(N)) = 1- \exp(-\Theta(N)).
\end{align*}
Following the same argument as in \cite{gamarnik2022algorithms}, we further define
\begin{align*}
    \cF = \cE_{\mathrm{St}} \cap \cE_{\mathrm{Suc}} \cap \cE_{\mathrm{Ch}}.
\end{align*}
Then
\begin{align*}
    \mathbb{P}[\mathcal{F}] & \geq 1-\mathbb{P}\left[\mathcal{E}_{\mathrm{St}}^c\right]-\mathbb{P}\left[\mathcal{E}_{\mathrm{Suc}}^c\right]-\mathbb{P}\left[\mathcal{E}_{\mathrm{Ch}}^c\right] \\ & \geq 1-3(T+1) Q p_{\mathrm{st}}-(T+1) \exp(-\Theta\left(n^2\right))-3 T(Q+1) p_f-\exp(-\Theta(N)) \\ & \geq \frac{1}{3}-\exp (-\Theta(N)).
\end{align*}
Condition on $\cF$ and with the choice of $C$ and $Q$, we have that
\begin{align*}
    |N^{-1}\inp{\theta_{i}(\tau_k),\theta_{j}(\tau_k)}-N^{-1}\inp{\theta_{i}(\tau_{k+1}),\theta_{j}(\tau_{k+1})}| \leq \frac{\eta}{5},
\end{align*}
for any $1 \leq i < j \leq T$ and $0 \leq k \leq Q-1$.
\begin{prop}[Proposition 6.15 in \cite{gamarnik2022algorithms}]\label{prop:combine}
    For every $A \subset [T]$ with $|A| = m$, there exists $1\leq i_A < j_A \leq m$ and $\tau_A\in\{ \tau_1, \cdots, \tau_Q\}$ such that for $\delta = \eta/100$, we have
    \begin{align*}
        N^{-1}\inp{\theta_{i_A}(\tau_A),\theta_{j_A}(\tau_A)} \in (\beta-\eta+3\delta, \beta-3 \delta) \subsetneq (\beta-\eta, \beta).
    \end{align*}
\end{prop}

\paragraph{Apply Ramsey Theory.}
We firstly cite two useful theorems from the two-colored and multi-colored Ramsey theory.

\begin{thm}[Theorem 6.9 in \cite{gamarnik2022algorithms}]\label{thm:ramsey1}
    Let $k, \ell \geq 2$ be integers; and $R(k, \ell)$ denotes the smallest $n \in \mathbb{N}$ such that any red/blue (edge) coloring of $K_n$ necessarily contains either a red $K_k$ or a blue $K_{\ell}$. Then,
\[
R(k, \ell) \leq\binom{ k+\ell-2}{k-1}=\binom{k+\ell-2}{\ell-1} .
\]
In the special case where $k=\ell=M \in \mathbb{N}$, we thus have
\[
R(M, M) \leq\binom{ 2 M-2}{M-1}.
\]
\end{thm}

\begin{thm}[Theorem 6.10 in \cite{gamarnik2022algorithms}, \cite{erdos1935combinatorial}]\label{thm:ramsey2}
    Let $q, m \in \mathbb{N}$. Denote by $R_q(m)$ the smallest $n \in \mathbb{N}$ such that any $q$-coloring of the edges of $K_n$ necessarily contains a monochromatic $K_m$. Then,
\[
R_q(m) \leq q^{q m}.
\]
\end{thm}

We construct a graph $G = ([T], E)$ as follows. For any $1\leq i < j \leq T$, we  put an edge $(i,j) \in E$ if and only if there exists a time $\tau \in [0,\pi/2]$ such that $N^{-1}\inp{\theta_{i}(\tau),\theta_{j}(\tau)} \in (\beta-\eta, \beta)$. We next color each $(i,j) \in E$ with one of $Q$ colors, i.e., we color the edge $(i,j)$ with the color $t$ if $\tau_t$ is the first time such that $N^{-1}\inp{\theta_{i}(\tau_t),\theta_{j}(\tau_t)} \in (\beta-\eta, \beta)$. 

\begin{prop}[Proposition 6.16 in \cite{gamarnik2022algorithms}]\label{prop:monochromatic}
    $G = ([T], E)$ defined above contains a monochromatic $m$-clique $K_m$.
\end{prop}
\begin{proof}
    The proof follows by first extracting a large clique based on \cref{thm:ramsey1} and further extracting a large monochromatic clique based on \cref{thm:ramsey2}.
\end{proof}

\paragraph{Complete the Proof of \cref{thm:hardness}.}
From \cref{prop:monochromatic}, we know that there exists an $m$-tuple, $1 \leq i_1<i_2<\cdots<i_m \leq T$ and a color $t \in\{1,2, \ldots, Q\}$ such that
\begin{align*}
N^{-1}\inp{\theta_{i_k}(\tau_t),\theta_{i_\ell}(\tau_t)} \in (\beta-\eta, \beta), \quad 1 \leq k<\ell \leq m.
\end{align*}
Now, set $\theta^{(k)} = \mathcal{A}^*\left(X_{i_k}\left(\tau_t\right)\right) \in \{\pm 1\}^N, 1 \leq k \leq m$. Note that we have
\begin{align*}
    X_{i_k}\left(\tau_t\right) \theta^{(k)} \geq \kappa \sqrt{N}, \quad 1 \leq k \leq m.
\end{align*}
And for $1 \leq k<\ell \leq m$,
\begin{align*}
    \beta-\eta<\frac{1}{N}\left\langle\theta^{(k)}, \theta^{(\ell)}\right\rangle<\beta.
\end{align*}
In particular, for the choice $\zeta=\left\{i_1, i_2, \ldots, i_m\right\}$ of the $m$-tuple of distinct indices, the set $\mathcal{S}_\zeta = \mathcal{S}_\kappa(\beta, \eta, m, \alpha, \mathcal{I})$ with $\mathcal{I}=\left\{\tau_i: 0 \leq i \leq Q\right\}$ is non-empty. This implies that
\begin{align*}
    \mathbb{P}\left[\exists \zeta \in[T]:|\zeta|=m, \mathcal{S}_\zeta \neq \varnothing\right] \geq \mathbb{P}[\mathcal{F}] \geq \frac{1}{3}-\exp (-\Theta(N)).
\end{align*}
On the other hand, using \cref{thm:mogp}, we have
\begin{align*}
    \mathbb{P}\left[\exists \zeta \in[T]:|\zeta|=m, \mathcal{S}_\zeta \neq \varnothing\right] \leq\binom{ T}{m} e^{-\Theta(N)}=\exp (-\Theta(N)),
\end{align*}
by taking a union bound and note that $\binom{T}{m}=O(1)$. Combining these, we therefore obtain
\begin{align*}
\exp (-\Theta(N)) \geq \frac{1}{3}-\exp (-\Theta(N))
\end{align*}
which is a contradiction for all $N$ large enough, completing the proof.
\end{proof}

\ifnum\anon=0
\section*{Acknowledgments}
We would like to thank Dylan Altschuler, Brice Huang, Mark Sellke, and Nike Sun for helpful conversations.

\fi

\newpage

\appendix
\section{An Improved Partial Coloring Lemma for Margin Zero}\label{app:improved}
In this section, we present an improved partial coloring lemma for the case $\kappa = 0$, which further refines \cref{thm:mainpartialcolor_zero} and allows us to push the algorithmic lower bound to $0.1$.
To state this improved result, we must first define two critical functions that govern the asymptotic behavior of the Lovett-Meka Edge-Walk algorithm:
\begin{defn}\label{defn:improved_Lovett_Meka}
	For scalars $\alpha > 0$, $r_0 \in [0, 1]$, and random variable $c \ge 0$, define the functions $u, v: \R_{\ge 0} \to \R_{\ge 0}$ as follows:
	\begin{equation}
		v(t) = 2 \alpha \E \left[ \Phi \left( - \frac{c}{\sqrt{t}} \right) \right], \ t \ge 0,
	\end{equation}
	and
	\begin{equation}
		u(0) = r_0, \ u'(t) = \max \left( 1 - u(t) - v(t), 0 \right), \ t \ge 0.
	\end{equation}
	Further, we let
	\begin{align}
		T_1 = \, & T_1 (\alpha, c, r_0) = \inf \left\{ t \ge 0: u(t) = 1 - v(t) \right\}, \\
		T_2 = \, & T_2 (\alpha, c, r_0) = \sup \left\{ t \le T_1: u(t) + \int_{t}^{T_1} \left( 1 - u(t) - v(s) \right) \d s \ge 1 \right\}.
	\end{align}
	We assume 
	\begin{equation*}
		u(0) + \int_{0}^{T_1} \left( 1 - u(0) - v(s) \right) \d s = r_0 + \int_{0}^{T_1} \left( 1 - r_0 - v(s) \right) \d s \ge 1,
	\end{equation*}
	so that $T_2$ is well-defined.
	Finally, we define
	\begin{align}\label{eq:defp0}
            p_0 = p_0 (\alpha, c, r_0) = \, & \frac{u(T_1) - u(T_2)}{1 - u(T_2)}, \\ \label{eq:defp1} 
		p_1 = p_1 (\alpha, c, r_0) = \, & u(T_2).
	\end{align}
\end{defn}

\begin{thm}\label{thm:improved_Lovett_Meka}
Let $X_1,\ldots,X_M \in \R^N$ be vectors where $M/N \to \alpha > 0$, and $\theta_0 \in [-1,1]^N$ be the starting point of the Lovett-Meka algorithm, satisfying $\norm{\theta_0}_2^2 \ge r_0 N$ for some $r_0 \in [0, 1]$.
Let $c_1,\ldots,c_M \ge 0$ be threshold parameters such that 
\begin{equation}
	\frac{1}{M} \sum_{i=1}^{M} \delta_{c_i} \stackrel{w}{\Rightarrow} \operatorname{Law} (c), \quad M \to \infty,
\end{equation}
where $c \ge 0$ is a random variable.
Let $T_1 = T_1 (\alpha, c, r_0)$, $p_0 = p_0 (\alpha, c, r_0)$, and $p_1 = p_1 (\alpha, c, r_0)$ be as defined in \cref{defn:improved_Lovett_Meka}.
Let $\delta>0$ be a small approximation parameter and $\gamma > 0$ be the step size, which satisfy the same requirements as in \cref{thm:mainpartialcolor}.
Then, for any $\veps > 0$, there exists some $\eta > 0$ (both $\veps$ and $\eta$ do not depend on $M, N$), such that with probability at least $\eta$, running the Lovett-Meka Edge-Walk algorithm for $T = T_1/\gamma^2$ steps finds a point $\theta \in [-1,1]^N$ such that
\begin{enumerate}
\item[(i)] The margin constraints are satisfied: $\<\theta-\theta_0, X_j\> \ge -c_j \|X_j\|_2$ for all $j \in [M]$.
\item[(ii)] The number of nearly tight variable constraints is at least $(p_1 - \veps) N$: $|\theta_i| \ge 1-\delta$ for at least $(p_1 - \veps) N$ indices $i \in [N]$.
\item[(iii)] The normalized length of the initialization for next round is at least $p_0$:
	\begin{equation}
		\frac{\norm{\theta}_2^2 - (p_1 - \veps) N}{N - (p_1 - \veps) N} \ge p_0.
	\end{equation}
\end{enumerate}
Moreover, the algorithm runs in time $O((M+N)^3 \cdot \delta^{-2} \cdot \log(NM/\delta))$.
\end{thm}
\begin{proof}
    Similar to the proof of \cref{thm:mainpartialcolor_zero}, we set $\theta = \theta_T$, the output of the Edge-Walk algorithm. Recalling the definition of $T_1$, $p_0$ and $p_1$ from \cref{defn:improved_Lovett_Meka}, it suffices to show that for any $\veps > 0$:
    \begin{enumerate}
        \item [(a)] $\E \left[ \left\vert \cC^{\mathrm{var}}_T \right\vert \right] \ge (p_1 - \veps ) N = (u(T_2) - \veps) N$.
        \item [(b)] $\norm{\theta_T}_2^2 \ge (p_1 + p_0 - p_1 p_0 - \veps) N = (u(T_1) - \veps) N$.
    \end{enumerate}
    To this end, we follow the proof strategy of \cite[Claim 16]{lovett2015constructive}. Note that for all $t \in [T]$:
    \begin{align}
        \E \left[ \norm{\theta_t}_2^2 \right] = \, & \E \left[ \norm{\theta_{t-1}}_2^2 \right] + \gamma^2 \E \left[ \dim(\cV_t) \right] \\
        = \, & \E \left[ \norm{\theta_{t-1}}_2^2 \right] + \gamma^2 \left( N - \E \left[ \left\vert \cC^{\mathrm{var}}_t \right\vert \right] - \E \left[ \left\vert \cC^{\mathrm{disc}}_t \right\vert \right] \right).
    \end{align}
    Using a similar argument as that in the proof of \cref{thm:mainpartialcolor_zero}, we know that
    \begin{equation}
        \E \left[ \left\vert \cC^{\mathrm{disc}}_t \right\vert \right] \le \, 2 \sum_{j=1}^{M} \Phi \left( - \frac{c_j - \delta}{\gamma \sqrt{t}} \right) = 2 \sum_{j=1}^{M} \Phi \left( -\frac{c_j - \delta}{ \sqrt{t T_1 / T}} \right),
    \end{equation}
    which further implies
    \begin{equation}
        \E \left[ \norm{\theta_t}_2^2 \right] \ge \, \E \left[ \norm{\theta_{t-1}}_2^2 \right] + \gamma^2 \left( N - \E \left[ \left\vert \cC^{\mathrm{var}}_t \right\vert \right] - 2 \sum_{j=1}^{M} \Phi \left( - \frac{c_j - \delta}{ \sqrt{t T_1 / T}} \right) \right).
    \end{equation}
    For any $x \in [0, T_1]$, we define
    \begin{equation}
        f_N (x) = \frac{1}{N} \E \left[ \norm{\theta_{[xT/T_1]}}_2^2 \right], \quad g_N (x) = \frac{1}{N} \E \left[ \left\vert \cC^{\mathrm{var}}_{[xT/T_1]} \right\vert \right].
    \end{equation}
    Then, we know that $f_N (x) \ge g_N (x)$ always, and
    \begin{equation}\label{eq:difference_ineq}
        f_N \left( \frac{tT_1}{T} \right) \ge f_N \left( \frac{(t-1)T_1}{T} \right) + \frac{T_1}{T} \left( 1 - g_N \left( \frac{tT_1}{T} \right) - \frac{2}{N} \sum_{j=1}^{M} \Phi \left( -\frac{c_j - \delta}{ \sqrt{t T_1 / T}} \right) \right)
    \end{equation}
    for all $t \in [T]$. Then, in order to prove claims (a)-(b), it suffices to show that $\forall \veps > 0$, $g_N (T_1) \ge u(T_2) - \veps$ and $f_N (T_1) \ge u(T_1) - \veps$ for large enough $M, N$.

    We will prove this statement by contradiction. Assume that there exists a subsequence $\{ N_k \}$, such that for each $k$, we have either $g_{N_k} (T_1) < u(T_2) - \veps$ or $f_{N_k} (T_1) < u(T_1) - \veps$. Since each $f_{N_k}$ and $g_{N_k}$ is non-decreasing and takes values in $[0, 1]$, by Helly's selection theorem, we know that there exists a further subsequence (still denoted as $\{ N_k \}$), along which $f_{N_k} \to f$, $g_{N_k} \to g$. Consequently, we have $f \ge g$, and taking the limit of \cref{eq:difference_ineq} (using Bounded Convergence Theorem) yields
    \begin{equation}\label{eq:integral_ineq}
        f(x) \ge f(y) + \int_{y}^{x} \left( 1 - g(s) - v(s) \right) \d s, \quad \forall 0 \le y \le x \le T_1,
    \end{equation}
    where we recall $v(s) = 2 \alpha \E [\Phi(-c / \sqrt{s})]$ from \cref{defn:improved_Lovett_Meka}. Further, we know that either $f(T_1) < u(T_1)$ or $g(T_1) < u(T_2)$. In what follows, we will show that $f(T_1) \ge u(T_1)$ and $g(T_1) \ge u(T_2)$, thus yielding a contradiction and consequently proving our claims (a) and (b).
    
    We first prove $f(T_1) \ge u(T_1)$. In fact, we will show that $f(x) \ge u(x)$ for all $x \in [0, T_1]$. By our assumption, we have $f(0) \ge r_0 = u(0)$. Further, since $f \ge g$, $f$ is non-decreasing, it follows that
    \begin{align}
        f(x) \ge \, & f(0) + \int_{0}^{x} \left( 1 - g(s) - v(s) \right) \d s \\
        \ge \, & f(0) + \int_{0}^{x} \left( 1 - f(s) - v(s) \right) \d s \\
        \ge \, & f(0) + \int_{0}^{x} \left( 1 - f(s) - v(s) \right)_+ \d s, \quad \forall x \in [0, T_1].
    \end{align}
    Applying the ODE comparison theorem yields that $f(x) \ge u (x)$ for all $x \in [0, T_1]$. Next, we show that $g(T_1) \ge u(T_2)$. Since $g$ is also non-decreasing, we have $g(x) \le g(T_1)$ for all $x \in [0, T_1]$. Therefore, \cref{eq:integral_ineq} implies
    \begin{align}
        1 \ge f(T_1) \ge \, & f(T_2) + \int_{T_2}^{T_1} \left( 1 - g(s) - v(s) \right) \d s \\
        \ge \, & f(T_2) + \int_{T_2}^{T_1} \left( 1 - g(T_1) - v(s) \right) \d s \\
        \ge \, & u(T_2) + \int_{T_2}^{T_1} \left( 1 - g(T_1) - v(s) \right) \d s.
    \end{align}
    By definition of $T_2$, we have
    \begin{equation}
        1 \le \, u(T_2) + \int_{T_2}^{T_1} \left( 1 - u(T_2) - v(s) \right) \d s.
    \end{equation}
    Comparing the above two inequalities immediately implies that $g(T_1) \ge u(T_2)$. This completes the proof of \cref{thm:improved_Lovett_Meka}.
\end{proof}

\paragraph{Computing the Effective $\alpha_0$ and $r_0$ for the Beginning of the Second Stage.} We denote by $\kappa_0 > 0$ the initial margin for the linear programming algorithm in the first stage, and $\alpha_0$ the effective aspect ratio for the second stage. Then, $\alpha_0$ and $r_0$ can be computed as follows:
\begin{enumerate}
	\item Norm of the output vector $\rho$ is the unique solution to
		\begin{equation}\label{eq:solve_rho}
			\alpha \Phi \left( \frac{\kappa_0}{\rho} \right) = \frac{\varphi(\rho) \varphi'(\rho)}{\rho}.	
		\end{equation}
	\item Proportion of tight variable constraints $r = 2 \Phi(-t)$, where $t$ solves the equation
		\begin{equation}\label{eq:solve_t}
			\rho^2 = \E \left[ \min \left( \frac{\vert G \vert}{t}, 1 \right)^2 \right].
		\end{equation}
	\item Effective aspect ratio for the next stage of iterative Edge-Walk: $\alpha_0 = \alpha / (1 - r)$.
	\item Squared norm of initialization of the second stage: $r_0 = (\rho^2 - r)/(1 - r)$.
\end{enumerate}
\paragraph{Choice of Parameters for $\alpha = 0.1$.} 
Using a computer, we have identified the following parameters which allow us to achieve $\alpha = 0.1$.
\begin{align*}
    &\kappa_0 = 2.31, \; \alpha_0 \approx 1.008960 > 1, \; r_0 \approx 0.332645 > 0.33,\\
    & c^{(1)} = 2.00, \; c^{(2)} = 2.60, \; c^{(3)} = 2.60, \; c^{(4)} = 3.10, \; c^{(5)} = 3.20, \; c^{(6)} = 3.50, \; c^{(7)} = 3.70, \; \\
    &c^{(8)} = 3.90, \; c^{(9)} = 4.20, \; c^{(10)} = 4.80, \; c^{(11)} = 5.00, \; c^{(12)} = 5.00, \; c^{(13)} = 6.00, \; c^{(14)} = 6.00, \; \\
    &c^{(15)} = 7.00, \; c^{(16)} = 7.00, \; c^{(17)} = 8.00, \; c^{(18)} = 8.00, \; c^{(19)} = 9.00, \; c^{(20)} = 9.00,\\
    &p_1^{(1)} = 0.34, \; p_1^{(2)} = 0.31, \; p_1^{(3)} = 0.41, \; p_1^{(4)} = 0.35, \; p_1^{(5)} = 0.40, \; p_1^{(6)} = 0.39, \; p_1^{(7)} = 0.40, \; \\
    &p_1^{(8)} = 0.40, \; p_1^{(9)} = 0.39, \; p_1^{(10)} = 0.32, \; p_1^{(11)} = 0.34, \; p_1^{(12)} = 0.36, \; p_1^{(13)} = 0.27, \; p_1^{(14)} = 0.30, \; \\
    &p_1^{(15)} = 0.22, \; p_1^{(16)} = 0.25, \; p_1^{(17)} = 0.19, \; p_1^{(18)} = 0.22, \; p_1^{(19)} = 0.17, \; p_1^{(20)} = 0.19,\\
    &p_0^{(1)} = 0.50, \; p_0^{(2)} = 0.55, \; p_0^{(3)} = 0.58, \; p_0^{(4)} = 0.58, \; p_0^{(5)} = 0.59, \; p_0^{(6)} = 0.60, \; p_0^{(7)} = 0.60, \; \\
    &p_0^{(8)} = 0.60, \; p_0^{(9)} = 0.61, \; p_0^{(10)} = 0.60, \; p_0^{(11)} = 0.61, \; p_0^{(12)} = 0.61, \; p_0^{(13)} = 0.61, \; p_0^{(14)} = 0.61, \; \\
    &p_0^{(15)} = 0.61, \; p_0^{(16)} = 0.61, \; p_0^{(17)} = 0.61, \; p_0^{(18)} = 0.61, \; p_0^{(19)} = 0.61, \; p_0^{(20)} = 0.61.
\end{align*}

\begin{thm}\label{thm:fullcoloring02}
    For $\kappa = 0$, our algorithm finds a solution when $\alpha \leq 0.1$ with high probability.
\end{thm}
\begin{proof}[Proof of \cref{thm:fullcoloring02}]
    By monotonicity, we restrict our attention to the case when $\alpha = 0.1$. The proof is rather similar to that of \cref{thm:fullcoloring0}, so we omit some unnecessary details. We choose our parameters as above. We check that $0<0.1<\frac{2}{\pi} \mathbb{E}\left[(\kappa_0-G)_{+}^2\right]^{-1}$. Therefore, by \cref{lem:feasibility_LP}, \cref{eq:abs_bound_alpha} is fulfilled. Next, we use \cref{def:LP_auxiliary} and \cref{lem:property_maxmin} to solve for $t$ and $\rho$ and compute that the right hand side of \cref{eq:tight_fraction} satisfies $2\Phi(-t(\alpha,\kappa_0)) > 0.9$. Define $R^{(0)}_{\text{tight}} = 0.9$. 

    We next apply the partial coloring lemma \cref{thm:improved_Lovett_Meka} for $k\leq 20$. More specifically, for round $k$, we choose 
    \begin{align*}
        c_j^{(k)} = c^{(k)} \frac{\sqrt{|I_{k-1}|}}{ \| (X_j)_{I_{k-1}} \|_2},
    \end{align*}
    and we choose $p_1^{(k)}$ and $p_0^{(k)}$ as above. For each round $k \leq 20$, by \cref{lem:normbound}, we have
    \begin{align*}
        \frac{1}{M} \sum_{i=1}^{M} \delta_{c_i^{(k)}} \stackrel{w}{\Rightarrow} \delta_{c^{(k)}}, \quad M \to \infty.
    \end{align*}
    Following the definitions in \cref{eq:defp0} and \cref{eq:defp1}, we also check that 
    \begin{align*}
        p_1^{(k)}>p_1 \left (\frac{1}{\prod_{i=1}^{k-1} (1-p_1^{(i)}) }  , c^{(k)}, p_0^{(k-1)} \right), \quad p_0^{(k)}<p_0 \left (\frac{1}{\prod_{i=1}^{k-1} (1-p_1^{(i)}) }  , c^{(k)}, p_0^{(k-1)} \right),
    \end{align*}
    where we identify $p_0^{(0)} = r_0$. We check inductively that this would imply that the conditions of \cref{thm:improved_Lovett_Meka} holds. Indeed, $\alpha = 0.1$, and at each round $k$, the number of not nearly tight variable constraints is bounded by $(1-R^{(0)}_{\text{tight}}) \prod_{i=1}^{k-1} (1-p_1^{(i)}) N = 0.1 \prod_{i=1}^{k-1} (1-p_1^{(i)}) N$. Therefore, we can apply \cref{thm:fullcoloring02} for $k\leq 20$. Now for $k >20$, we apply the original partial coloring lemma \cref{thm:mainpartialcolor_zero}, and follow the same proof as \cref{thm:fullcoloring0} (we use the same choice of $c^{(k)}$ and check the conditions hold). All together,
    \begin{align*}
        & \frac{1}{\sqrt{N}}|\<\hat \theta^{(K)}-\hat \theta^{(0)},X_j\>|
         \leq \frac{1}{\sqrt{N}}\sum_{k = 1}^\infty c_j^{(k)} \|(X_j)_{I_{k-1}}\|_2 \\
        &< \sum_{k = 1}^{20} c^{(k)} \sqrt{(1-R^{(0)}_{\text{tight}}) \prod_{i=1}^{k-1} (1-p_1^{(i)})} + \sum_{k = 21}^\infty c^{(k)} \sqrt{(1-R^{(0)}_{\text{tight}}) (1-1/K_3)^{k-21} \prod_{i=1}^{20} (1-p_1^{(i)})}\\
        &<1.96+0.01  <2.
    \end{align*}
    This means that before the Edge Walk steps, all rows have a margin $\frac{1}{\sqrt{N}}\<\hat \theta^{(0)}, X_j\> \geq \kappa_0 = 2.31$, and after iterations of the Edge Walk algorithms, each inner product decreases by at most $2$. Therefore, every inner product remains above $0.31$ after the iterations. Finally, we study the randomized sign rounding step. As $\delta = o(1/\log N)$, by \cref{eq:randomsignrounding} we have that the change in the inner product $\<\chi,X_j\> - \<\hat{\theta}^{(K)},X_j\> \geq -o(\sqrt{N})$. Thus, we have that $\<\chi,X_j\> > 0$ for all $j \in [M]$, which finishes the proof.
\end{proof}

\newpage

\bibliographystyle{alpha}
\bibliography{ising_perc}

\end{document}